\documentclass[journal, lettersize]{IEEEtran}
\usepackage{amsmath,amsfonts}
\usepackage{amsthm}
\usepackage{amssymb}
\usepackage{array}
\usepackage[caption=false,font=normalsize,labelfont=sf,textfont=sf]{subfig}
\usepackage{textcomp}
\usepackage{stfloats}
\usepackage{url}
\usepackage{verbatim}
\usepackage{graphicx}
\usepackage{cite}
\usepackage{upgreek}
\usepackage{multirow}
\usepackage{multicol}
\usepackage{numprint}
\usepackage{dsfont}
\usepackage{tikz}
\usepackage{xfrac}

\newtheorem{theorem}{Theorem}

\newtheorem{proposition}{Proposition}
\newtheorem{corollary}{Corollary}
\usepackage{textcomp}
\usepackage{stfloats}
\usepackage{graphicx}

\usepackage{listings,xfp}
\usepackage{anyfontsize}

\npdecimalsign{\ensuremath{.}}

\makeatletter
{\small 
\xdef\f@size@small{\f@size}
\xdef\f@baselineskip@small{\f@baselineskip}
\normalsize 
\xdef\f@size@normalsize{\f@size}
\xdef\f@baselineskip@normalsize{\f@baselineskip}
}
\newcommand{\smalltonormalsize}{%
  \fontsize
    {\fpeval{(\f@size@small+\f@size@normalsize)/2}}
    {\fpeval{(\f@baselineskip@small+\f@baselineskip@normalsize)/2}}%
  \selectfont
}
\makeatother

\makeatletter
{\footnotesize 
\xdef\f@size@footnotesize{\f@size}
\xdef\f@baselineskip@footnotesize{\f@baselineskip}
\small 
\xdef\f@size@small{\f@size}
\xdef\f@baselineskip@small{\f@baselineskip}
}
\newcommand{\footnotesizetosmall}{%
  \fontsize
    {\fpeval{(\f@size@footnotesize+\f@size@small)/2}}
    {\fpeval{(\f@baselineskip@footnotesize+\f@baselineskip@small)/2}}%
  \selectfont
}
\makeatother

\makeatletter
{\scriptsize 
\xdef\f@size@scriptsize{\f@size}
\xdef\f@baselineskip@scriptsize{\f@baselineskip}
\footnotesize 
\xdef\f@size@footnotesize{\f@size}
\xdef\f@baselineskip@footnotesize{\f@baselineskip}
}
\newcommand{\scriptsizetofootnotesize}{%
  \fontsize
    {\fpeval{(\f@size@scriptsize+\f@size@footnotesize)/2}}
    {\fpeval{(\f@baselineskip@scriptsize+\f@baselineskip@footnotesize)/2}}%
  \selectfont
}
\makeatother

\begin{document}

\title{Joint Metrics for EMF Exposure and Coverage in Real-World Homogeneous and Inhomogeneous Cellular Networks}

\author{Quentin~Gontier,~\IEEEmembership{Graduate Student Member,~IEEE,} Charles~Wiame,~\IEEEmembership{Member,~IEEE,} Shanshan~Wang,~\IEEEmembership{Member,~IEEE,} Marco~Di~Renzo,~\IEEEmembership{Fellow,~IEEE,} Joe~Wiart,~\IEEEmembership{Senior Member,~IEEE,} François~Horlin,~\IEEEmembership{Member,~IEEE,} Christo~Tsigros, Claude~Oestges,~\IEEEmembership{Fellow,~IEEE,} and  Philippe~De~Doncker,~\IEEEmembership{Member,~IEEE}
\thanks{Manuscript received 17 August 2023; revised 14 February 2024 and 2 May 2024; accepted 3 May 2024. The work of Q. Gontier is supported by Innoviris under the Stochastic Geometry Modeling of Public Exposure to EMF (STOEMP-EMF)~grant. The work of M. Di Renzo is supported in part by the European Commission through the Horizon Europe project titled COVER under grant agreement number 101086228, the Horizon Europe project titled UNITE under grant agreement number 101129618, and the Horizon Europe project titled INSTINCT under grant agreement number 101139161, as well as by the Agence Nationale de la Recherche (ANR) through the France 2030 project titled ANR-PEPR Networks of the Future under grant agreement NF-Founds 22-PEFT-0010, and by the CHIST-ERA project titled PASSIONATE under grant agreement CHRIST-ERA-22-WAI-04 through ANR-23-CHR4-0003-01.}
\thanks{Q. Gontier, F. Horlin and Ph. De Doncker are with Universit\'e Libre de Bruxelles, OPERA-WCG, Avenue Roosevelt 50 CP 165/81, 1050 Brussels, Belgium (quentin.gontier@ulb.be).}
\thanks{C. Wiame is with NCRC Group, Massachusetts Institute of Technology, Cambridge, MA 02139 USA.}
\thanks{S. Wang and J. Wiart are with Chaire C2M, LTCI, Télécom Paris, Institut Polytechnique de Paris, 91120 Palaiseau, France.}
\thanks{M. Di Renzo is with Universit\'e Paris-Saclay, CNRS, CentraleSup\'elec, Laboratoire des Signaux et Syst\`emes, 3 Rue Joliot-Curie, 91192 Gif-sur-Yvette, France. (marco.di-renzo@universite-paris-saclay.fr).}
\thanks{C. Tsigros is with Department Technologies et Rayonnement, Brussels Environment, Belgium.}
\thanks{C. Oestges is with ICTEAM Institute, Université Catholique de Louvain, Louvain-la-Neuve, Belgium.}
}


\maketitle
\begin{abstract}
This paper evaluates the downlink performance of cellular networks in terms of coverage and electromagnetic field exposure (EMFE), in the framework of stochastic geometry. The model is constructed based on datasets for sub-6~GHz macro cellular networks but it is general enough to be applicable to millimeter-wave networks as well. On the one hand, performance metrics are calculated for $\beta$-Ginibre point processes which are shown to faithfully model a large number of motion-invariant networks. On the other hand, performance metrics are derived for inhomogeneous Poisson point processes with a radial intensity measure, which are shown to be a good approximation for motion-variant networks. For both cases, joint and marginal distributions of the EMFE and the coverage, and the first moments of the EMFE are provided and validated by Monte Carlo simulations using realistic sets of parameters from two sub-6~GHz macro urban cellular networks, i.e., 5G~NR~2100 (Paris, France) and LTE~1800 (Brussels, Belgium) datasets. In addition, this paper includes the analysis of the impact of the network parameters and discusses the achievable trade-off between coverage and EMFE.
\end{abstract}


\section{Introduction}

\IEEEPARstart{T}{elecommunication} operators are faced with the challenge of optimizing the coverage of their cellular networks while ensuring compliance with public electromagnetic field exposure (EMFE) limits. On the one hand, the study of coverage, outage, data rate or spectral efficiency can be carried out through the evaluation of the signal-to-interference-and-noise ratio (SINR). On the other hand, the EMFE is subject to restrictions for public health reasons specified in terms of incident power density (IPD), or, equivalently, in terms of electric field strength\cite{GontierAccess, app10238753}. However, both metrics are still most often considered independently even though their correlation is high and their combination is necessary to fully address network optimization problems. Both also heavily depend on the randomness in the network topology. To capture this aspect, stochastic geometry (SG) theory can be efficiently employed as an alternative to numerical simulations. Using this framework, base stations (BSs) are modeled as spatial point processes (PPs)\cite{Baccelli1997StochasticGA, tutorial}, for which closed-form expressions characterizing the average network performance can be derived. Motivated by these considerations, \textbf{the main aim of this paper is to introduce a mathematical framework, applicable to real-world networks, for jointly evaluating the trade-offs between SINR and EMFE to electromagnetic fields for both motion-invariant (MI)}, i.e. both stationary (translation invariant) and isotropic (rotation invariant) networks, \textbf{and motion-variant (MV) networks}.

\subsection{Related Works}
\subsubsection{Evaluation of the SINR by using SG}\label{sssec:IPD_soa}
The SINR was the initial metric explored within the SG framework\cite{baccellimetrics}. Regular hexagonal lattices\cite{hexagonal2016} or perfect square lattices\cite{InterferenceLargeNetworkHaenggi} were the first models of cellular networks being considered. They were embedded in the SG framework because they are extreme cases for modeling the repulsion between the points of a PP. These models are often limited to simulations because of their lack of mathematical tractability. Another body of studies \cite{Baccelli1997StochasticGA, Lee2013, 7733098} employed the homogeneous Poisson point process (H-PPP), consisting of randomly located points in the region of interest, without any spatial dependence between them. Although the H-PPP leads to highly tractable analytical expressions, which explains its use in most of the SG literature applied to telecommunication networks, it cannot model spatial repulsive or attractive behaviors between points. Numerical results have shown that H-PPPs provide a lower bound for the coverage probability of real cellular networks while hexagonal lattices provide an upper bound\cite{6042301}. In another body of works, more advanced PP models are investigated (e.g., Strauss Hardcore PP\cite{GSPP}, Grey Saturation PP\cite{GSPP}, Poisson Hard-core process\cite{PHCP}). A complete taxonomy of PP models frequently used in the literature of cellular networks can be found in\cite{DiRenzoBible2021}. Some models lead to tractable expressions but have not been tested in real-world networks\cite{8187697, PHCP}. Other models have been validated in real-world networks but the analytical expressions are highly intractable mathematically: approximations are usually required to obtain mathematical expressions for useful performance metrics\cite{Kibilda2016, Anjin2013, Zhang_2021, 6524460, GSPP}. Finally, some studies provide models addressing tractability and accuracy aspects, but they are only applicable to MI networks\cite{IDT, Baccelli_DPP2}. This is for example the case of $\beta$-Ginibre point processes ($\beta$-GPPs) which allow for tractable and accurate modeling of many MI networks, as shown for the city of Paris, France, in \cite{BGPP_Gomez} and for thirteen networks in Western Europe in\cite{cost2022}. Previous analyses for $\beta$-GPP cellular networks are however limited to the SINR cumulative distribution function (CDF)\cite{Ginibre_theory, GinibreNaoto14}. Regarding MV networks, the literature is very scarce. Inhomogeneous Poisson point processes (I-PPPs) models have been widely used in the literature, but only to replace equivalent more complex MI PPs, from the point of view of the typical user for downlink (DL) analyses\cite{IDT} or from the point of view of the typical BS for uplink (UL) analyses\cite{uplink_meta, Bai16}. These works cannot capture the dependence of the network performance on the location of the user under investigation. The only attempt to derive the SINR complementary cumulative distribution function (CCDF) in an inhomogeneous network is in\cite{alpha-stable} where the authors consider a generalized PPP setup with an $\alpha$-stable distributed BS density (which can be seen as a generalization of the I-PPP), whose parameters are fitted to empirical data but at a cost of analytical tractability. 

\subsubsection{Evaluation of the IPD by using SG}
More recently, the IPD has progressively gained interest in the SG community to evaluate the performance of wireless power transfer systems\cite{HE2, HE3}, with the objective to estimate the amount of power that can be harvested in the context of the Internet-of-Things. In \cite{ECC}, the energy correlation is investigated in a wireless power transfer system where the transmitters use dynamic beamforming (BF). The evaluation of the IPD to assess the EMFE for public health concerns is more recent. Indeed, the IPD and EMFE are highly coupled since the objective of EMF-aware systems is to ensure that the incident power is low enough to fulfill EMFE thresholds. In the SG context, a first attempt to model the EMFE can be found in\cite{app10238753}. The authors use an empirical propagation model for a 5G massive multiple-input multiple-output (mMIMO) network in the millimeter wave (mmWave) band. In\cite{GontierAccess}, the theoretical distribution of the EMFE is compared to an experimental distribution obtained from measurements in an urban environment. This model is then used in\cite{GontierICC} to study BS densification scenarios. It is worth mentioning that numerous works employed deterministic models to evaluate the impact of network densification on the EMFE\cite{chiaraviglio2021cellular, Deruyck18}. The EMFE has also been numerically assessed in an indoor environment in \cite{Shikhantsov20}, by using a methodology akin to the SG framework. This involves obtaining a large number of ray-launching realizations by employing a randomized arrangement of scatterers for each realization. In the context of SG modeling, the EMFE is analyzed considering a max-min fairness power control in a 5G mMIMO network in\cite{power_control}. At last, the EMFE is analyzed in networks where sub-6~GHz and mmWave BSs coexist in\cite{9511258}. It is worth noting that all existing works characterizing the IPD represent the network topology by relying on a H-PPP, which cannot capture the spatial repulsions and attractions that characterize general network deployments. 
\subsubsection{Joint evaluation of the IPD and SINR by using SG}
In the context of energy harvesting, SG brought a new perspective to simultaneous wireless information and power transfer (SWIPT) analyses by allowing the computation of the joint CCDF in order to find a trade-off between coverage and harvested power. These studies take into account many features including line-of-sight and non-line-of-sight links, time switching and power splitting schemes, dynamic BF\cite{TVC_SWIPT} and mMIMO \cite{SWIPT_MIMO}. Similar tools can be used for a joint analysis of the EMFE and the SINR. In a recent study \cite{chen2023joint}, the authors delved into the analysis of the EMFE for DL and UL transmissions, originating from both BSs and active users' smartphones, alongside the SINR for DL communication. A more comprehensive investigation \cite{Qin24} extends the analysis to include passive EMFE from BSs and active users, while also considering the impact of EMFE limits on network coverage. Furthermore, resource allocation optimization to maximize the number of connected users to a BS is explored in \cite{Pardo24} via MC simulations, taking into account coverage requirements and EMFE considerations. However, a joint analysis of EMFE and SINR, merging them into a single metric, remains largely unexplored in the existing literature. To the best of the authors' knowledge, the only analytical frameworks studying EMFE and SINR jointly are reported in\cite{manhattan} for Manhattan networks and\cite{cell-free} for user-centric cell-free mMIMO networks. 

\subsection{Contributions}
Motivated by these considerations, the aims of the present paper are (i) to introduce an analytical framework for jointly evaluating the trade-offs between coverage and EMFE for two different PPs (MI and MV) and (ii) to validate the approach by using realistic datasets for sub-6 GHz macro BSs tailored to the large majority of urban and rural environments. The specific contributions of this paper are as follows:
    \subsubsection{Motion-invariant networks}~For MI networks, the proposed mathematical approach is based on a $\beta$-GPP. The first contribution is to complement the approaches described in Subsection \ref{sssec:IPD_soa} by developing a framework for calculating mathematical expressions for the following metrics:   
    \begin{itemize}
        \item Mean and variance of the EMFE;
        \item Marginal CDF of the EMFE;
        \item Joint CDF of the EMFE and SINR.
    \end{itemize}
    \subsubsection{Motion-variant networks} The MI assumption does not hold anymore for cities with a historic city center that is characterized by a high BS deployment density and a lower deployment density as the distance from the city center increases. In these scenarios, in addition, the BS density is often found to be angle-independent. We, therefore, go beyond classical studies applied to MI networks by introducing an I-PPP model for MV networks, which is characterized by a radial intensity measure. The second contribution is the development of a comprehensive framework for calculating mathematical expressions in MV networks. In particular, the following metrics are derived: 
    \begin{itemize}
        \item Mean and variance of the EMFE;
        \item Marginal CDF of the EMFE and CCDF of the SINR;
        \item Joint CDF of the EMFE and SINR.
    \end{itemize} 

Finally, the developed frameworks are applied to real-world networks for both the MI and the MV cases. The network performance is evaluated using realistic system parameters.

A comparison between this work and the most related works is summarized in Table~\ref{tab:comparison}.

{\small
\begin{table*}[t]
\begin{center}
\caption{Comparison between the relevant SG literature and this work. $^{*}$: Study of the SNR only.}
\begin{tabular}{|c|c|c|c|c|c|c|c|c|c|c|}
 \hline
 \multirow{3}{*}{Ref.} & \multicolumn{6}{c|}{Topology} & \multicolumn{4}{c|}{Mathematical performance metrics of interest} \\
 \hline
 & Spatial & Motion- & Real- & \multirow{2}{*}{Tractability} & Nakagami-$m$ & Dynamic & SI(N)R & IPD & Joint & \multirow{2}{*}{DL/UL} \\
  & repulsion & variance & world & & fading & beamforming & CCDF & CCDF & distribution & \\
 \hline
 \cite{GontierAccess, GontierICC, app10238753} &  &  & \checkmark & \checkmark &  &  &  & \checkmark &  & DL\\
 \hline
 \cite{Baccelli1997StochasticGA, tutorial} &  &  & \checkmark & \checkmark &  &  &  &  &  & /\\
 \hline
 \cite{baccellimetrics} &  &  &  & \checkmark &  &  & \checkmark &  &  & DL\\
 \hline
 \cite{hexagonal2016, GSPP} & \checkmark &  &  &  &  &  & \checkmark &  &  & DL\\
 \hline
 \cite{InterferenceLargeNetworkHaenggi, 6524460} & \checkmark &  &  &  & \checkmark &  & \checkmark &  &  & DL\\
 \hline
\cite{Lee2013, 6042301} &  &  & \checkmark &  &  &  & \checkmark &  &  & DL\\
 \hline
\cite{7733098} &  &  & \checkmark &  & \checkmark &  & \checkmark &  &  & DL+UL\\
 \hline
 \cite{PHCP, 8187697, Ginibre_theory, GinibreNaoto14} & \checkmark &  &  & \checkmark &  &  & \checkmark &  &  & DL\\
 \hline
 \cite{Kibilda2016, BGPP_Gomez, cost2022} & \checkmark &  & \checkmark & \checkmark &  &  &  &  &  & DL\\
 \hline
 \cite{Anjin2013, Zhang_2021, IDT, Baccelli_DPP2} & \checkmark &  & \checkmark & \checkmark &  &  & \checkmark &  &  & DL\\
 \hline
 \cite{uplink_meta, Bai16} &  &  &  &  &  &  & \checkmark &  &  & UL\\
 \hline
 \cite{alpha-stable} & \checkmark & \checkmark & \checkmark &  &  &  & \checkmark &  &  & DL\\
 \hline
 \cite{HE2, HE3} &  &  &  & \checkmark &  &  &  & \checkmark &  & DL\\
 \hline
 \cite{ECC} &  &  &  & \checkmark & \checkmark & \checkmark &  &  &  & DL\\
 \hline
 \cite{power_control} &  &  &  & \checkmark &  &  &  &  &  & DL\\
 \hline
 \cite{9511258} &  &  &  & \checkmark & \checkmark & \checkmark &  & \checkmark &  & DL\\
 \hline
 \cite{TVC_SWIPT, SWIPT_MIMO} &  &  &  & \checkmark &  & \checkmark &  &  & \checkmark & DL\\
 \hline
 \cite{chen2023joint} &  &  &  & \checkmark & \checkmark &  & \checkmark$^{*}$ & \checkmark &  & DL+UL\\
 \hline
 \cite{Qin24} &  &  &  & \checkmark &  & \checkmark & \checkmark & \checkmark &  & DL+UL\\
 \hline
 \cite{manhattan, cell-free} & \checkmark &  &  & \checkmark &  &  & \checkmark & \checkmark & \checkmark & DL\\
 \hline
  This work & \checkmark & \checkmark & \checkmark & \checkmark & \checkmark & \checkmark & \checkmark & \checkmark & \checkmark & DL\\
 \hline
\end{tabular}
\end{center}
\label{tab:comparison}
\end{table*}
}
\subsection{Structure of the Paper}
The paper is organized as follows: Section~\ref{sec:system_model} introduces the network topologies and the system model. Section~\ref{sec:analytical_results} provides mathematical expressions for the performance metrics of interest by using SG. Numerical validations based on two real-world network deployments are provided in Section~\ref{sec:numerical_results}. Finally, conclusions are given in Section~\ref{sec:conclusion}.

\section{System Model}
\label{sec:system_model}

\subsection{Mathematical Background}
\label{ssec:background}
Let $\mathcal{B} \in \mathbb{R}^2$ be the two-dimensional area where the considered network is located. Let $\Psi = \{\!X_i\!\}$ be the PP of BSs in $\mathcal{B}$, which are assumed to have the same technology, to belong to the same network provider and to transmit at the same frequency $f$. $\Psi$ is modeled as a realization of a PP of density $\lambda(u), u \in \mathcal{B}$. $\Psi$ is stationary if its statistical properties remain unaffected after any translation. It is isotropic if its statistical properties are invariant under any rotation. $\Psi$ is called MI if it is stationary and isotropic.

\subsubsection{Motion-Invariant Networks}\label{ssec:motion_invariant_model}
In the MI case, a good estimator of the true BS density $\lambda$ is the number of BSs within $\mathcal{B}$ divided by the area of $\mathcal{B}$. Motion-invariance implies that the performance metrics are statistically identical at any point in an infinite network. It is common practice to consider a typical user centered at the origin to facilitate the analysis. In a finite network, it is often assumed that the results do not vary as the typical user stays away from the boundaries of the area $\mathcal{B}$. 

In Section \ref{sec:analytical_results}, the performance metrics for MI networks will be derived for $\beta$-GPP models, which are characterized by a constant density $\lambda$ and a parameter $\beta$. This latter is between 0 and 1\cite{Ginibre_theory}, with $\beta\!=\!0$ corresponding to a~H-PPP and $\beta\!=\!1$ to a Ginibre Point Process (GPP). The main difficulty compared to other PP models comes from the fact that $\beta$-GPPs are constructed from GPPs, which are defined from a complex kernel\cite{Ginibre_theory}. For these PP models, an analytical expression for the probability density function (PDF) of the distance to the nearest BS, $f_{R_0, \Theta_0}$, can be obtained as follows. Let $X_i, i\in \mathbb{N}$ denote the points of a $\beta$-GPP with density $\lambda$, which is denoted by $\Phi_\lambda^\beta$. $\Phi_\lambda^\beta$ is constructed from the GPP $\Phi_{\lambda/\beta}^1 = \{{X}_i\}_{i\in \mathbb{N}}$ with density $\lambda/\beta$ by independent thinning with probability $\beta$. The set of square distances, $\{|{X}_i|^2\}_{i\in \mathbb{N}}$, has the same distribution as $\{Y_i\}$ such that $Y_i \sim \text{Gamma}(i, \pi\lambda/\beta), i \in \mathbb{N}$ are mutually independent\cite{GinibreNaoto14}. This property can be used to derive performance metrics since it is possible to take advantage of the existence of an analytical PDF for the square distance to each BS $X_i$. The PDF of $Y_i$ is given by {\small\begin{equation}\label{eq:BGPP_pdf}
    f_i(u) = {u^{i-1}}\, e^{-\frac{c u}{\beta}}\, (c/\beta)^i/{(i-1)!}
\end{equation}}
with $c\! = \!\lambda\pi$. To obtain $\Phi_\lambda^\beta$ from $\Phi_{\lambda/\beta}^1$, each element of $\Phi_\lambda^\beta$ is associated with a mark $\xi$. Then, $\{\xi_i\}_{i \in \mathbb{N}}$ is the set of marks of $\Phi_{\lambda/\beta}^1$, which are mutually independent and identically distributed random variables with $\mathbb{P}(\xi_i = 1) = \beta$ and $\mathbb{P}(\xi_i = 0) = 1-\beta$. Accordingly, the serving BS denoted by $X_0$ is $X_i$ if $\xi_i = 1$ and if all the other BSs are located farther than $X_i$, i.e. $\{X_0 = X_i\} = \{\xi_i = 1\} \cap \mathcal{A}_i$ where $\mathcal{A}_i = {\left\{\xi_j=1, |X_j| > |X_i|\right\} \cup \{\xi_j = 0\}} \forall j\in \mathbb{N} \setminus \{i\}$. Further details can be found in \cite{Ginibre_theory}.

\subsubsection{Motion-Variant Networks}\label{sssec:min_math}
If $\Psi$ is MV, an inhomogeneous PP model needs to be selected. Unfortunately, inhomogeneous models are intractable mathematically, except for the I-PPP, which can be viewed as an approximation of some more complex models. The computation of performance metrics requires the use of a spatially-varying density function $\lambda(u), u \in \mathcal{B}$, which is fitted to the empirical BS density by, e.g., using a least-square method. 

Large European cities are often characterized by the presence of a densely populated historic center, with old buildings and an irregular street organization, leading to a high density of BSs to accommodate the large data traffic. As the distance from the city center increases, the density of antennas decreases, leading to an almost radial density\cite{cost2022}. Based on these considerations, a flexible radial density model is chosen, which is monotonically decreasing with $\rho$ and is characterized by 6 parameters $\Tilde{a}$, $\Tilde{b}$, $\Tilde{c}$, $\Tilde{d}$, $\Tilde{\rho}$, $\Tilde{\theta}$:
{\smalltonormalsize
\begin{align}\label{eq:IPPP_density}
    \lambda(\rho, \theta) = \frac{\Tilde{a}}{\Delta(\rho, \theta)}+\Tilde{b}+\Tilde{c} \Delta(\rho, \theta) + \Tilde{d} \left(\Delta(\rho, \theta)\right)^2
\end{align}}with $\Delta^2(\rho, \theta) = {\Tilde{\rho}^2+\rho^2-2\Tilde{\rho} \rho \cos(\theta - \Tilde{\theta})}$. The user is assumed to be located at the origin of the coordinate system. The 4 parameters $\Tilde{a}$, $\Tilde{b}$, $\Tilde{c}$, $\Tilde{d}$ must be chosen so that the conditions $\lambda(\rho, \theta) \geq 0$ and $\partial \lambda / \partial \Delta \leq 0$ are met for $\rho \leq \tau$, where $\tau$ is the radius of the disk centered at the origin inside which the network is studied and can take arbitrary large or infinite values. For this choice of parameters, $(\Tilde{\rho}, \Tilde{\theta})$ is the point of maximal BS density. The mathematical tractability of the performance metrics increases if the intensity measure
{\smalltonormalsize
\begin{equation}\label{eq:intensity_measure}
    \Lambda(\mathcal{D}) = \int_{\mathcal{D}}\lambda(u)\, du,
\end{equation}}where $\mathcal{D} \in \mathcal{B}$, is purely radial, which is the case analyzed here. The mathematical expressions derived for MV networks in this paper are only valid for the typical user located at the origin of the coordinate system but, as shown in Subsection~\ref{ssec:analysis_IPPP}, they can be calculated at other user locations through a change of coordinate system so that the user is located at the center of this new coordinate system. Practically speaking, the approach consists of replacing $\lambda(u)$ by $\lambda(u-u_C)$ in all the expressions, where $u_C$ is the point of interest in the old coordinate system.

\subsection{Propagation Model}
\label{ssec:propagation_model}
The propagation model is defined as
\begin{equation}\label{eq:model}
    P_{r, i} = P_t G_{i} G_r|h_i|^2 l_i
\end{equation}
where $P_{r, i}$ is the received power from BS $X_i$, $P_t$ is the transmit power of $X_i$, $G_{i}$ is the gain of $X_i$ in the direction of the user, $G_r$ is the gain of the receiver, assumed isotropic and equal to 1 to simplify the analysis, $|h_i|^2$ accounts for the fading and
$l_i = l(X_i) = \kappa^{-1} \left(r_i^2+z^2\right)^{-\alpha/2}$ is the path loss attenuation with exponent $\alpha > 2$, $\kappa = (4\pi f/c_0)^2$ where $c_0$ is the speed of light, $r_i$ is the distance between the user and the BS $X_i$, and $z > 0$ is the height of $X_i$. 

In the following, a Nakagami-$m$ fading model is assumed, which means that $|h_i|^2$ is gamma distributed with a shape parameter $m$ and a scale parameter $1/m$. The CDF of $|h_i|^2$ is then $F_{|h|^2}(x) = \gamma(m,m x)/\Gamma(m)$ where $\gamma(\cdot,\cdot)$ is the lower incomplete gamma function and $\Gamma(\cdot)$ is the gamma function.
 
The BSs are assumed to employ dynamic BF, and an analysis based on BS databases \cite{Cartoradio, IBGE} revealed that BSs implementing dynamic BF have significantly lower sidelobes compared to the main lobe. Additionally, the analysis indicated that network providers equip each BS with three similar antenna arrays oriented at 120$^\circ$ intervals, each capable of covering 120$^\circ$. From these considerations and for the sake of mathematical tractability, the actual antenna array patterns are approximated using a sectored antenna model. The gain of the antenna array of a generic BS can be expressed as follows:
{\smalltonormalsize
\begin{equation}\label{eq:gain}
    G_t = \begin{cases}
         G_{\text{max}} \qquad &\text{ if } |\theta| \leq \omega\\
         0 \qquad &\text{ if } |\theta| > \omega
    \end{cases}
\end{equation}}
where $\omega \in [0, 2\pi/3]$ and $\theta \in [0;2\pi[$ is the angle off the boresight direction and $\omega$ is the beamwidth of the main lobe. The serving BS is assumed to estimate the angle of arrival and to adjust its antenna steering orientation accordingly. The model assumes no alignment error. The beams of all the interfering BSs are considered to be randomly oriented with respect to each other, uniformly distributed in $[0;2\pi[$. Consequently, the gains, $G_i$, of the interfering BSs in the direction of the user are Bernouilli random variables with probability $3\omega/(2\pi)$, which is the probability that the typical user is illuminated by the main lobe of each interfering BS. 
For notational simplicity, $G_{\text{max}} = 1$ is normalized, and then we can simplify the notation as $P_t G_{\text{max}} \to P_t$.

The mathematical expressions derived in the following section are defined for a circular area $\mathcal{B}$ of radius $\tau$ located in the $z=0$ plane and centered on the calculation point. The calculations take an exclusion radius $r_e \geq 0$ into account, representing a non-publicly accessible area around the BSs. A closest BS association policy is assumed. 

Define $\Bar{P}_{r,i}\! =\! \Bar{P}_{r,i}(r_i^2) \!= \!P_t l_i$ to simplify the notation. This quantity accounts for the maximal received power at distance $R_i$ from $X_i$, averaged over fading. It incorporates both the transmit power and the attenuation due to the path loss. Let $S_0\! =\! \Bar{P}_{r,0}|h_0|^2$ be the useful power received from the serving BS $X_0$ that is assumed to be the closest to the user and let $I_0 = \sum_{i \in \Psi \setminus\left\{X_0\right\}} \Bar{P}_{r,i} G_{i} |h_i|^2$ be the aggregate interference. Based on these definitions, the SINR conditioned on the distance to the serving BS is given by
{\smalltonormalsize
\begin{align}\label{eq:SINR}
    \text{\normalfont{SINR}}_0 = \frac{S_0}{I_0 + \sigma^2}
\end{align}}where $\sigma^2$ is the thermal noise power. In the following, the performance metrics will be derived for the user DL power EMFE defined as
{\smalltonormalsize
\begin{align}\label{eq:expP}
    \mathcal{P} = \sum_{i \in \Psi} \Bar{P}_{r,i}\,G_{i}\, |h_i|^2 = S_0+I_0,
\end{align}}which can be converted into a total IPD, expressed in W/m$^2$,~as
{\smalltonormalsize
\begin{equation}\label{eq:ExpWM2}
    \mathcal{S} = \sum_{i \in \Psi} \frac{P_{t}\,G_{i}\,|h_i|^2}{4\pi \left(r_i^2+z^2\right)^{\alpha/2}} = \frac{\kappa}{4\pi}\mathcal{P}\footnote{Note that the IPD is frequency-independent. The frequency dependence of $\kappa$ in the right-hand term of the relationship $\mathcal{S} = \frac{\kappa}{4\pi}\mathcal{P}$ cancels out due to the $\kappa^{-1}$ dependence of $\mathcal{P}$.}
\end{equation}}by definition, and, finally, into a root-mean-square electric field strength in V/m as
{\smalltonormalsize\begin{equation}\label{eq:ExpVM}
    E[\text{V/m}] = \sqrt{120 \pi \mathcal{S}}.
\end{equation}}

\section{Mathematical Framework}
\label{sec:analytical_results}
\begin{figure}[!ht]
    \begin{center}
    \includegraphics[width=1\linewidth, trim={5cm, 4cm, 5cm, 5cm}, clip]{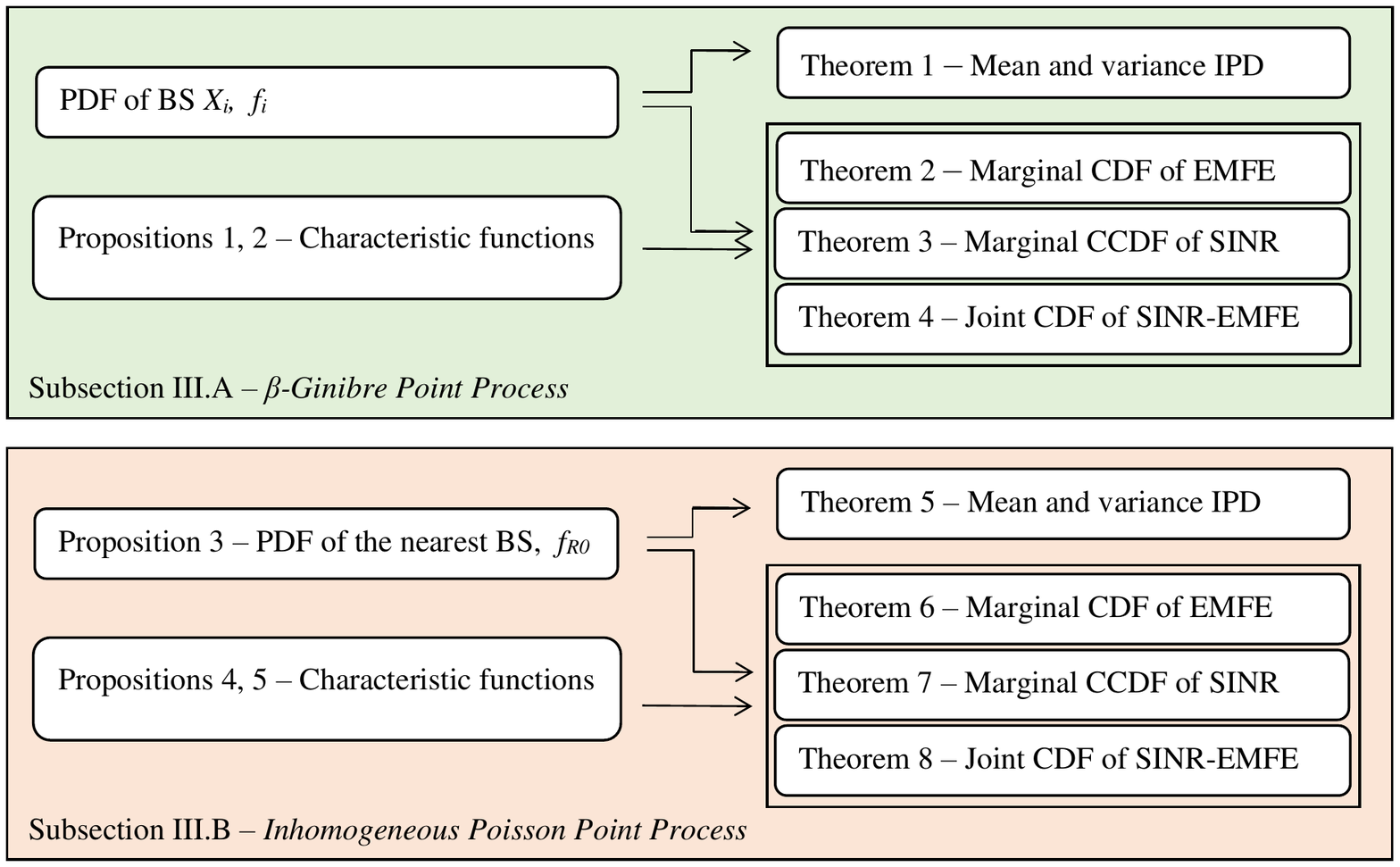}
    \end{center}
    \caption{Organization of Section \ref{sec:analytical_results}}
    \label{fig:Schema section III}
\end{figure}

The structure of this section is illustrated in Fig.~\ref{fig:Schema section III}. Let $T$ be the {\normalfont SINR} threshold for reliable transmission at the user equipment, and let $T'$ be the maximum allowed level of DL EMFE in power scale. The CDFs, CCDFs and PDFs can be calculated using Gil-Pelaez's inversion theorem\cite{gil-pelaez} from the knowledge of (i) the characteristic functions (CFs) of the useful signal $\phi_S$ and of the interference $\phi_I$, and (ii) the density function $f_{R_0, \Theta_0}$ of the position to the serving BS denoted by the random variables $(R_0, \Theta_0)$.

\subsection{Motion-Invariant Networks: $\beta$-Ginibre Point Process}
\label{ssec:mi} 
Based on the considerations in Subsections~\ref{ssec:background} and \ref{ssec:propagation_model}, the mean and the variance of EMFE can be calculated.
\begin{theorem}\label{eq:BGPP_meanexp}The mean of the EMFE of a $\beta$-GPP given the propagation model in~\eqref{eq:model} is given by
{\small
\begin{multline}\label{eq:BGPP_meanexpeq}
        \kern-.6em\!\mathbb E\left[\mathcal{P}\right] \!=\! \beta  \int_{r_e^2}^{\tau^2}\!\left[\left(\Bar{P}_{r}(u) + \frac{2 p_g c \left[\Bar{P}_{r}(r^2)\left(r^2+z^2\right)\right]_{r = \tau}^{r = \sqrt{u}}}{\alpha-2} \!\!\right) \!\Omega(u)\right.\\
        \left.-p_g \beta\int_{u}^{\tau^2} \Omega^*(u,v) \Bar{P}_{r}(v) dv\right] du
\end{multline}} where 
\begin{equation*}
    \Omega(u) = \sum_{i \in \mathbb{N}} f_i(u) \Upsilon_i^\beta(u), \quad \Omega^*(u,v) = \sum_{i \in \mathbb{N}} f_i(u) f_i(v) \Upsilon_i^\beta(u),
\end{equation*}
\begin{equation*}
    \Omega^{**}(u,v,w) = \sum_{i \in \mathbb{N}} f_i(u) f_i(v) f_i(w)\Upsilon_i^\beta(u),
\end{equation*}
\begin{equation*}
    \Upsilon_i^\beta(u) = \prod_{j \in \mathbb{N} \setminus \{i\}}\Bigg(1-\beta+\beta \,\frac{\Gamma\left(j, \frac{c u}{\beta}, \frac{c\tau^2}{\beta}\right)}{(j-1)!}\Bigg).
\end{equation*}
The associated variance is given by $\mathbb V\left[\mathcal{P}\right] = \mathbb E\left[\mathcal{P}^2\right] - \left(\mathbb E\left[\mathcal{P}\right]\right)^2$, where  $\mathbb E\left[\mathcal{P}^2\right]$ is given in \eqref{eq:BGPP_varexpeq} shown at the top of the next page. We denote $[f(x)]_{x=a}^{x=b} = f(b) - f(a)$ for ease of writing.\begin{figure*}[!ht]
{\footnotesizetosmall
\begin{equation}\label{eq:BGPP_varexpeq}
    \begin{split}
        &\mathbb E\left[\mathcal{P}^2\right] =\! \int_{r_e^2}^{\tau^2}\! \Omega(u)\left(\beta \tfrac{m+1}{m} \Bar{P}_{r}^2(u) + \beta p_g \, \Bar{P}_{r}(u) \tfrac{2 c}{\alpha-2} \left[\Bar{P}_{r}(r^2)\left(r^2+z^2\right)\right]_{r = \tau}^{r = \sqrt{u}}+\tfrac{m+1}{m} \tfrac{p_g c}{\alpha-1} \left[\Bar{P}_{r}^2(r^2)\left(r^2+z^2\right)\right]_{r = \tau}^{r = \sqrt{u}}\right.\\
        &\quad-\left.p_g^2 c^2 \!\int_{u}^{\tau^2}\!\int_{u}^{\tau^2}\!e^{-\tfrac{c(v+w)}{\beta}} I_0\left(2c\sqrt{vw}/\beta\right)\Bar{P}_{r}(v)\Bar{P}_{r}(w) dw\,dv+\left(\tfrac{2 p_g c }{\alpha-2} \left[\Bar{P}_{r}(r^2)\left(r^2+z^2\right)\right]_{r = \tau}^{r = \sqrt{u}}\right)^2\right)  du\\
        &\quad-\!\int_{r_e^2}^{\tau^2}\! \int_{u}^{\tau^2}\! p_g \beta \Bar{P}_{r}(v)\!\left[\Omega^*(u, v)\!\left( \beta  \Bar{P}_{r}(u)+ \tfrac{m+1}{m} \Bar{P}_{r}(v) \!+\! \tfrac{4 p_g c }{\alpha-2} \!\left[\Bar{P}_{r}(r^2)\left(r^2+z^2\right)\right]_{r = \tau}^{r = \sqrt{u}}  \right)-2 p_g \beta  \int_{u}^{\tau^2}\! \Omega^{**}(u,v,w) \Bar{P}_{r}(w) dw\right] dv \,du
    \end{split}
\end{equation}}
{\footnotesizetosmall
\begin{multline}\label{eq:BGPP_varexpeq_nobf}
    \mathbb E\left[\mathcal{P}^{*2}\right] = \tfrac{2(m+1)}{m(\alpha-1)}\,\left[\Bar{P}_{r}^2(r)\left(r^2+z^2\right)\right]_{r = \tau}^{r = r_e} + \left(\tfrac{2  c}{\alpha-2} \left[\Bar{P}_{r}(r^2)\left(r^2+z^2\right)\right]_{r = \tau}^{r = r_e}\right)^2 - c^2 \!\int_{r_e^2}^{\tau^2}\!\int_{r_e^2}^{\tau^2}\!e^{-\tfrac{c(v+w)}{\beta}} I_0\left(\tfrac{2c\sqrt{vw}}{\beta}\right)\Bar{P}_{r}(v)\Bar{P}_{r}(w) dw\,dv
\end{multline}}
{\small
\begin{equation}\label{eq:BGPP_jointeq}
\begin{split}
    G(T,T') &\triangleq \mathbb E_0\left[ \mathbb{P}\left[\frac{S_0}{I_0+\sigma^2} > T,S_0+I_0 < T'\right]\right]\! \\
    &= \beta \int_{{r_e^2}}^{{\tau^2}}\!\left(\frac{\Omega(u)}{2}F_{|h|^2}\!\left(\frac{T'}{\Bar{P}_{r}(u)}\right)- \frac{1}{\pi q}\int_{0}^{\infty} \text{\normalfont Im}\left[ \sum_{i \in \mathbb{N}} f_i(u) \Upsilon_i^\beta(u)\phi_{I,i}(q|u) \zeta(T,T',\Bar{P}_{r}(u)) \right]  dq \right)du
\end{split}
\end{equation} with
\begin{equation}\label{eq:zeta}
    \zeta(T,T',\Bar{P}_{r})= \frac{m^m}{(m-1)!}\frac{\gamma\left(m, \frac{-T''}{\Bar{P}_{r}}\left(m+j\frac{q\Bar{P}_{r}}{T}\right)\right)
    }{(m+j\frac{q\Bar{P}_{r}}{T})^m}\,e^{jq\sigma^2} + \frac{m^m}{(m-1)!}\frac{\Gamma\left(m; \frac{T' \left(m-jq\Bar{P}_{r}\right)}{\Bar{P}_{r}}, \frac{T''\left(m-jq\Bar{P}_{r}\right)}{\Bar{P}_{r}}\right)}{(m-jq\Bar{P}_{r})^m}e^{-jqT'}.
\end{equation}
}
\hrulefill
\end{figure*}
\end{theorem}
\begin{proof}
    The proof of \eqref{eq:BGPP_meanexpeq} is provided in Appendix~\ref{sec:BGPP_mean_proof} and the proof of \eqref{eq:BGPP_varexpeq} is provided in Appendix~\ref{sec:BGPP_var_proof}.
\end{proof}
The mean has the same expression as for a PPP, as obtained in \cite{GontierAccess} and \cite{ECC}, without and with dynamic BF, respectively. The expressions of the mean and the variance of the EMFE are given for a network of BSs not employing dynamic BF in Corollary~\ref{eq:BGPP_meanexp_nobf}.
\begin{corollary}\label{eq:BGPP_meanexp_nobf}The mean and the variance of the EMFE of a $\beta$-GPP when no dynamic BF is employed are respectively given by
{\smalltonormalsize
\begin{align}\label{eq:BGPP_meanexpeq_nobg}
    \mathbb E\left[\mathcal{P^*}\right] &= \frac{2}{\alpha-2}\, c \left[\Bar{P}_{r}(r)\left(r^2+z^2\right)\right]_{r = \tau}^{r = r_e}
\end{align}} and $\mathbb V\left[\mathcal{P^*}\right] = \mathbb E\left[\mathcal{P^*}^2\right] - \left(\mathbb E\left[\mathcal{P^*}\right]\right)^2$, where  $\mathbb E\left[\mathcal{P^*}^2\right]$ is given in \eqref{eq:BGPP_varexpeq_nobf} shown at the top of the next page. 
\end{corollary}
\begin{proof}
    The proof of \eqref{eq:BGPP_meanexpeq_nobg} is provided at the end of Appendix~\ref{sec:BGPP_mean_proof} and the proof of \eqref{eq:BGPP_varexpeq_nobf} is provided at the end of Appendix~\ref{sec:BGPP_var_proof}.
\end{proof}

The calculation of the higher-order moments of the EMFE and the moments of the interference or SINR can be obtained similarly. The CFs of the interference and EMFE are provided in Proposition~\ref{eq:BGPP_PhiI} and Corollary~\ref{eq:BGPP_PhiEmf_nobf}.

\begin{proposition}\label{eq:BGPP_PhiI}\textit{The CF of the interference of a $\beta$-GPP for the propagation model in~\eqref{eq:model}, conditioned on the BS $X_i$ located at a distance $r_i=|X_i|$ is}
{\small
\begin{align*}
    \Phi_{I,i}(q|u) = \prod\limits_{\substack{k \in \mathbb{N} \\ k \neq i}} \!\Bigg[\int_{u}^{{\tau^2}} \!f_k(v)  \frac{p_g\beta}{\left(1- {j q \Bar{P}_{r}(v)}/{m}\right)^m} dv+\!1\!-p_g\beta\Bigg].
\end{align*}
}
\end{proposition}
\begin{proof} 
The proof is provided in Appendix \ref{sec:BGPP_PhiI_Proof}.
\end{proof}
It is worth noting that from the Cauchy-Schwarz theorem and from the definition of a PDF, the integral is always smaller than $\beta$, proving that the infinite product converges. To mitigate numerical inaccuracies, the product can be truncated to the $N$th term. The impact of such truncation will be analyzed in Section~\ref{ssec:analysis_BGPP}.

Based on this proposition, the CDF of the EMFE, the CCDF of the coverage and the joint CDF of the EMFE and SINR are provided in Theorems~\ref{eq:exp_BGPP}, \ref{eq:BGPP_Cov} and \ref{eq:BGPP_joint}, respectively.
\begin{theorem}\label{eq:exp_BGPP}The CDF of the EMFE of a $\beta$-GPP for the propagation model in~\eqref{eq:model} is given by
{\footnotesize
\begin{equation*}
\begin{split}
    &F_{\text{emf}}(T') \triangleq \mathbb{P}\left[\mathcal{P} < T'\right] \\
    &= \beta \!\int_{r_e^2}^{\tau^2}\!\left(\frac{\Omega(u)}{2}\!-\!\int_0^\infty \! \text{\normalfont Im} \!\left[\sum_{i \in \mathbb{N}} f_i(u) \Upsilon_i^\beta(u)\phi_{E,i}(q|u)\frac{e^{-j q T'}}{\pi q}\!\right]\!dq\!\right)\!du.
\end{split}
\end{equation*}}where $\phi_{E,i}(q|u) = \Phi_S(q|u)\, \Phi_{I,i}(q|u)$ and
\begin{equation*}
    \Phi_S(q|u) = \mathbb{E}_{|h|}\left[\exp\left(j q S(u)\right)\right] = (1-jq \Bar{P}_{r}(u)/m)^{-m}.
\end{equation*}
\end{theorem}
\begin{proof}
The proof is provided in Appendix~\ref{BGPP_emf_proof}.
\end{proof}
It is noteworthy that, although Gil-Pelaez's theorem is widely applied in SG, it may not be easy to be computed numerically. This originates from the highly oscillatory integrands, particularly caused by the presence of complex exponential functions. To mitigate these challenges, we recommend expressing $\mathbb{P}\left[\mathcal{P}/T' < 1\right]$ instead of $\mathbb{P}\left[\mathcal{P} < T'\right]$. As a result, the inner component of the imaginary part operator can be represented as $\phi_{E}(q/T') e^{-jq}$. In a network of BSs not employing BF, the expression of $F_{emf}(T')$ simplifies as stated in Corollary~\ref{eq:BGPP_PhiEmf_nobf}.

\begin{corollary}\label{eq:BGPP_PhiEmf_nobf}The CF of the EMFE of a $\beta$-GPP where no dynamic BF is employed is given by
\begin{equation*}
        F_{\text{emf}}^*(T') =  \frac{1}{2}- \int_{0}^{\infty} \frac{1}{\pi q}\text{\normalfont Im}\left[\phi_{E}^*(q)\, e^{-jqT'} \right] dq.
\end{equation*}where $\phi_{E}^*(q) = \Phi_I(q|r_e^2)$ with $p_g = 1$.
\end{corollary}
\begin{proof}
The proof of the CF is similar to the proof of Proposition~\ref{eq:BGPP_PhiI} in Appendix~\ref{sec:BGPP_PhiI_Proof} except that $p_g\!=\!1$ and that the summation and product are over $k \in \mathbb{N}$ instead of ${k \in \mathbb{N} \setminus \{i\}}$. The proof of the CDF is similar to the proof of Theorem~\ref{eq:exp_BGPP} in Appendix~\ref{BGPP_emf_proof} except that the index $i$ of the serving BS can be ignored.
\end{proof}

\begin{theorem}\label{eq:BGPP_Cov}The CCDF of the SINR of a $\beta$-GPP for the propagation model in \eqref{eq:model} is given by
{\footnotesizetosmall
\begin{multline*}
    F_{\text{cov}}(T) \triangleq \mathbb E_0\left[\mathbb{P}\left[\text{\normalfont SINR}_0 > T\right]\right] \\
    = \!\beta \!\int_{{r_e^2}}^{{\tau^2}} \!\left(\frac{\Omega(u)}{2}\! +\!\int_0^\infty \! \text{\normalfont Im} \!\left[\sum_{i \in \mathbb{N}} \!f_i(u) \Upsilon_i^\beta(u)\phi_{\text{\normalfont{SINR}},i}(q|u)\!\right]\!\frac{dq}{\pi q}\!\right)\!du
\end{multline*}}where $\phi_{\text{\normalfont{SINR}},i}(q|u) = \phi_S(q|u)  \phi_{I,i}(-Tq|u) e^{-jTq\sigma^2}$.

\end{theorem}
\begin{proof}
The proof is similar to the proof of Theorem~\ref{eq:exp_BGPP} in Appendix~\ref{BGPP_emf_proof}. The only differences are that for the CCDF of the SINR, $F_i(T') = 1-\mathbb P\left(S_0(i) - T I_0(i) \leq T\sigma^2, \mathcal{A}_i\right)$ and that $\phi_{E,i}(q|u)$ must be replaced by $\phi_{\text{\normalfont{SINR}},i}(q|u)$.
\end{proof}
In the case of Rayleigh fading, the CCDF of the SINR has been computed in\cite{GinibreNaoto14} for a GPP and in\cite{Ginibre_theory} for a $\beta$-GPP. Theorem~\ref{eq:BGPP_Cov} is more general since it can be applied to other fading distributions, such as the Nakagami-$m$ fading given as an example, and because it includes dynamic BF. 

\begin{theorem}\label{eq:BGPP_joint}The joint CDF of the EMFE and SINR for a $\beta$-GPP is given in \eqref{eq:BGPP_jointeq} shown at the top of the page.
\end{theorem}
\begin{proof}
The proof is provided in Appendix~\ref{sec:join_proof}.
\end{proof}
It can be noticed that the conditional distributions of the EMFE and coverage can also be calculated from these three CDFs using Bayes' rule.

\subsection{Motion-Variant Networks: Inhomogeneous Poisson Point Process}
\label{ssec:nmi}

\begin{proposition}\label{eq:IPPP_density_fct}
The PDF of the distance from the user located at the origin to the nearest BS, is given by
\begin{align*}
    f_{R_0}(r_0)= \frac{\Lambda^{\!(1)\!}(r_0)\,e^{-\Lambda(r_0)}}{e^{-\Lambda(r_e)}-e^{-\Lambda(\tau)}}
\end{align*}where 
\begin{multline}\label{eq:intensity_function}
    \Lambda(r) = 4\Tilde{a} \int_0^{r} \, |u-\Tilde{\rho}|^{-1}  K\left(u\right) u \,du + \pi \Tilde{b} r^2\\
    +4 \Tilde{c} \int_0^{r} \, |u-\Tilde{\rho}|  E\left(u\right) u \,du\, + \pi \Tilde{d} \left(\frac{r^4}{2}+\Tilde{\rho}^2\, r^2\right)
\end{multline}is the intensity measure,
{\small
\begin{equation}\label{eq:intensity_der}
    \Lambda^{\!(1)\!}(r)\! = \!2\pi r\!\left(\frac{2 \Tilde{a} K\left(r\right)}{\pi|r-\Tilde{\rho}|} +  \Tilde{b}+\frac{2 \Tilde{c}}{\pi}  |r-\Tilde{\rho}|  E\left(r\right) + \Tilde{d} \left(r^2+\Tilde{\rho}^2\right)\right)
\end{equation}}is its derivative,
{\small
\begin{align*}
    K(u) = \int_0^{\pi/2} \frac{1}{\sqrt{1- k(u) \, \sin^2{\varphi}}} \, d\varphi = \frac{\pi}{2}\, _2F_1\left(\frac{1}{2}, \frac{1}{2}, 1 ;k(u)\right)
\end{align*}}is the complete elliptic integral of the first kind,
{\footnotesizetosmall
\begin{align*}
    E(u) = \int_0^{\frac{\pi}{2}} \sqrt{1- k(u)  \sin^2{\varphi}} \, d\varphi = \frac{\pi}{2}\, _2F_1\left(\frac{1}{2}, \frac{-1}{2}, 1 ;k(u)\right),
\end{align*}}is the complete elliptic integral of the second kind and $k(u) = - {4u\,\Tilde{\rho}}{\left(u-\Tilde{\rho}\right)^{-2}}$.
\end{proposition}
\begin{proof}
    From the definition of Poisson's law, it follows that the void probability, i.e., the probability of having 0 BS within a disk of radius $r_0$ centered at the origin of an infinite network, is $\exp\left({-\Lambda(r_0)}\right)$ with $\Lambda$ the intensity measure of the PP. The probability of having the nearest BS at a distance $r_0$ is then $1-\exp\left({-\Lambda(r_0)}\right)$. The nearest BS can be located in a ring of inner radius $r_e$ and outer radius $\tau$. Thus, the CDF of the distance to the nearest BS is given by
    \begin{equation*}
        F_{R_0}(r_0)= \frac{e^{-\Lambda(r_0)}}{e^{-\Lambda(r_e)}-e^{-\Lambda(\tau)}}
    \end{equation*}
    so that the integration over the ring gives 1. Proposition~\ref{eq:IPPP_density_fct} is then obtained by differentiating this CDF with respect to $r_0$. The intensity measure is obtained from \eqref{eq:intensity_measure}.
\end{proof}

We calculate hereafter the mean and variance of the EMFE.

\begin{figure*}[!t]
\begin{multline}\label{eq:IPPP_meanexpeq_nobf}
    \Bar{P}_{I_{r}}(r_0) = p_g \int_{r_0}^{\tau} \Bar{P}_{r}(r) \Lambda^{\!(1)\!}(r) dr= 
    - \tfrac{2\pi \Tilde{d}\,p_g }{(4-\alpha)(\alpha-2)} \left[\left((\alpha-2) r^2+2z^2\right) \Bar{P}_{r}(r) \left(r^2+z^2\right)\right]_{r = \tau}^{r = r_0}\\
    + p_g \left(2\pi \Tilde{b}+2\pi \Tilde{d} \Tilde{\rho}^2\right)  \frac{\left[\Bar{P}_{r}(r)\left(r^2+z^2\right)\right]_{r = \tau}^{r = r_0}}{\alpha-2} + 4 p_g \int_{r_0}^{\tau} \Bar{P}_{r}(r) r \left(\tfrac{\Tilde{a}}{|r-\Tilde{\rho}|}   K\left(r\right)+\Tilde{c}   |r-\Tilde{\rho}|   E\left(r\right)\right) dr
\end{multline}
{\normalsize
\begin{multline}\label{eq:IPPP_varexpeq_nobf}
    \Bar{P}_{I^2_{r}}(r_0) = p_g \frac{m+1}{m(\alpha-1)} \left(\pi \Tilde{b}+\pi \Tilde{d} \Tilde{\rho}^2\right)  \left[\Bar{P}_{r}^2(r)\left(r^2+z^2\right)\right]_{r = \tau}^{r = r_0}+4 p_g \frac{m+1}{m} \int_{r_0}^{\tau} \Bar{P}_{r}^2(r) r \left(\tfrac{\Tilde{a}}{|r-\Tilde{\rho}|}   K\left(r\right)+\Tilde{c}   |r-\Tilde{\rho}|   E\left(r\right)\right) dr\\
    - p_g \tfrac{(m+1) \pi \Tilde{d}}{m(\alpha-2)(\alpha-1)} \left[\left((\alpha-1) r^2+z^2\right) \Bar{P}_{r}^2(r) \left(r^2+z^2\right)\right]_{r = \tau}^{r = r_0} + p_g^2 \Bar{P}_{I_{r}}^2(r_0).
\end{multline}
}
\hrulefill
\end{figure*}
\begin{theorem}\label{eq:IPPP_meanexp}
The mean of the EMFE for an I-PPP network with the intensity measure in~\eqref{eq:intensity_function} for the propagation model in~\eqref{eq:model} is given by 
\begin{equation*}
    \mathbb E\left[\mathcal{P}\right] = \int_{r_e}^{\tau} \left(\Bar{P}_{r}(r_0) + \Bar{P}_{I_{r}}(r_0)\right)f_{R_0}(r_0)dr_0
\end{equation*}
where $\Bar{P}_{I_{r}}(r_0)$ is given in 
\eqref{eq:IPPP_meanexpeq_nobf} shown at the top of the page. The associated variance is given by $\mathbb V\left[\mathcal{P}\right] = \mathbb E\left[\mathcal{P}^2\right] - \left(\mathbb E\left[\mathcal{P}\right]\right)^2$, where  $\mathbb E\left[\mathcal{P}^2\right]$ is given by {\footnotesizetosmall
\begin{equation}\label{eq:IPPP_varexpeq}
    \mathbb E\left[\mathcal{P}^2\right] \!= \!\int_{r_e}^{\tau} \!\!\left(\tfrac{m+1}{m} \Bar{P}_{r}^2(r_0) + \Bar{P}_{r}(r_0) \Bar{P}_{I_{r}}(r_0)+ \Bar{P}_{I^2_{r}}(r_0)\right)f_{R_0}\!(r_0)dr_0
\end{equation}} with $\Bar{P}_{I^2_{r}}(r_0)$ given in \eqref{eq:IPPP_varexpeq_nobf} shown at the top of the page.
\end{theorem}
\begin{proof}
    The proof is provided in Appendix~\ref{sec:IPPP_mean_proof}.
\end{proof}
\begin{corollary}
\label{eq:IPPP_meanexp_nobf}
The mean of the EMFE for an I-PPP network with the intensity measure in~\eqref{eq:intensity_function} when no dynamic BF is employed is given by $\mathbb E\left[\mathcal{P^*}\right] = \Bar{P}_{I_{r}}(r_e)$ with $p_g = 1$. The associated variance is given by $\mathbb V\left[\mathcal{P^*}\right] = \mathbb E\left[\mathcal{P^*}^2\right] - \left(\mathbb E\left[\mathcal{P^*}\right]\right)^2$, where  $\mathbb E\left[\mathcal{P^*}^2\right]$ is given by $\Bar{P}_{I^2_{r}}(r_e)$.
\end{corollary}
\begin{proof}
    If no dynamic BF is used, there is no need to analyze the terms $S_0$ and $I_0$ separately. Campbell's theorem can therefore be applied to the sum $S_0+I_0$.
\end{proof}

It is worth noting that the mean of the EMFE is proportional to the BS density. The calculation of the higher-order moments of the EMFE and the moments of the interference or SINR can be obtained similarly. The CF of the interference is provided in Proposition~\ref{eq:IPPP_PhiI}.
\begin{proposition}\label{eq:IPPP_PhiI}
The CF of the interference for an I-PPP with the intensity measure in~\eqref{eq:intensity_function} for the propagation model in~\eqref{eq:model}, conditioned on the distance between the typical user at the origin and the nearest BS, is given in~\eqref{eq:IPPP_PhiIeq} shown at the top of the next page.

\begin{figure*}[!t]
{\footnotesizetosmall
\begin{multline}\label{eq:IPPP_PhiIeq}
    \phi_{I}(q|r_0) = \exp\left[4 p_g\int_{r_0}^\tau \frac{\Tilde{a} |\Tilde{\rho}-u|^{-1}\, K(u)+\Tilde{c}|\Tilde{\rho}-u|\, E(u)}{\left(1-j q \Bar{P}_{r}(u)/m\right)^m} \,u du-\left[\Lambda(r)\right]_{r = r_0}^{r = \tau}\right.\\
    + \left.p_g\!\left[\tfrac{\pi\Tilde{d}}{\alpha-2} \left(r^2+z^2\right)^{2-\alpha} \! _2F_1\left(m, 2-\tfrac{4}{\alpha}, 3-\tfrac{4}{\alpha}; \frac{j q \Bar{P}_{r}(r)}{m}\right)+\pi\tfrac{\Tilde{b}+\Tilde{d}(\Tilde{\rho}^2\!-\!z^2)}{\alpha-1}  \left(r^2+z^2\right)^{1-\alpha} \! _2F_1\!\left(m, 2\!-\!\tfrac{2}{\alpha}, 3\!-\!\tfrac{2}{\alpha}; \frac{j q \Bar{P}_{r}(r)}{m}\right)\right]_{r = r_0}^{r = \tau}\right]
\end{multline}
}
{\small
\begin{multline}\label{eq:joint_eq}
    G(T,T') \triangleq \mathbb E_0\left[ \mathbb{P}\left[\frac{S_0}{I_0+\sigma^2} > T,S_0+I_0 < T'\right]\right] = \int_{r_0} \left[ \frac{1}{2}  F_{|h|^2}\left(\tfrac{T'}{\Bar{P}_{r,0}}\right) - \!\int_{0}^{\infty} \!\frac{1}{\pi q}\text{\normalfont Im}\left[ \phi_I(q|r_0)\,\zeta(T,T',\Bar{P}_{r})\right]dq\right] f_{R_0}(r_0) dr_0
\end{multline}}
where $\zeta(T,T',\Bar{P}_{r, 0})$ is given by \eqref{eq:zeta}.

\hrulefill
\end{figure*}
\end{proposition}
\begin{proof}
The proof is provided in Appendix~\ref{proof_IPPP_PhiI}.
\end{proof}

Based on these propositions, the CDF of the EMFE, the CCDF of the coverage and the joint CDF are provided as Theorems~\ref{eq:exp_IPPP}, \ref{eq:cov} and \ref{eq:joint}, respectively.

\begin{theorem}\label{eq:exp_IPPP}\textit{The CDF of the EMFE for an I-PPP with the intensity measure in~\eqref{eq:intensity_function} for the propagation model in~\eqref{eq:model} is}
\begin{multline*}
        F_{\text{emf}}(T') \triangleq \mathbb{P}\left[\mathcal{P} < T'\right]\\ 
        =  \int_{r_0}\left(\frac{1}{2}- \int_{0}^{\infty} \frac{1}{\pi q}\text{\normalfont Im}\left[\phi_{E}(q|r_0)\, e^{-jqT'} \right] dq\right)f_{R_0}(r_0) dr_0
\end{multline*}where $\phi_{E}(q|r_0) = \Phi_S(q|r_0)\, \Phi_{I}(q|r_0)$ and
\begin{align*}
    \phi_S(q|r_0) = \mathbb{E}_{|h|^2}\left[\exp\left(j q S_0\right)\right] = \left(1-jq\Bar{P}_{r,0}/m\right)^{-m}.
\end{align*}
\end{theorem}
\begin{proof}
    The proof follows from Gil-Pelaez’s theorem.
\end{proof}
In a network of BSs not employing BF, the expression of $F_{emf}(T')$ simplifies as stated in Corollary~\ref{eq:IPPP_PhiEmf_nobf}.

\begin{corollary}\label{eq:IPPP_PhiEmf_nobf}The CF of the EMFE of an IPPP where no dynamic BF is employed is given by
\begin{equation*}
        F_{\text{emf}}^*(T') =  \frac{1}{2}- \int_{0}^{\infty} \frac{1}{\pi q}\text{\normalfont Im}\left[\phi_{E}^*(q)\, e^{-jqT'} \right] dq.
\end{equation*}with $\phi_{E}^*(q) = \Phi_I(q|r_e^2)$ with $p_g = 1$.
\end{corollary}
\begin{theorem}\label{eq:cov}The CCDF of the SINR for an I-PPP with the intensity measure in~\eqref{eq:intensity_function} for the propagation model in~\eqref{eq:model} is given by
{\smalltonormalsize
\begin{multline*}\label{eq:coveq}
        F_{\text{cov}}(T) \triangleq \mathbb E_0\left[\mathbb{P}\left[\text{\normalfont SINR}_0 > T\right]\right] \\
        =\int_{r_0}\left(\frac{1}{2} + \int_{0}^{\infty} \text{\normalfont Im}\left[\phi_{\text{\normalfont{SINR}}}(q, T|r_0)\right]\,\frac{1}{\pi q} dq \right)f_{R_0}(r_0) dr_0
\end{multline*}}where $\phi_{\text{\normalfont{SINR}}}(q, T|r_0) = \phi_S(q|r_0)  \phi_I(-Tq|r_0) \exp\left({-jTq\sigma^2}\right)$.
\end{theorem}
\begin{proof}
    The proof is obtained by applying Gil-Pelaez’s theorem and by applying the expectation over the distance to the nearest BS $R_0$ and over the angle $\Theta_0$.
\end{proof}

\begin{theorem}\label{eq:joint}Let $T'' = T(T'+\sigma^2)/(1+T)$. The trade-off between the EMFE and network coverage for an I-PPP is provided by the joint CDF of the EMFE and SINR, given in \eqref{eq:joint_eq} shown at the top of the next page.

\end{theorem}
\begin{proof}
 The proof is provided in Appendix~\ref{sec:join_proof}.
\end{proof}

The expressions of Theorems~\ref{eq:cov} and \ref{eq:joint} can be simplified because the integral over $\theta_0$ is equal to $2\pi$.

\section{Numerical results}
\label{sec:numerical_results}

\subsection{Motion-Invariant Networks: $\beta$-Ginibre Point Process}
\label{ssec:analysis_BGPP}

In this section, the performance of a $\beta$-GPP network is analyzed based on the system model described in Section~\ref{sec:system_model}. In order to have a realistic set of parameters, the 5G~NR~2100 network of a major network provider in Paris (France) is studied. The BSs located inside a disk of radius $\tau \!= \!\numprint[m]{6000}$ and center (652, 6862)~km in the Lambert~93 conformal conic projection\cite{L93} are considered. The $\beta$-GPP with $\beta \!= \!0.75$ was found to be the best fit for this MI network by applying the methodology proposed in\cite{cost2022}. The system parameters extracted from the operator's database \cite{Cartoradio} are given in Table~\ref{tab:BGPP_study}. The half-power beamwidth of the observed antenna patterns, 0.0982~rad, is taken as the value of $\omega$, leading to $p_g\! = \!0.0469$. Based on propagation models used in similar urban environments\cite{GontierAccess, Ichitsubo00, COST231}, we set the path loss exponent to $\alpha\! =\! 3.2$. The noise power is $\sigma^2 \!= \!10 \log_{10}(k\,T_0\,B_w) + 30 + \mathcal{F}_{dB}$ in dBm where $k$ is the Boltzmann constant, $T_0$ is the standard temperature (290~K), $B_w$ is the bandwidth and $\mathcal{F}_{dB} = \numprint[dB]{6}$ is the receiver noise figure\cite{10.5555/3294673}. The carrier frequency $f$ and the bandwidth $B_w$ correspond to the official band allocated to the operator in DL. Numerical results are presented for Rayleigh fading ($m \!= \!1$), justified by the considered frequency band for non line-of-sight propagation environments.
{\small
\begin{table}[h!]
\begin{center}
\caption{System parameters used for Subsection~\ref{ssec:analysis_BGPP}}
\begin{tabular}{ |c|c||c|c| } 
 \hline
 $f$ & \numprint[MHz]{2132.7} & $B_w$ & \numprint[MHz]{14.8}\\
 $\lambda$ & $\numprint[BS/km^2]{6.17}$ & $r_e$ & \numprint[m]{0}\\
 $\alpha$ & 3.2 & $\tau$ & \numprint[km]{6} \\
 $P_t G_{max}$ & \numprint[dBm]{66} & $\beta$ & 0.75 \\
 $z$ & \numprint[m]{33} & $\sigma^2$ & \numprint[dBm]{-96.27} \\
 $p_{g}$ & 0.0469 &  & \\ 
 \hline
\end{tabular}
\end{center}
\label{tab:BGPP_study}
\end{table}}

\subsubsection{Numerical validation}

The mean EMFE is found to be $\numprint[W/m^2]{1.38e-4}$ ($\numprint[V/m]{0.23}$) with a variance of $\numprint[W^2/m^4]{3.92e-7}$ using Theorem~\ref{eq:BGPP_meanexp}, \eqref{eq:ExpWM2} and \eqref{eq:ExpVM}. The value of the mean EMFE obtained using $10^8$ realizations of Monte Carlo (MC) simulations shows a difference of $\numprint[W/m^2]{2e-8}$, showing that the approximation made in the proof in Appendix~\ref{sec:BGPP_mean_proof} is insignificant. The marginal distributions of the EMFE (Theorem~\ref{eq:exp_BGPP}) and SINR (Theorem~\ref{eq:BGPP_Cov}) are validated in Figs.~\ref{fig:Exp_BGPP_Paris_BOUYGUES_5G_NR_2100} and~\ref{fig:Cov_BGPP_Paris_BOUYGUES_5G_NR_2100} via MC simulations. The convergence of the corresponding expressions is also illustrated as a function of $N$, representing the number of terms considered within the summation in $F_{\textit{emf}}(T')$ and $F_{\textit{cov}}(T)$. Figs.~\ref{fig:Exp_BGPP_Paris_BOUYGUES_5G_NR_2100} and~\ref{fig:Cov_BGPP_Paris_BOUYGUES_5G_NR_2100} illustrate that truncating the infinite sums and products to $N\!=\! 10$ yields a highly accurate approximation of the CDFs. Extending the truncation to $N \!=\! 50$ results in a maximal absolute error of 0.4\% for the CDF of the EMFE in the head of the distribution and an absolute error of 0.2\% in the CCDF of the SINR. The area around the 95th percentile has been magnified in Fig.~\ref{fig:Exp_BGPP_Paris_BOUYGUES_5G_NR_2100} as it is an important statistical measure for EMFE assessment. By way of comparison, the Paris authorities set the maximum cumulative EMFE threshold at 5~V/m 900MHz equivalent\cite{charte_paris}, i.e. 7.44~V/m at \numprint[MHz]{2132.7}, which corresponds to -6.36~dBm using the formulas in \eqref{eq:ExpWM2} and \eqref{eq:ExpVM}. The values of EMFE obtained in this section are well below the legal threshold but we draw the reader's attention to the fact that the legal limit corresponds to a cumulative sum of EMFE caused by all operators and all cellular frequency bands, whereas this study deals with a single frequency band of a single operator at a time.

\begin{figure}[h!]
\centering
    \includegraphics[width=0.7\linewidth, trim={3cm, 9cm, 4cm, 10cm}, clip]{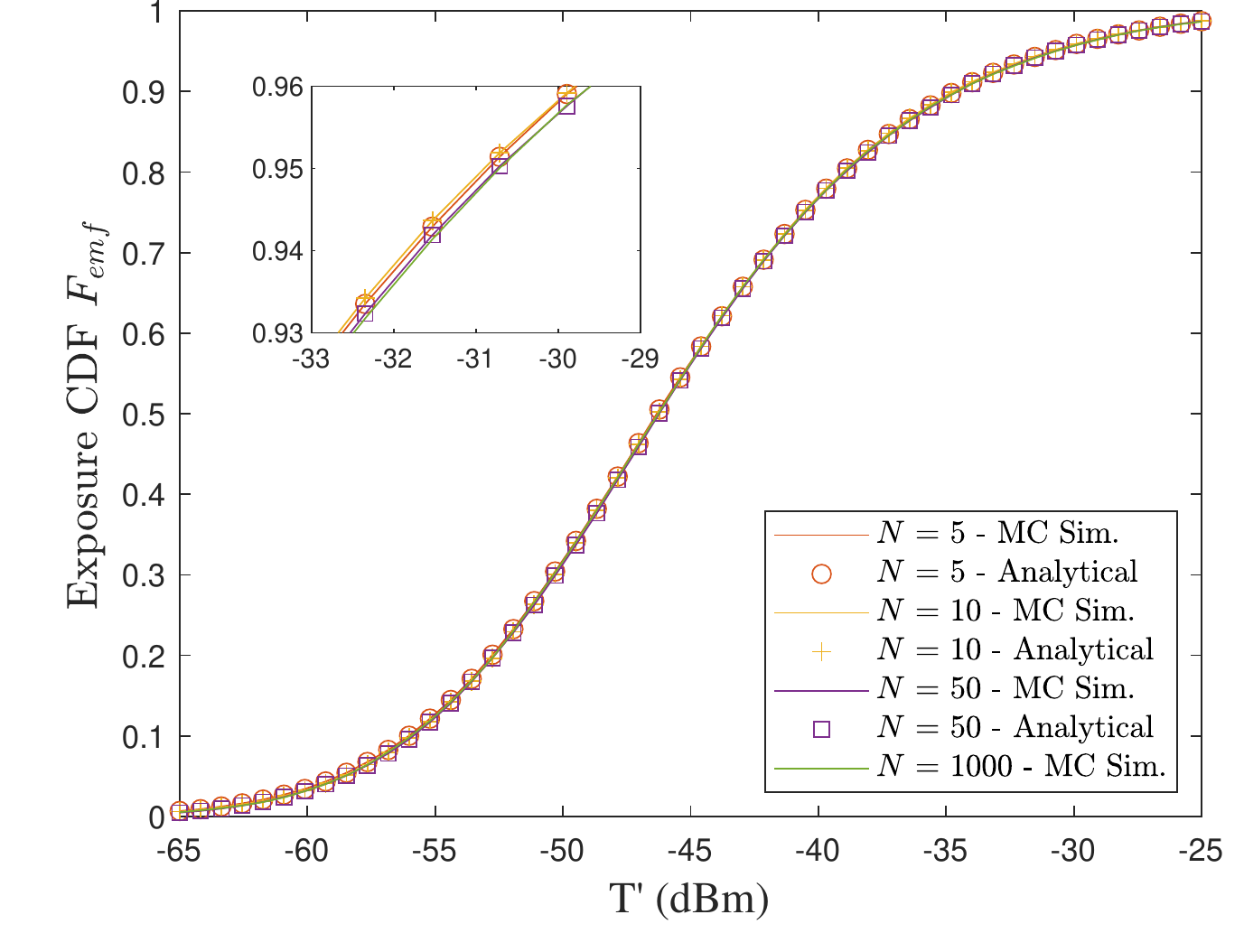}
    \caption{CDF of EMFE with the model parameters from Table~\ref{tab:BGPP_study}
    }
    \label{fig:Exp_BGPP_Paris_BOUYGUES_5G_NR_2100}
\end{figure}
\begin{figure}[h!]
  \centering
    \includegraphics[width=0.7\linewidth, trim={3cm, 9cm, 4cm, 10cm}, clip]{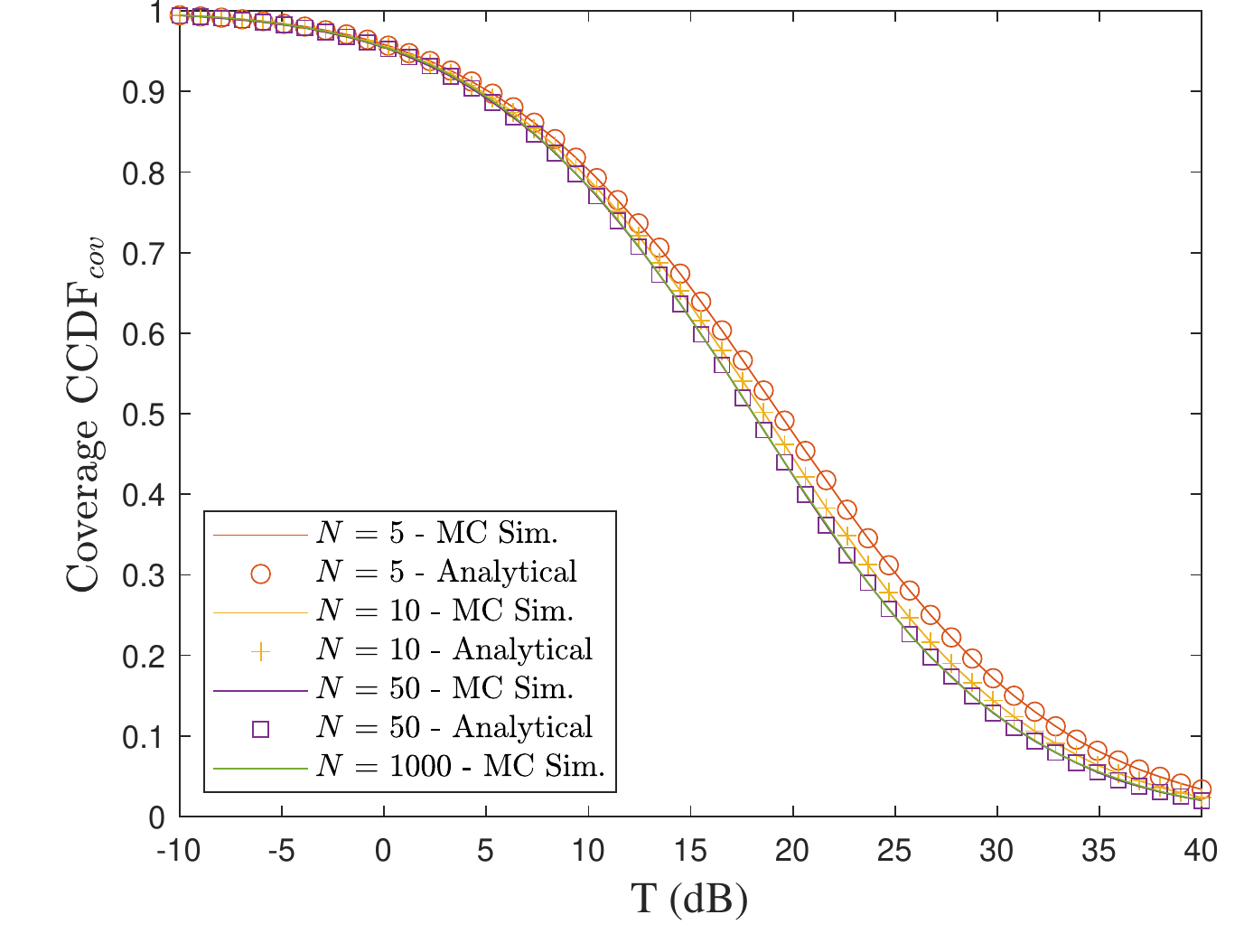}
    \caption{CCDF of the SINR with the model parameters from Table~\ref{tab:BGPP_study}
    }
    \label{fig:Cov_BGPP_Paris_BOUYGUES_5G_NR_2100}
\end{figure}

Isocurves depicting the joint CDF of the EMFE and SINR are presented in Fig.~\ref{fig:Joint_iso_BGPP_Paris_BOUYGUES_5G_NR_2100} for different values for the probability $p$ such that $G(T,T') = p$. Due to the intricate nature of the equations, the analytical computation of $G(T,T')$ can be time-consuming. To address this, an efficient solution is to determine the Fréchet lower bound (FLB) and the Fréchet upper bound (FUB), which are respectively defined by
\begin{equation}
       \text{FLB} = \max\left(0, F_{\text{cov}}(T) + F_{\text{emf}}(T') - 1\right)
\end{equation}and
\begin{equation}
       \text{FUB} = \min\left(F_{\text{cov}}(T), F_{\text{emf}}(T')\right). 
\end{equation}
such that $\text{FLB} \leq G(T,T') \leq \text{FUB}$.
Observing Fig.~\ref{fig:Joint_iso_BGPP_Paris_BOUYGUES_5G_NR_2100}, it is evident that the lower bound is close to $G(T,T')$ and becomes even closer as $p$ increases. In the pursuit of a conservative analysis, the lower bound can be employed as a substitute for the more intricate expression of $G(T,T')$.

\begin{figure}[t!]
    \centering
    \includegraphics[width=0.8\linewidth, trim={3cm, 9cm, 3cm, 10cm}, clip]{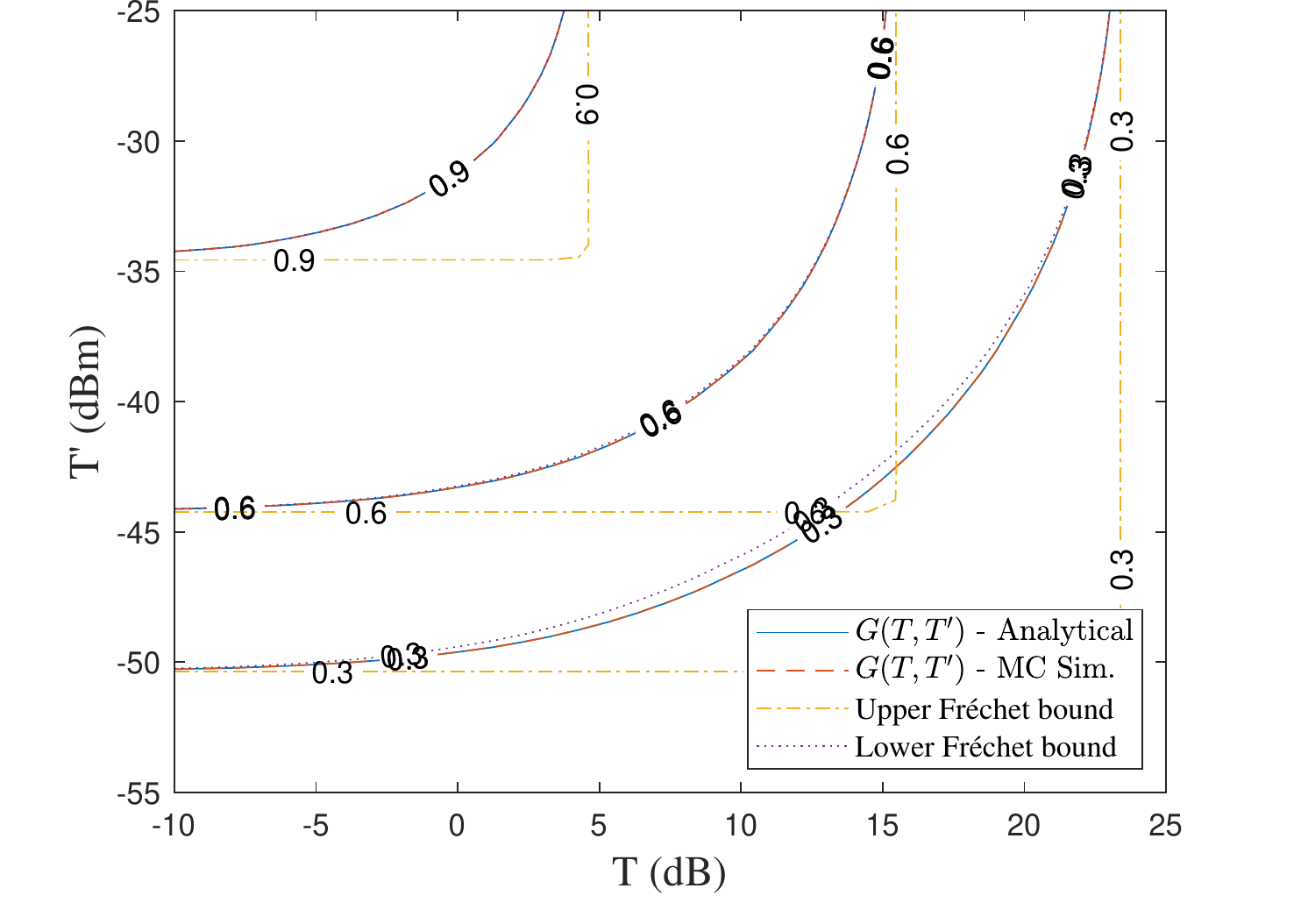}
    \caption{Isocurves of the joint CDF of the EMFE and SINR with the model parameters from Table~\ref{tab:BGPP_study}.
    }
    \label{fig:Joint_iso_BGPP_Paris_BOUYGUES_5G_NR_2100}
\end{figure}

\subsubsection{Impact of the $\beta$-parameter}
The marginal distributions of the EMFE (Theorem~\ref{eq:exp_BGPP}) and SINR (Theorem~\ref{eq:BGPP_Cov}) for several values of $\beta$ are shown in Fig.~\ref{fig:Exp_BGPP_study_beta} and \ref{fig:Cov_BGPP_study_beta}, respectively. Recalling that the limiting case $\beta=0$ corresponds to a H-PPP and $\beta=1$ corresponds to a GPP with more regularity, the EMFE is lower when the distributions of points is more random while the coverage is improved in networks that are more regularly deployed. This can be explained from three observations. First, in a more regular network, the typical user is on average closer to the serving BS but also to the most interfering BSs, as can be observed from the simulations in Fig.~\ref{fig:BGPP_study_beta_interpretation}. Second, the $n$th nearest BS gets closer on average as $n$ is smaller. Third, the signal power decreases as $r^{-\alpha}$. As a consequence, the more regular the network, the larger the power of the useful signal and the larger this latter power in comparison with the power of the interfering signals. To maximize both the SINR and the EMFE, the joint CDF of the EMFE and SINR can be analyzed.

\begin{figure}[h!]
  \centering
    \includegraphics[width=0.65\linewidth, trim={3cm, 8.5cm, 4cm, 10cm}, clip]{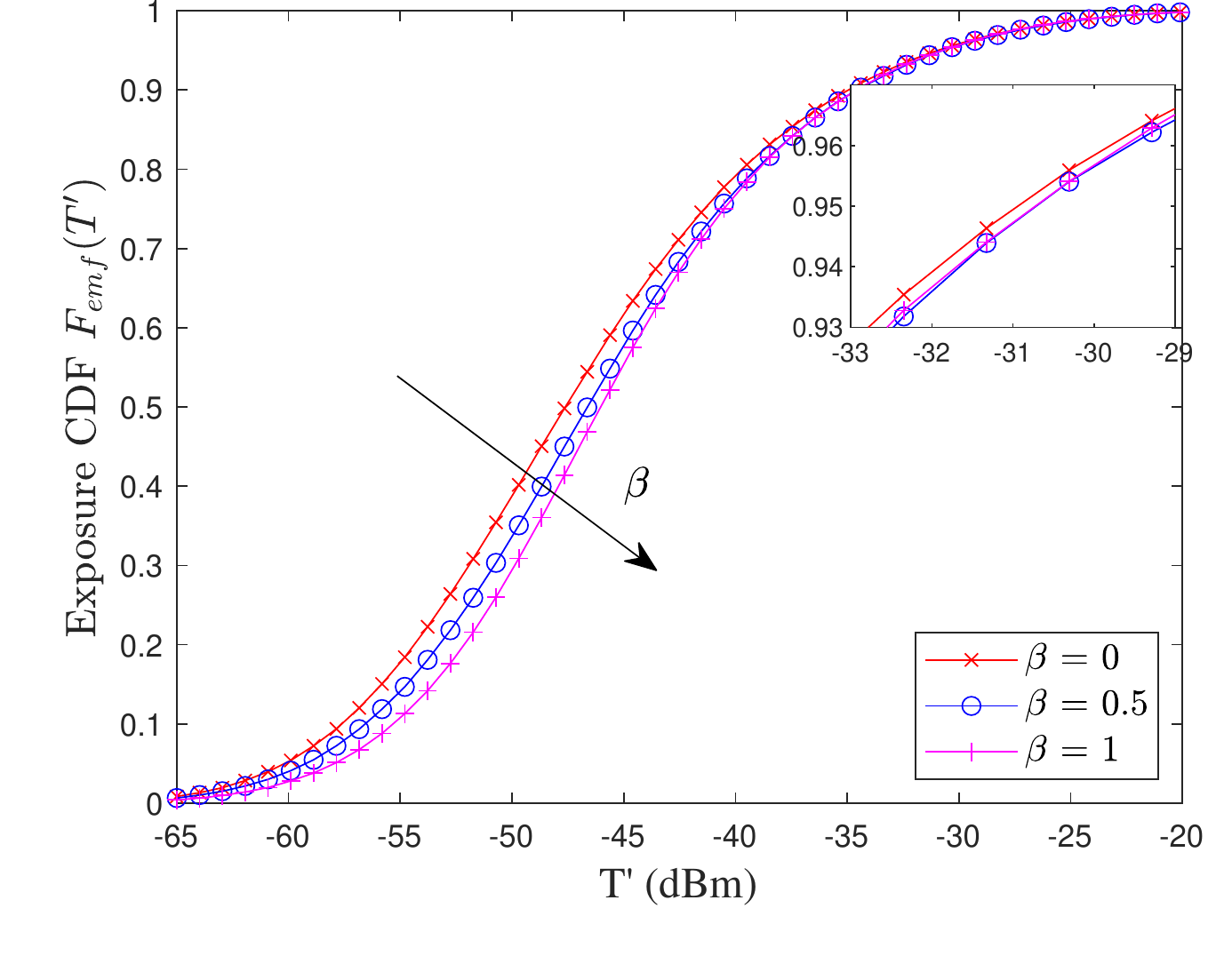}
    \caption{CDF of the EMFE for different values of $\beta$ in a $\beta$-GPP}
    \label{fig:Exp_BGPP_study_beta}
\end{figure}
\begin{figure}[h!]
  \centering
    \includegraphics[width=0.65\linewidth, trim={3cm, 8.5cm, 4cm, 10cm}, clip]{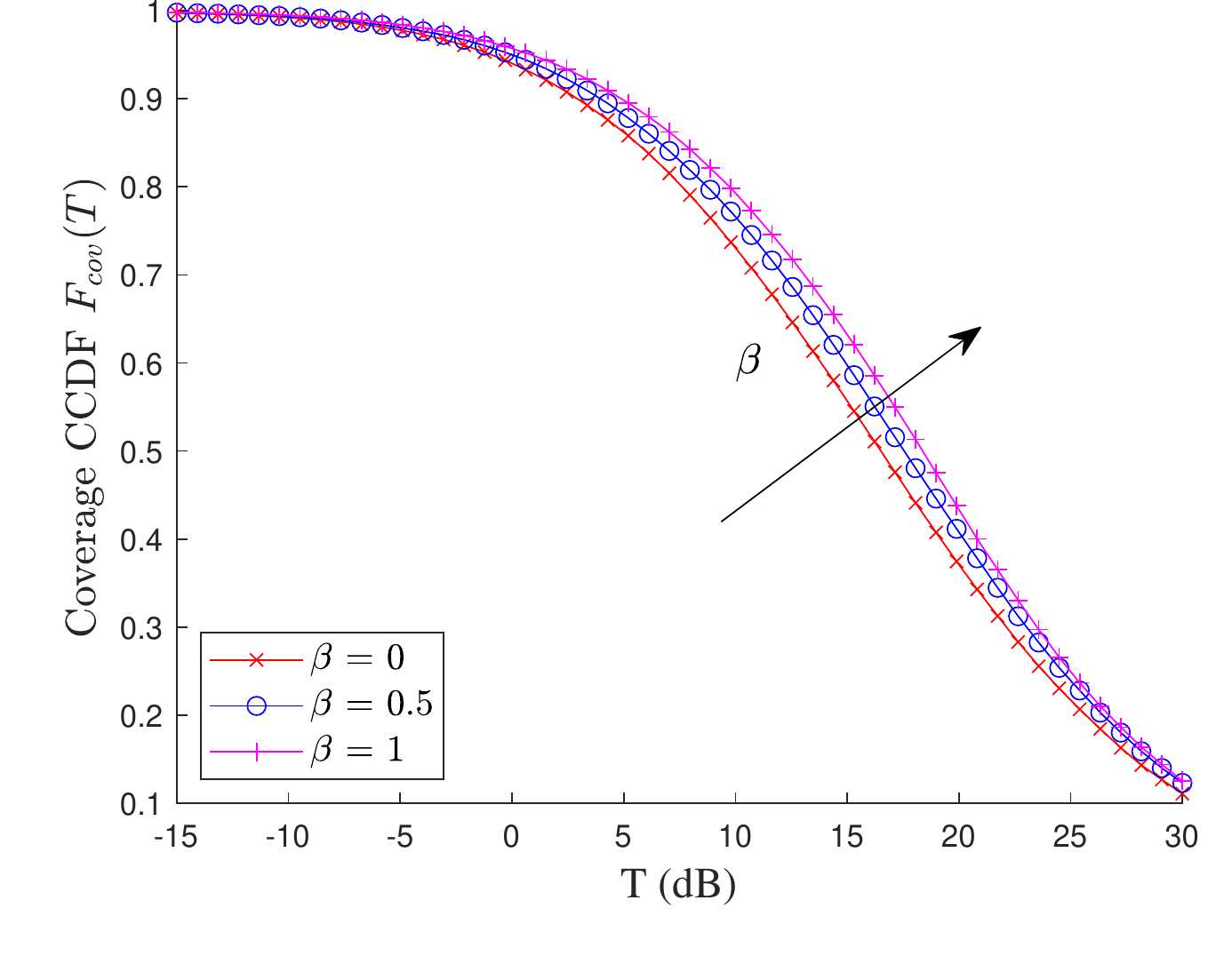}
    \caption{CCDF of the SINR for different values of $\beta$ in a $\beta$-GPP}
    \label{fig:Cov_BGPP_study_beta}
\end{figure}
\begin{figure}[h!]
    \centering
    \includegraphics[width=0.65\linewidth, trim={3cm, 9cm, 3cm, 10cm}, clip]{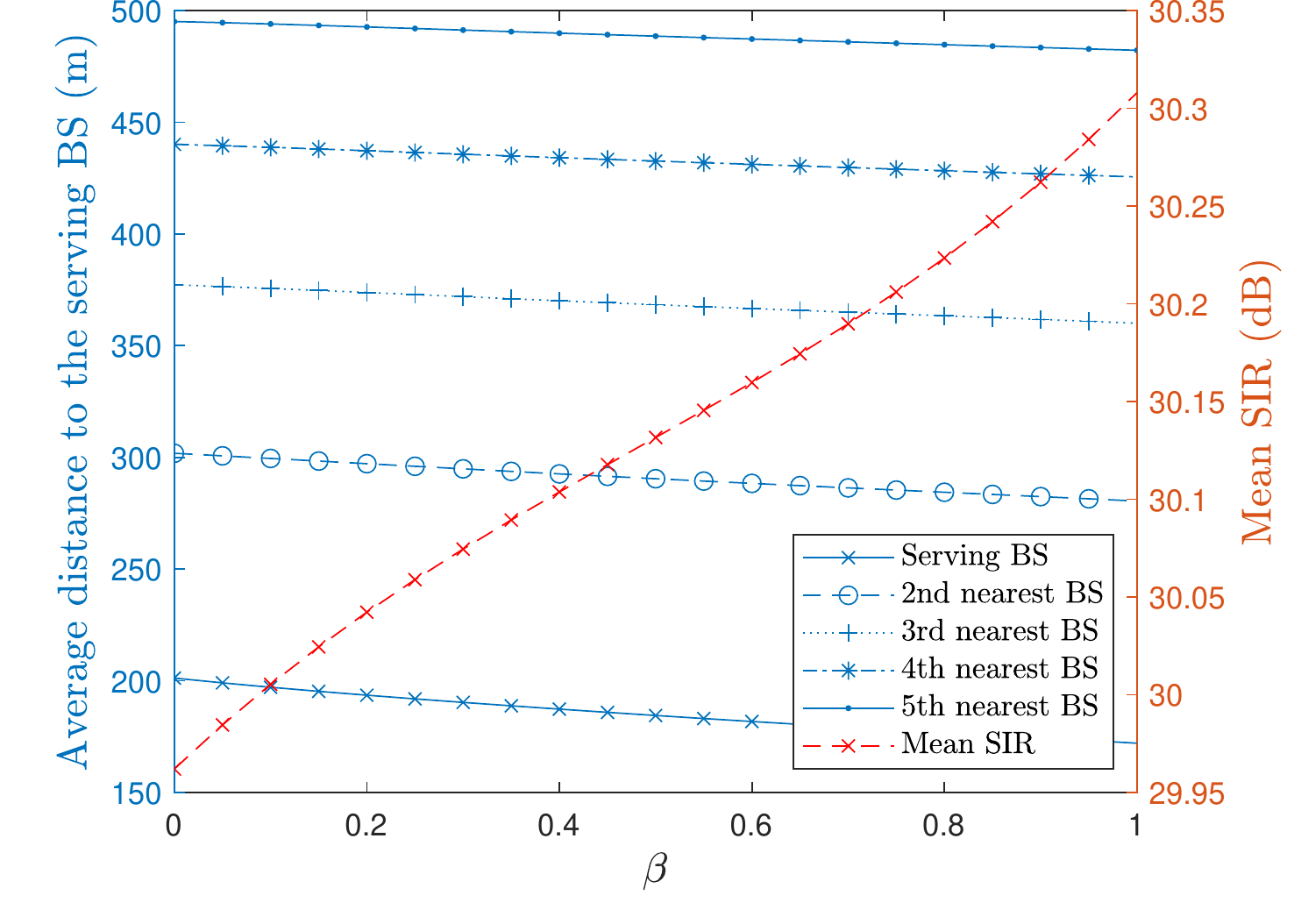}
    \caption{Average distance to the serving BS and average SINR as a function of $\beta$}
    \label{fig:BGPP_study_beta_interpretation}
\end{figure}
\begin{figure}[h!]
    \centering
    \includegraphics[width=0.65\linewidth, trim={3cm, 8.5cm, 3cm, 10cm}, clip]{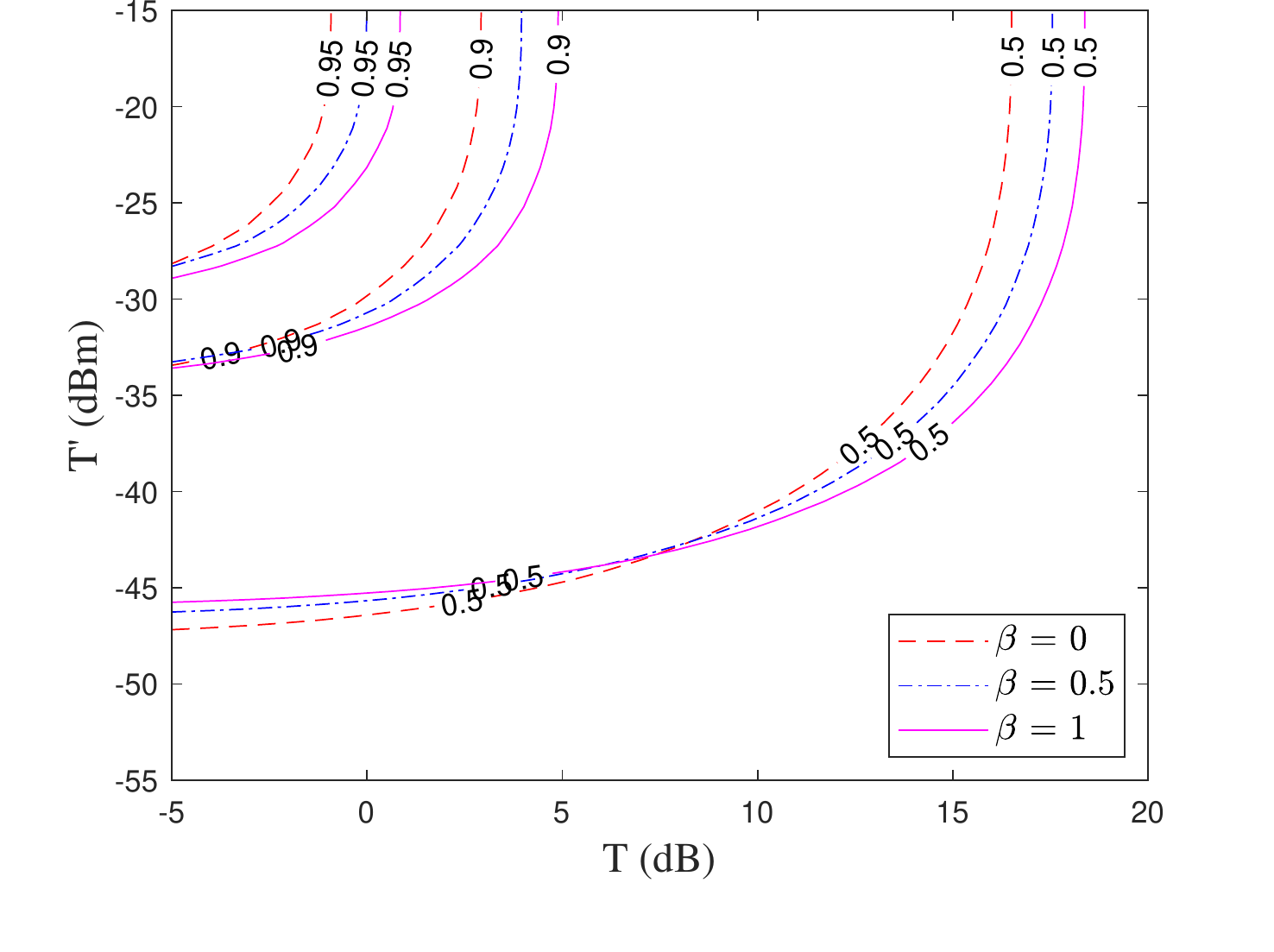}
    \caption{Isocurves of the joint CDF of the EMFE and SINR
    }
    \label{fig:Joint_BGPP_study_beta_isocurves}
\end{figure}

Isocurves of the joint CDF of the EMFE and SINR are shown in Fig.~\ref{fig:Joint_BGPP_study_beta_isocurves} for $G(T, T') = 0.5$, $G(T, T') = 0.9$ and $G(T, T') = 0.95$ for a H-PPP, a $\beta$-GPP with $\beta = 0.5$ and a GPP. In general, for a fixed pair $(T, T')$, the higher $\beta$ the better the network performance. For low values of $T'$, the trend is opposite. In other words, the benefit of an increase in $\beta$ in terms of coverage outweighs the negative impact on the EMFE, provided that the EMFE threshold is not too low. Figs. \ref{fig:Exp_BGPP_study_beta}, \ref{fig:Cov_BGPP_study_beta} and \ref{fig:Joint_BGPP_study_beta_isocurves} reveal the error made by modeling the Paris network with a simple H-PPP rather than a $\beta$-GPP. The Kolmogorov-Smirnov distance between the two, indicating the maximal absolute error, is 7.3\% for the CDF of the EMFE, 6.1\% for the CCDF of the SINR and 7.3\% for the joint CDF. These results highlight the importance of an accurate representation of the network topology.

\subsection{Motion-Variant Networks: Inhomogeneous Poisson Point Process}
\label{ssec:analysis_IPPP}

In this section, we illustrate the results for the I-PPP based on the density model in~\eqref{eq:IPPP_density}. As an example, we consider a MV network that corresponds to a cellular network in Brussels, Belgium. At the time of writing, no 5G network has been deployed in Brussels. We decided to use the LTE~1800 BS network of a major telecommunications operator whose data have been extracted from the database \cite{IBGE}. These BSs do not use dynamic BF. The zone under investigation is a disk centered at $(797,7085)$~km in Lambert~93 projection. It is the origin of the coordinate system in which the calculations are done. In order to cover the whole Brussels-Capital region, the radius of the disk is set at \numprint[km]{12} as can be seen in Fig~\ref{fig:IPPP_comparison_PPP_density}. The fitted parameters and the network parameters are listed in Table~\ref{tab:IPPP_study}. The central carrier frequency $f$ and the bandwidth $B_w$ correspond to the official band allocated to the operator for the DL. Numerical results are presented for Rayleigh fading ($m\! =\! 1$), justified by the considered frequency band for non line-of-sight propagation environments. In the following, the network performance is calculated for several user locations, using the change of variable explained in Subsection~\ref{sssec:min_math}. The area under investigation is each time a disk with arbitrary radius $\tau\! = \!\numprint[km]{7}$, contained in the larger \numprint[km]{12}-radius disk.
{\small
\begin{table}[h!]
\caption{System parameters used for subsection \ref{ssec:analysis_BGPP}}
\begin{center}
\begin{tabular}{ |c|c||c|c| } 
 \hline
 $f$ & \numprint[MHz]{1837.5} & $B_w$ & $\numprint[MHz]{15}$\\ 
 $\alpha$ & 3.2 & $x_0$ & $\numprint[km]{-0.145}$\\ 
 $P_t G_{max}$ & \numprint[dBm]{62.75} & $y_0$ & $\numprint[km]{-0.569}$\\
 $z$ & \numprint[m]{33} & $\Tilde{a}$ & $\numprint[km^{-1}]{0.050}$\\
 $r_e$ & \numprint[m]{0} & $\Tilde{b}$ & $\numprint[km^{-2}]{5.241}$\\
 $\tau$ & \numprint[km]{7} & $\Tilde{c}$ & $\numprint[km^{-3}]{-0.973}$\\
 $\sigma^2$ & \numprint[dBm]{-96.21} & $\Tilde{d}$ & $\numprint[km^{-4}]{0.048}$\\
 $p_{g}$ & 1 &  & \\
 \hline
\end{tabular}
\end{center}
\label{tab:IPPP_study}
\end{table}
}

\subsubsection{Impact of the user location in the network}\label{para:position}
The metrics of the I-PPP network are calculated at the origin and at (\numprint[km]{-3}, \numprint[km]{-3}) in the first coordinate system. For each user location, the network area is limited to a disk of radius $\tau$ centered at the user location. The two disks are illustrated with a red and green border in Fig.~\ref{fig:IPPP_comparison_PPP_density}. Due to the potential time-consuming computation of the mathematical expressions, we suggest local approximations using a H-PPP. For each location, the network performance is further compared with two approximations: (a) a H-PPP network with a BS density equal to the average BS density inside a small disk $\epsilon$ of radius \numprint[m]{150} centered at the user location and (b) a H-PPP network with a density equal to the average density inside the $\numprint[km]{7}$-radius disk centered at the user location. The mathematical expressions for the metrics of these H-PPPs can be obtained by setting $\Tilde{a} = \numprint[km^{-1}]{0}$, $\Tilde{c} = \numprint[km^{-3}]{0}$, $\Tilde{d} = \numprint[km^{-4}]{0}$ and
\begin{equation*}
    \Tilde{b} = \frac{\int_{\epsilon} \lambda(S) dS}{\int_{\epsilon}dS}.
\end{equation*}
The marginal distribution of the EMFE, given by Corollary~\ref{eq:IPPP_PhiEmf_nobf}, is displayed in Fig.~\ref{fig:Exp_IPPP_comparison_PPP}. The solid lines correspond to the results obtained from MC simulations and the markers correspond to values obtained from the mathematical expressions. The user at the origin of the coordinate system experiences a higher EMFE because of its smaller distance to the maximal density $(x_0, y_0)$, leading to a higher number of neighboring BSs. The approximation (a) is very good and gets better as the distance from the maximal density increases, contrarily to the approximation (b). This shows that homogenizing the BS density over a very local area gives a good estimate of the EMFE experienced by the user. When considering large networks however, the homogeneity assumption does not hold anymore. The $95$th quantile is $\numprint[dBm]{-34.77}$ ($\numprint[V/m]{0.24}$) for the user at the origin and $\numprint[dBm]{-39.34}$ ($\numprint[V/m]{0.14}$) for the user at
(\numprint[m]{-3000}, \numprint[m]{-3000}). The same remarks as for Paris apply when comparing these values to the 14.57~V/m 900MHz equivalent\cite{ordonnance} threshold set by the Brussels authorities (i.e. 20.82~V/m at \numprint[MHz]{1837.5}, corresponding to 2.58~dBm using formulas in \eqref{eq:ExpWM2} and \eqref{eq:ExpVM}). The CCDF of the SINR, given by Theorem~\ref{eq:cov} and shown in Fig.~\ref{fig:Cov_IPPP_comparison_PPP}, is on the contrary almost unchanged at the different locations, with only a slight improvement for the user at the origin of the coordinate system. 

\begin{figure}[h!]
  \centering
    \includegraphics[width=0.7\linewidth, trim={3cm, 9cm, 3cm, 10cm}, clip]{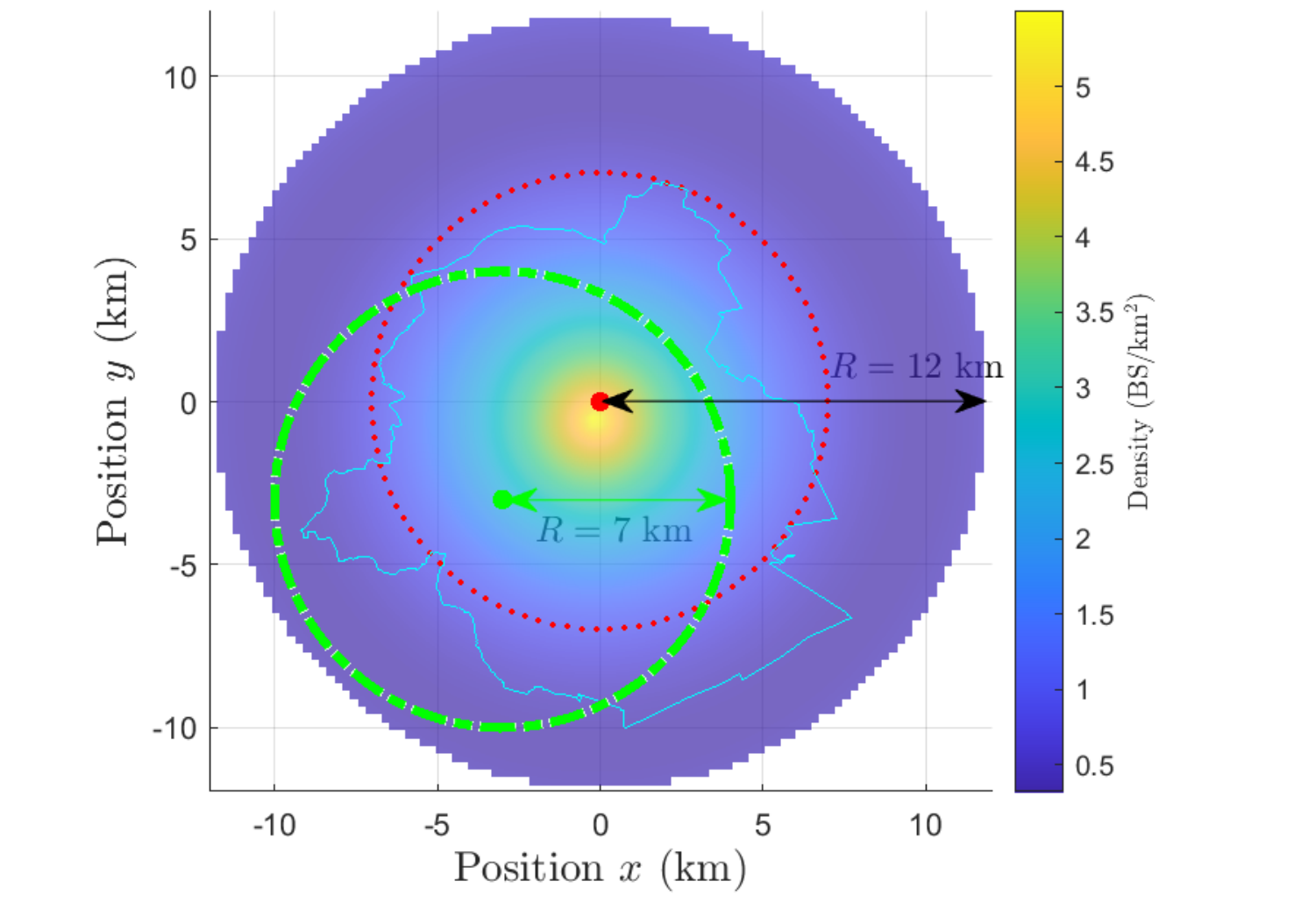}
    \caption{Density function in \eqref{eq:IPPP_density} with model parameters from Table~\ref{tab:IPPP_study}. The two disks are centered at the two calculation points used in Figs.~\ref{fig:Exp_IPPP_comparison_PPP},~\ref{fig:Cov_IPPP_comparison_PPP}. The borders of the Brussels-Capital region are in light blue.}
    \label{fig:IPPP_comparison_PPP_density}
\end{figure}
\begin{figure}[h!]
  \centering
    \includegraphics[width=0.65\linewidth, trim={3cm, 9cm, 4cm, 10cm}, clip]{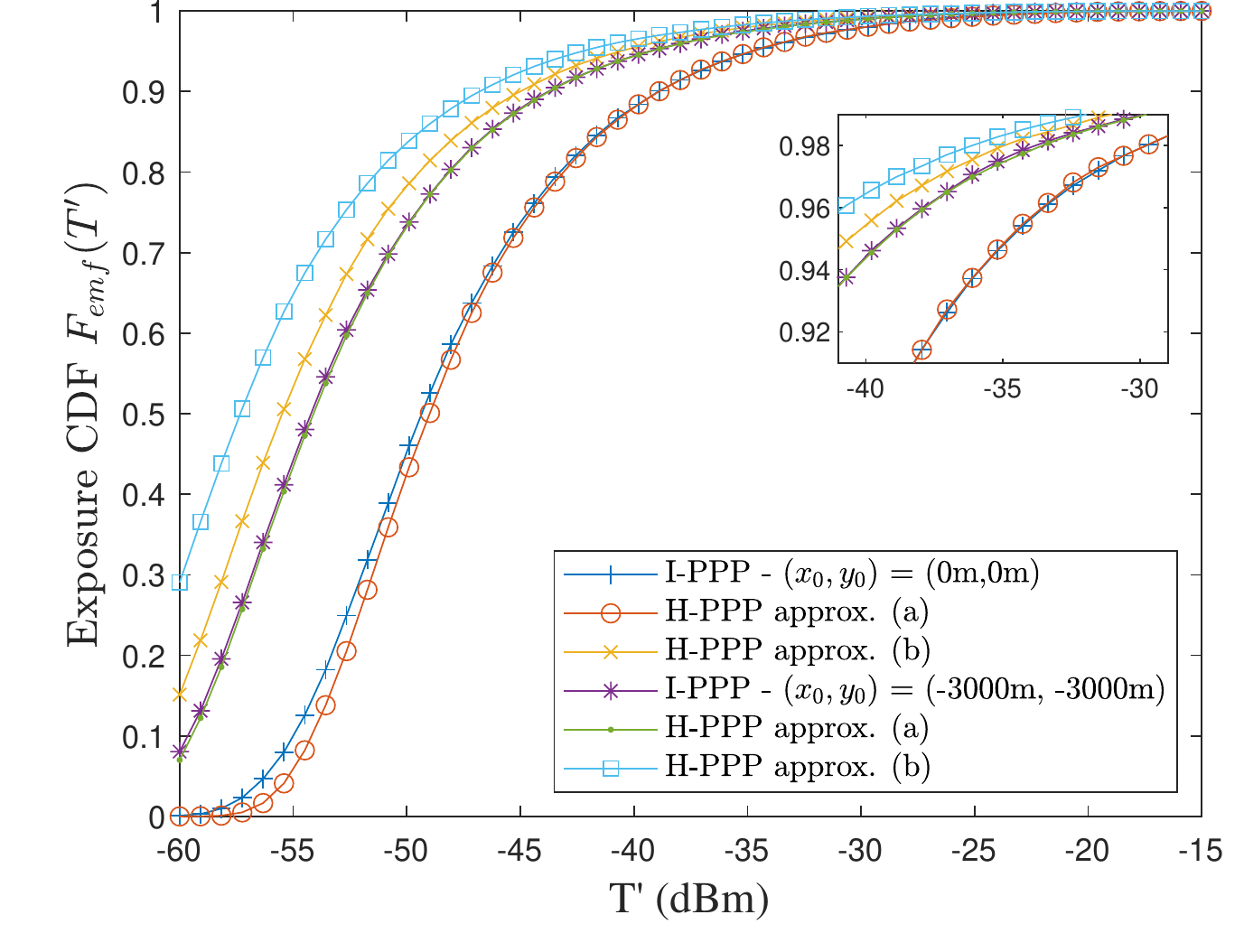}
    \caption{CDF of the EMFE at several locations in the LTE~1800 network of a major Belgian operator in Brussels and comparison with a H-PPP approximation}
    \label{fig:Exp_IPPP_comparison_PPP}
\end{figure}
At last, isocurves of the joint CDF of the EMFE and SINR, $G(T,T') = p$, given by Theorem~\ref{eq:joint} are shown in Fig.~\ref{fig:Joint_IPPP_comparison_PPP_isocurve} for $p = [0.5;0.95;0.99]$ and validated via MC simulations.

\begin{figure}[h!]
\centering
    \includegraphics[width=0.65\linewidth, trim={3cm, 9cm, 4cm, 10cm}, clip]{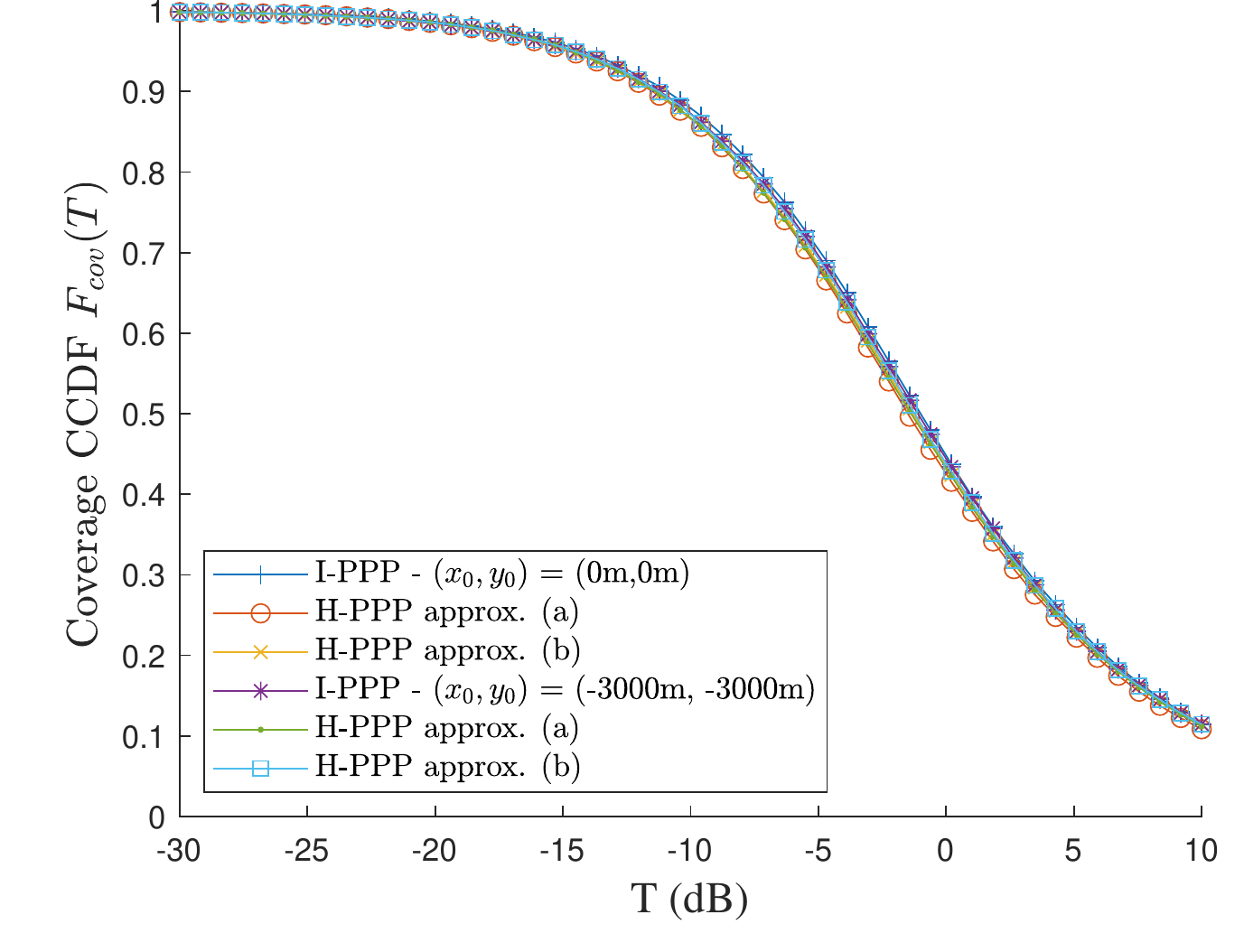}
    \caption{CCDF of the SINR at several locations in the considered LTE~1800 network and comparison with a H-PPP approximation}
    \label{fig:Cov_IPPP_comparison_PPP}
\end{figure}
\begin{figure}[h!]
    \centering
    \includegraphics[width=0.65\linewidth, trim={3cm, 9cm, 3cm, 10cm}, clip]{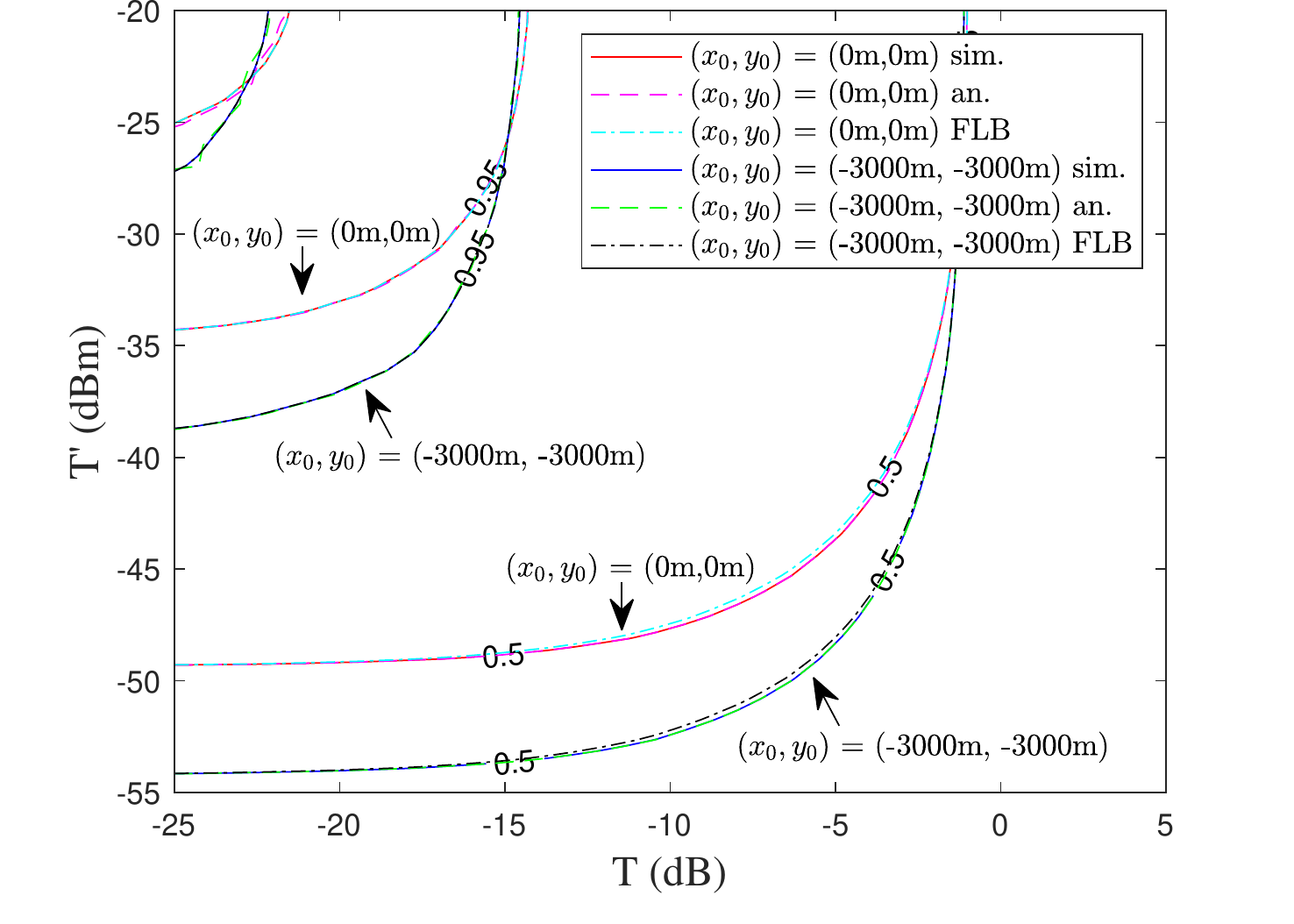}
    \caption{Isocurves of the joint CDF of the EMFE and SINR $G(T,T') = p$ with $p = [0.5;0.95;0.99]$. FLB = Fréchet lower bound.}
    \label{fig:Joint_IPPP_comparison_PPP_isocurve}
\end{figure}

\subsubsection{Overview of the metrics in the network}


In the following, the focus is put on the EMFE since it could be concluded from the previous numerical results that the spatial dependence of the coverage is relatively low. The average EMFE calculated by Theorem~\ref{eq:IPPP_meanexp} at 2500~locations in the center of Brussels is presented in Fig.~\ref{fig:IPPP_Brussels_map_LTE_1800mean_exp}, and is superimposed on the Brussels map. As expected, the EMFE radially decreases as the distance to the center increases. Statistics of the EMFE regardless of the user location are often required. Taking the expected value of the considered network over the two dimensions of space in Theorem~\ref{eq:IPPP_meanexp} gives a value of $\numprint[W/m^2]{3.50e-5}$ ($\numprint[V/m]{0.11}$).

\begin{figure}[h!]
    \centering
    \includegraphics[width=0.65\linewidth, trim={3cm, 9cm, 3cm, 10cm}, clip]{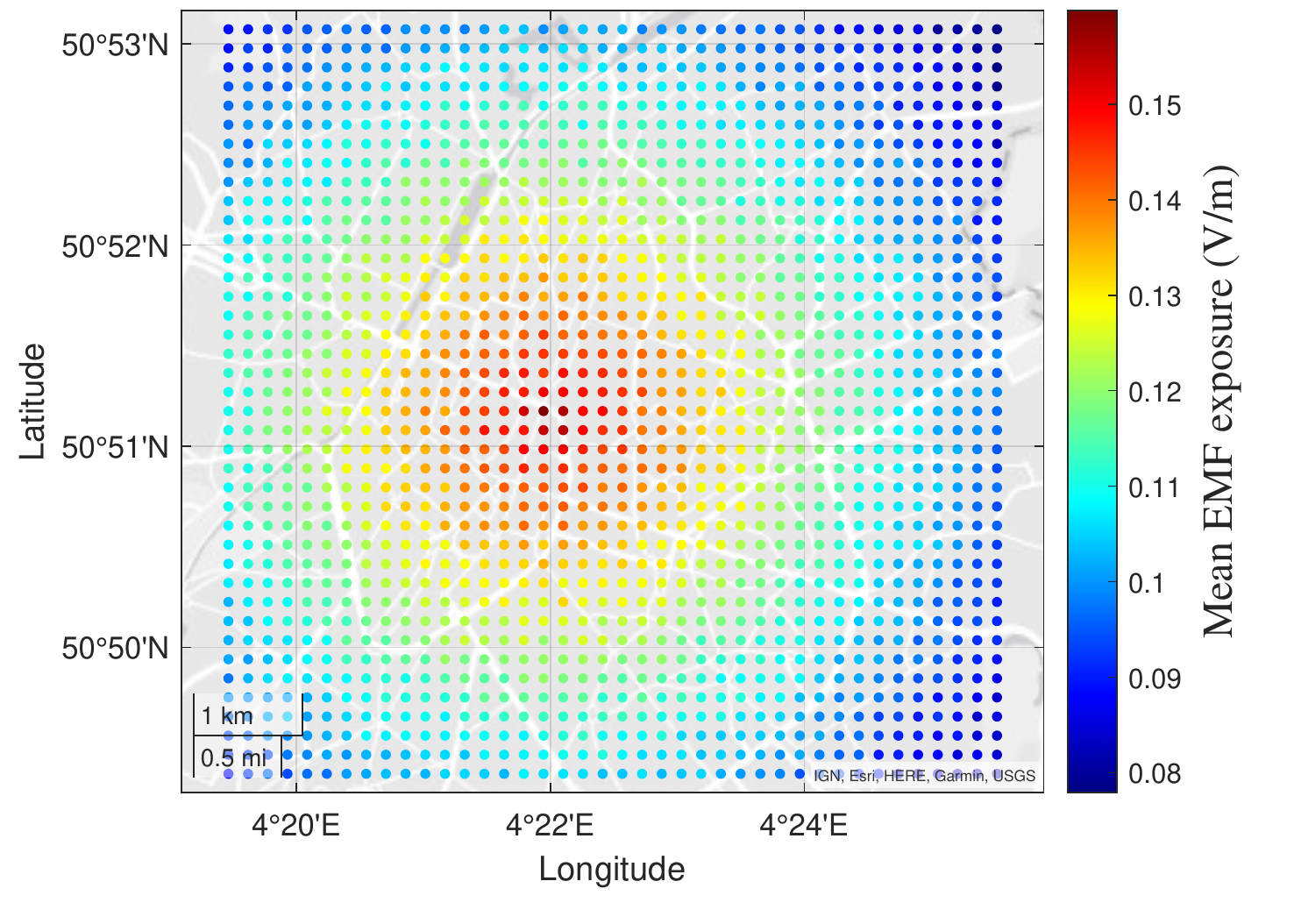}
    \caption{Mean EMFE from the considered LTE~1800 network at 2500 locations in the center of Brussels}
    \label{fig:IPPP_Brussels_map_LTE_1800mean_exp}
\end{figure}
The probability to be below the threshold $T'$ = -35.7~dBm (0.22~V/m) is shown at the same 2500 locations in Fig.~\ref{fig:IPPP_Brussels_map_LTE_1800expcdf}. The relatively low values of the mean EMFE compared to the legal threshold leads to choose a relatively low value of $T'$ for the illustration. Again, taking the expected value of the spatially-dependent CDF of the EMFE in Theorem~\ref{eq:exp_IPPP} provides a CDF of the EMFE experienced by any user, regardless of its location, as shown in Fig.~\ref{fig:IPPP_Brussels_map_LTE_1800expcdf_mean}.

\begin{figure}[h!]
\centering
    \includegraphics[width=0.65\linewidth, trim={3cm, 9cm, 3cm, 10cm}, clip]{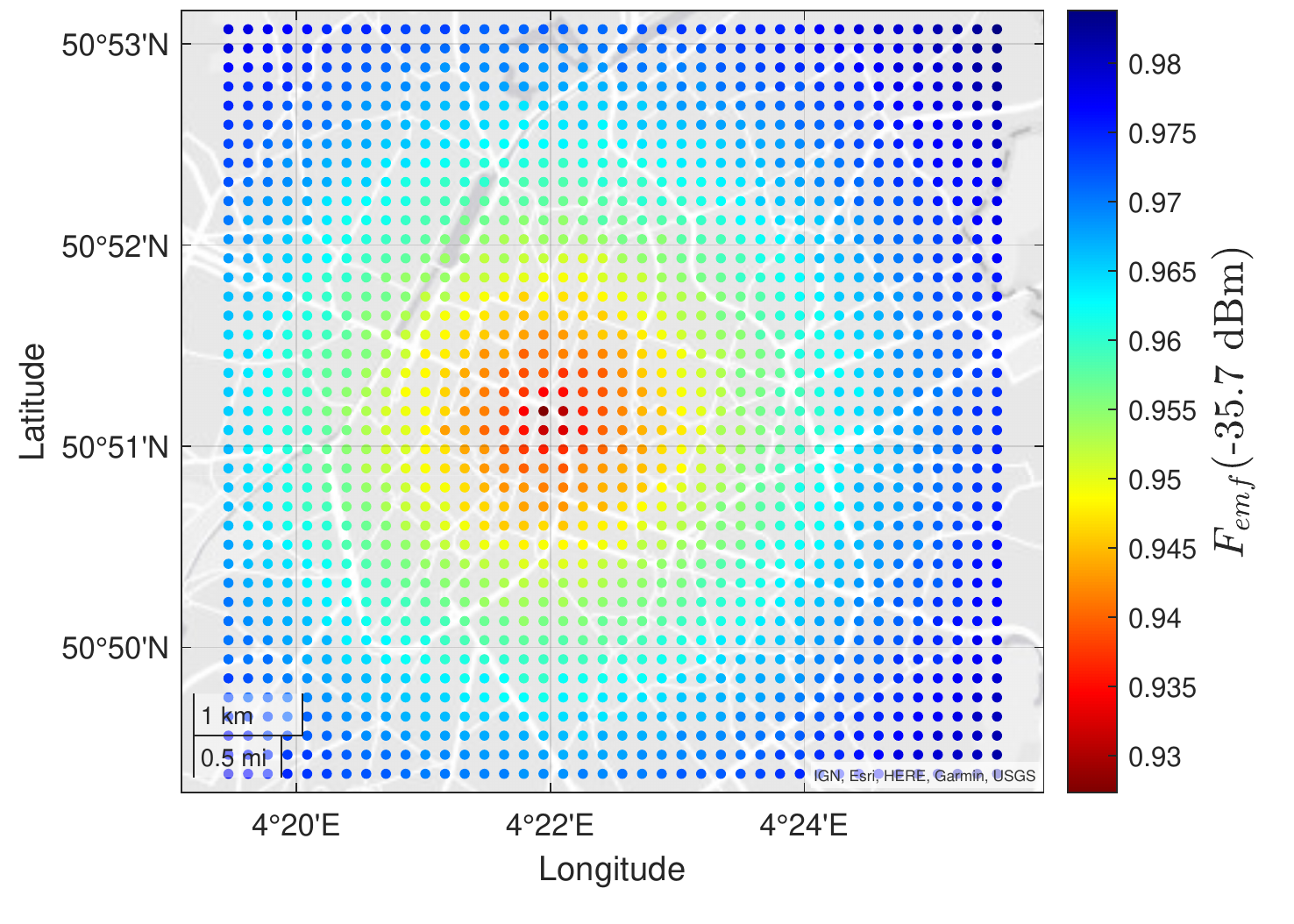}
    \caption{$F_{\text{emf}}$(-35.7~dBm) for the considered LTE~1800 network at 2500 locations in the center of Brussels}
    \label{fig:IPPP_Brussels_map_LTE_1800expcdf}
\end{figure}
\begin{figure}[h!]
  \centering
    \includegraphics[width=0.6\linewidth, trim={3cm, 9cm, 3cm, 10cm}, clip]{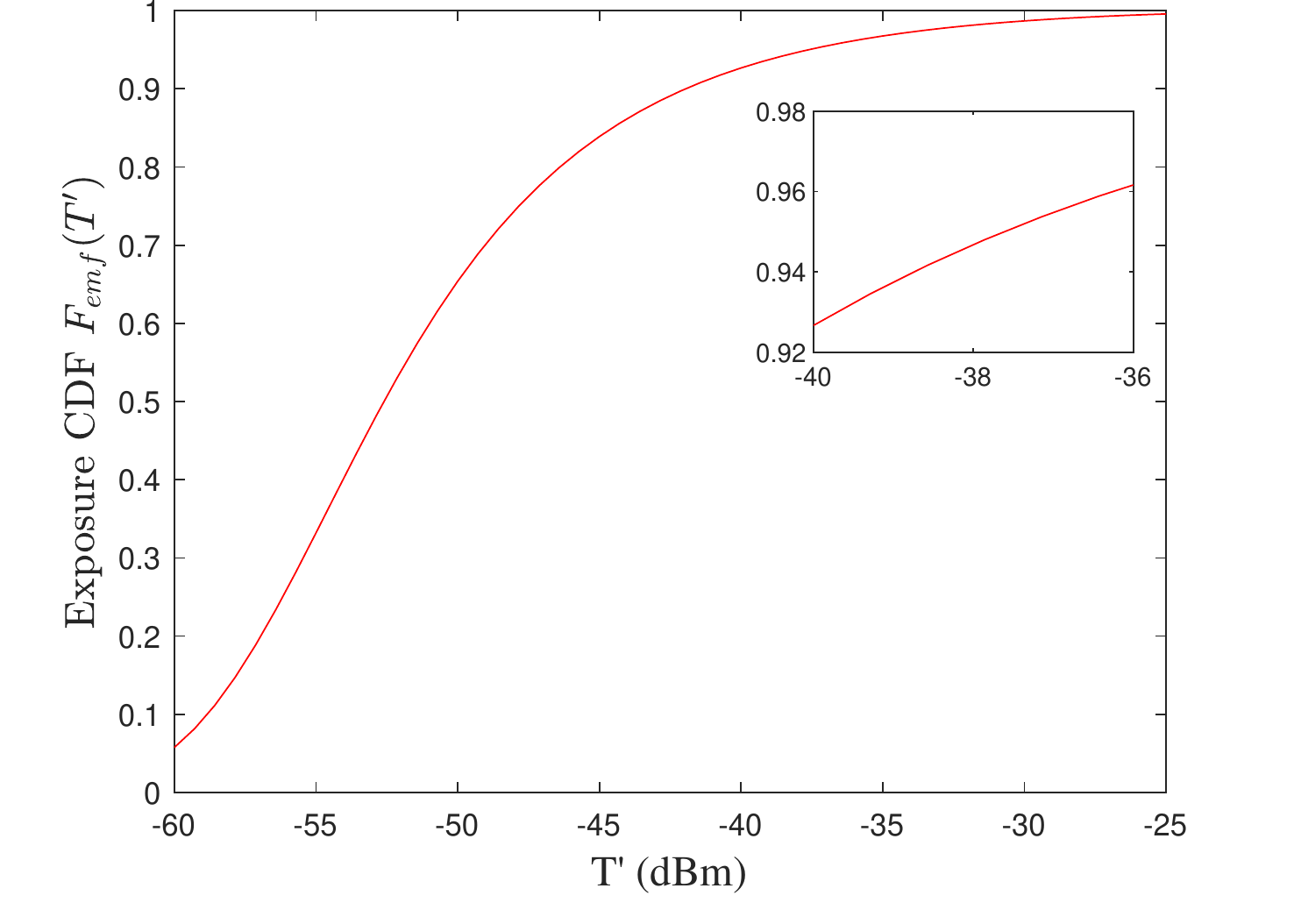}
    \caption{$F_{emf}(T')$ averaged over all 2500 locations in the considered LTE~1800 network.}
    \label{fig:IPPP_Brussels_map_LTE_1800expcdf_mean}
\end{figure}

\section{Conclusion}
\label{sec:conclusion}
In this paper, performance metrics are provided to jointly analyze the EMFE and coverage in MI and MV cellular networks. Based on network topologies usually encountered in European cities, the $\beta$-GPP is used as an example of tractable PP for MI networks. In MV networks, the network is modeled by an I-PPP. The approach is illustrated for the 5G~NR~2100 network in Paris (France) and a LTE~1800 network in Brussels (Belgium) using realistic network parameters. An analysis of the network parameters is also provided and shows that an optimal value of these parameters can be found to maximize the coverage while minimizing the EMFE.
For future works, performance metrics could be easily adapted for other network models, in other cities, using additional features (more complex BF model, blockage, uplink, heterogeneous cellular networks). The adaptation of the proposed framework to evaluate global EMFE, incorporating the correlation between the location of BSs of different technologies and operators, in order to have a faithful way of comparison with the legal thresholds, is another potential avenue for further investigation.


{\appendices

\section{Proof of the Mean of the EMFE in a $\beta$-GPP}\label{sec:BGPP_mean_proof}

As the main beam of the serving BS is assumed to be directed towards the user, while the beams of the interfering BSs are randomly oriented, it is essential to separate the EMFE caused by the serving BS and the EMFE caused by the interfering BSs. As explained in Subsection~\ref{ssec:motion_invariant_model}, since each BS can potentially be the serving BS, we have{\small
\begin{equation}\label{eq:mean_IPD_1}
    \mathbb E\left[\mathcal{P}\right] \!=\! \mathbb E_{\Psi, \xi, |h|, G}\!\Bigg[\sum_{i \in \mathbb{N}} \!\Big(\Bar{P}_{r,i} G_i |h_i|^2 + \sum\limits_{\substack{j \in \mathbb{N} \\ j\neq i}} \Bar{P}_{r,j} G_j |h_j|^2 \xi_j \Big) \mathds{1}_{\xi_i} \mathds 1_{\mathcal{A}_i}\Bigg]
\end{equation}}where $\mathds{1}_{\mathcal{A}_i}\! = \! \{\mathcal{A}_i \!=\! 1\}$ is the indicator function.
The expectation operators and the sum over $i$ can be interchanged. Moreover, $\Psi$, $\xi$, $|h|$, and $G$ are pairwise independent. Given the model, the normalized gain of the serving BS is $G_i \!= \!1$, and $\mathbb E\left[|h|^2\right] \!= \!1$. With these observations and the application of the expectation operators $\mathbb E_{\xi}$ and $\mathbb E_{G_i}$, \eqref{eq:mean_IPD_1} is expressed as
\begin{equation*}
    \mathbb E\left[\mathcal{P}\right] = \beta \sum_{i \in \mathbb{N}} \mathbb E_{\Psi, \xi}\Bigg[\Bigg(\Bar{P}_{r,i} + \mathbb E_{G}\Big[\sum\limits_{j \in \mathbb{N} \setminus \{i\}} \Bar{P}_{r,j} G_j \xi_j\Big] \Bigg)\mathds 1_{\mathcal{A}_i}\Bigg].
\end{equation*}
The expectation operator $\mathds E_\Psi$ can be split into the expectation operators $\mathds E_{|X_i|}\!=\!\mathds E_{R_i}$ and $\mathds E_{\Psi \setminus \{\!X_i\!\}}$. We make a small approximation by applying the expectation operators $\mathds E_{\Psi \setminus \{\!X_i\!\}}$ and $\mathds E_\xi$ to the sum of the interference and to $\mathds 1_{\mathcal{A}_i}$ independently. The impact of this approximation is shown to be insignificant in the numerical results section. Given that $\mathbb E\left[\xi_i\right]\! = \!\beta$, this leads~to
\begin{multline}\label{eq:mean_IPD_2}
    \mathbb E\left[\mathcal{P}\right] = \beta \sum_{i \in \mathbb{N}} \mathbb E_{R_i}\Bigg[\Bigg(\Bar{P}_{r,i} + \beta \mathbb E_{G, \Psi \setminus \{\!X_i\!\}}\Big[\sum\limits_{j \in \mathbb{N} \setminus \{i\}} \Bar{P}_{r,j} G_j\Big] \Bigg)\Bigg.\\
    \Bigg.\times \underbrace{\mathbb E_{\Psi \setminus \{\!X_i\!\}}\Big[\prod_{j \in \mathbb{N} \setminus \{i\}}\mathbb E_{\xi_j}\left[\mathds{1}_{\left\{\xi_j=1, |X_j| > |X_i|\right\} \cup \{\xi_j = 0\}}\right]\Big]}_{\Upsilon_i^\beta(R_i^2)}\Bigg].
\end{multline}
Focusing first on $\Upsilon_i^\beta(R_i^2)$, upon applying the expectation operators, the resulting expression is obtained:
\begin{equation}\label{eq:mean_IPD_3}
    \Upsilon_i^\beta(R_i^2) = \prod_{j \in \mathbb{N} \setminus \{i\}}\left(\beta \,\mathbb E_{R_j}\left[\mathds{1}_{\left\{R_j^2 > R_i^2\right\}}\right]+1-\beta\right).
\end{equation}
Then using the PDF in \eqref{eq:BGPP_pdf} for the inner expectation results in
\begin{equation}\label{eq:E_Xi}
    \mathbb E_{R_j}\left[\mathds{1}_{\left\{R_j^2 > R_i^2\right\}}\right] = \int_{r_e^2}^{\tau^2} \mathds{1}_{\left\{v > R_i^2\right\}} \frac{v^{j-1}\, e^{-\frac{cv}{\beta}}}{(j-1)!\, \left(\beta/c\right)^j}\,dv.
\end{equation}
Using the generalized incomplete gamma function $\Gamma(z; a, b) = \int_a^b  t^{j-1}\, e^{-t}\,dt$ and replacing \eqref{eq:E_Xi} in \eqref{eq:mean_IPD_3} leads to
\begin{equation}\label{eq:upsilon}
    \Upsilon_i^\beta(R_i^2) = \prod_{j \in \mathbb{N} \setminus \{i\}}\Bigg(1-\beta+\beta \,\frac{\Gamma\left(j, \frac{c R_i^2}{\beta}, \frac{c\tau^2}{\beta}\right)}{(j-1)!}\Bigg).
\end{equation}

Then in \eqref{eq:mean_IPD_2}, the sum over $j$, excluding the index $i$, can be rewritten as the difference between the sum over $j$ and the term with index $i$. As $\mathbb E_G\left[G\right] = p_g$ and as the $G_i$'s are pairwise independent, by using again the PDF in \eqref{eq:BGPP_pdf}, we have
\begin{equation}\label{eq:mean_IPD_4}
\begin{split}
    &\mathbb E_{G, \Psi \setminus \{\!X_i\!\}}\Big[\sum\limits_{j \in \mathbb{N} \setminus \{i\}} \Bar{P}_{r,j} G_j\Big] = \mathbb E_{G, \Psi}\Big[\sum\limits_{j \in \mathbb{N}} \Bar{P}_{r,j} G_j-\Bar{P}_{r,i} G_i\Big]\\
    &= p_g \int_{R_i^2}^{\tau^2} \Bigg(\sum_{j \in \mathbb{N}}  \frac{(\frac{c}{\beta})^j v^{j-1}e^{-\frac{c u}{\beta}}}{(j-1)!}-\frac{(\frac{c}{\beta})^i v^{i-1}e^{-\frac{c u}{\beta}}}{(i-1)!}\Bigg) \Bar{P}_{r}(v) dv.
\end{split}
\end{equation}
The sum inside the integral converges and gives
\begin{equation}\label{eq:sum_converges}
    \sum_{j \in \mathbb{N}} \frac{v^{j-1}}{(\beta/c)^{j}\,(j-1)!}
    = \sum_{j \in \mathbb{N}_0} \frac{v^{j}}{(\beta/c)^{j+1}\,j!} 
    = \frac{c}{\beta}\, e^{vc/\beta}.
\end{equation}
Replacing \eqref{eq:sum_converges} in \eqref{eq:mean_IPD_4} gives
\begin{equation}\label{eq:mean_IPD_6}
\begin{split}
    &\mathbb E_{G, \Psi \setminus \{\!X_i\!\}}\Big[\sum\limits_{j \in \mathbb{N} \setminus \{i\}} \Bar{P}_{r,j} G_j\Big]\\
    &= \frac{2 p_g c/\beta}{\alpha-2} \left[\Bar{P}_{r}(r^2)\left(r^2+z^2\right)\right]_{r = \tau}^{r = R_i}-p_g \int_{R_i^2}^{\tau^2} f_i(v) \Bar{P}_{r}(v) dv.
\end{split}
\end{equation}
Replacing \eqref{eq:mean_IPD_3} and \eqref{eq:mean_IPD_6} in \eqref{eq:mean_IPD_2} and applying the expectation operator $\mathbb E_{R_i}$ using \eqref{eq:BGPP_pdf} leads to 
{\small
\begin{multline*}
    \mathbb E\left[\mathcal{P}\right] \!=\! \beta \sum_{i \in \mathbb{N}} \int_{r_e^2}^{\tau^2}\!\left(\Bar{P}_{r}(u) + \frac{2 \, c \,p_g}{\alpha-2} \!\!\left[\Bar{P}_{r}(r^2)\left(r^2+z^2\right)\right]_{r = \tau}^{r = \sqrt{u}}\right.\\
    \left.-{p_g}{\beta} \int_{u}^{\tau^2} f_i(v) \Bar{P}_{r}(v) dv\right)\times f_i(u) \Upsilon_i^\beta(u) du.
\end{multline*}}At last, the expression of the mean~\eqref{eq:BGPP_meanexpeq} in Theorem~\ref{eq:BGPP_meanexp} is obtained by switching the sum and the integral, by distributing the product and by defining 
\begin{equation*}
    \Omega(u) = \sum_{i \in \mathbb{N}} f_i(u) \Upsilon_i^\beta(u), \quad \Omega^*(u,v) = \sum_{i \in \mathbb{N}} f_i(u) f_i(v) \Upsilon_i^\beta(u).
\end{equation*}

If no dynamic BF is employed, \eqref{eq:mean_IPD_1} simplifies to $\mathbb E\left[\mathcal{P^*}\right] = \mathbb E_{\Psi,|h|,\xi}\left[\sum_{i \in \mathbb{N}} \Bar{P}_{r,i} |h_i|^2 \, \xi_i\right]$. Applying the expectation operators leads to
\begin{equation*}
    \mathbb E\left[\mathcal{P^*}\right] = \beta \int_{r_e^2}^{\tau^2} \sum_{i \in \mathbb{N}}  \frac{(c/\beta)^i}{(i-1)!}u^{i-1}e^{-cu/\beta} \Bar{P}_{r}(u) du.
\end{equation*}
Repeating step \eqref{eq:sum_converges} and integrating leads to \eqref{eq:BGPP_meanexpeq} in Corollary~\ref{eq:BGPP_meanexp_nobf}.
\section{Proof of the Second Moment of the EMFE in a $\beta$-GPP}\label{sec:BGPP_var_proof}
Since every BS can potentially be the serving BS, the same reasoning used for the proof of the mean can be applied to derive the second moment of the EMFE. By using the result \eqref{eq:upsilon}, this gives
\begin{equation}\label{eq:var_BGPP_1}
    \begin{split}
        \mathbb E\left[\mathcal{P}^2\right] &= \mathbb E\left[(S_0+I_0)^2\right] = \mathbb E\left[S_0^2+2 S_0 I_0 + I_0^2\right]\\
        &= \mathbb E\left[\sum_{i \in \mathbb{N}}\left(S_i^2+2 S_i I_i + I_i^2\right) \mathbb P\left[X_0 = X_i\right]\right]\\
        &= \sum\limits_{i \in \mathbb{N}} \mathbb E_{R_i}\!\left[\mathbb E_{\Psi\!\setminus \!\{\!X_i\!\}, G, |h|, \xi}\!\left[S_i^2+2 S_i I_i + I_i^2\right]\!\Upsilon_i^\beta(X_i)\right]
    \end{split}
\end{equation}
We then analyze the the inner expectation operator on each term, separatly. Starting with $\mathbb E[S_i^2]$ and knowing that $\mathbb E_{|h|}\left[|h|^4\right] = {(m+1)}/{m}$, we have
\begin{equation}\label{eq:var_BGPP_2}
    \mathbb E_{\Psi\!\setminus \!\{\!X_i\!\}, G, |h|, \xi}\!\left[S_i^2\right] = \beta \frac{m+1}{m} \Bar{P}_{r,i}^2.
\end{equation}
Now for the cross-term $\mathbb E[S_i I_i]$, since the $|h_i|$'s are pairwise independent, just like the $\xi_i$'s and the $G_i$'s, we obtain
\begin{equation*}
    \begin{split}
        \mathbb E\!\left[S_i I_i\right] &= \mathbb E_{\Psi\!\setminus \!\{\!X_i\!\}, G, |h|, \xi}\!\Bigg[\Bar{P}_{r,i}|h_i|^2 G_i \xi_i \sum\limits_{j \in \mathbb{N} \setminus \{i\}}\Bar{P}_{r,j}|h_j|^2 G_j \xi_j\Bigg] \\
        &= \beta^2 \Bar{P}_{r,i}\mathbb E_{G, \Psi \setminus \{\!X_i\!\}}\Big[\sum\limits_{j \in \mathbb{N} \setminus \{i\}}\Bar{P}_{r,j} G_j\Big]
    \end{split}
\end{equation*}
Similarly to what is done in \eqref{eq:mean_IPD_6} for the mean, this gives
\begin{equation}\label{eq:var_BGPP_3}
\begin{split}
    \mathbb E\!\left[S_i I_i\right] = \beta p_g \Bar{P}_{r,i}\left( \frac{2 c}{\alpha-2} \left[\Bar{P}_{r}(r^2)\left(r^2+z^2\right)\right]_{r = \tau}^{r = R_i}\right.\\
    \left.-\beta\int_{R_i^2}^{\tau^2} f_i(v) \Bar{P}_{r}(v) dv\right).
\end{split}
\end{equation}
At last, $\mathbb E[I_i^2]$ can be split into two terms:
{\small
\begin{multline*}
   \mathbb E[I_i^2] = \overbrace{\mathbb E_{\Psi, G, |h|, \xi}\Big[\sum_{j \in \mathbb{N} \setminus \{i\}} \Bar{P}_{r,j}^2 G_j^2 |h_j|^4\xi_j^2 \Big]}^{T_1} \\
   +  \underbrace{\mathbb E_{\Psi, G,  |h|, \xi}\Bigg[\sum\limits_{j \in \mathbb{N} \setminus \{i\}}\sum\limits_{k \in \mathbb{N} \setminus \{i,j\}} \!\left(\Bar{P}_{r,j} G_j |h_j|^2\xi_j\right)\!\left(\Bar{P}_{r,k} G_k |h_k|^2\xi_k\right)\!\Bigg]}_{T_2}
\end{multline*}}

By following the same reasoning as before and by using $\mathbb E_G[G^2] = p_g$ for the interfering BSs, $T_1$ is obtained by
\begin{equation}\label{eq:var_BGPP_4}
\begin{split}
    T_1 &= \beta \frac{m+1}{m} p_g\! \sum_{j \in \mathbb{N} \setminus \{i\}} \int_{R_i^2}^{\tau^2}\! f_j(v) \Bar{P}_{r}^2(v) dv\\
    &= \frac{m+1}{m} \frac{p_g c}{\alpha-1} \left[\Bar{P}_{r}^2(r^2)\left(r^2+z^2\right)\right]_{r = \tau}^{r = R_i}\\
    &\qquad\qquad\qquad\qquad-\frac{m+1}{m} p_g\beta \int_{R_i^2}^{\tau^2} f_i(v) \Bar{P}_{r}^2(v) dv
\end{split}
\end{equation}
and $T_2$ is obtained by
{\footnotesizetosmall
\begin{equation*}
    T_2 = \beta^2 p_g^2\!\sum\limits_{j \in \mathbb{N} \setminus \{i\}}\! \int_{R_i^2}^{\tau^2}\!f_j(v) \Bar{P}_{r}(v) \! \sum\limits_{k \in \mathbb{N} \setminus \{i,j\}}\! \int_{R_i^2}^{\tau^2}\!f_j(w) \Bar{P}_{r}(w) dv\, dw.
\end{equation*}}This can be rewritten as
{\footnotesizetosmall
\begin{equation}\label{eq:var_BGPP_5}
\begin{split}
    T_2 = \underbrace{\beta^2 p_g^2 \Big(\sum_{j \in \mathbb{N} \setminus \{i\}} \int_{R_i^2}^{\tau^2}f_j(v) \Bar{P}_{r}(v) dv\Big)^2}_{T_{2a}}\\
    -\underbrace{\beta^2 p_g^2 \sum_{j \in \mathbb{N} \setminus \{i\}} \left(\int_{R_i^2}^{\tau^2}f_j(v) \Bar{P}_{r}(v) dv\right)^2}_{T_{2b}}
\end{split}
\end{equation}}where $T_{2a}$ gives
{\footnotesizetosmall
\begin{equation}\label{eq:var_BGPP_6}
\begin{split}
    T_{2a} &= p_g^2\left(\frac{2 c}{\alpha-2} \left[\Bar{P}_{r}(r^2)\left(r^2+z^2\right)\right]_{r = \tau}^{r = R_i}-\beta \!\int_{R_i^2}^{\tau^2}\! f_i(v) \Bar{P}_{r}(v) dv.\right)^2\\
    &= \left(\frac{2 p_g c }{\alpha-2} \left[\Bar{P}_{r}(r^2)\left(r^2+z^2\right)\right]_{r = \tau}^{r = R_i}\right)^2\\
    &\qquad - \frac{4 p_g^2 c \beta}{\alpha-2} \left[\Bar{P}_{r}(r^2)\left(r^2+z^2\right)\right]_{r = \tau}^{r = R_i} \!\int_{R_i^2}^{\tau^2}\! f_i(v) \Bar{P}_{r}(v) dv\\
    &\qquad + \beta^2 p_g^2 \!\int_{R_i^2}^{\tau^2}\!\int_{R_i^2}^{\tau^2}\! f_i(v)f_i(w) \Bar{P}_{r}(v)\Bar{P}_{r}(w) dw \,dv
\end{split}
\end{equation}}and $T_{2b}$ can be developed as follows:
{\footnotesizetosmall
\begin{equation*}
\begin{split}
        T_{2b} = p_g^2\beta^2\! \sum_{j \in \mathbb{N}} \int_{R_i^2}^{\tau^2}\!\int_{R_i^2}^{\tau^2}\!\frac{(vw)^{j-1}e^{-\frac{c(v+w)}{\beta}}(\frac{c}{\beta})^{2j}}{(j-1)!(j-1)!}\Bar{P}_{r}(v)\Bar{P}_{r}(w) dw\,dv\\
        - p_g^2\beta^2\! \int_{R_i^2}^{\tau^2}\!\int_{R_i^2}^{\tau^2}\!f_i(v)f_i(w)\Bar{P}_{r}(v)\Bar{P}_{r}(w) dw\,dv
\end{split}
\end{equation*}}
The integrals and the sum in the first line can be switched. The sum over $j$ converges and gives a modified Bessel function of the first kind $I_0(x)$ of order 0:
\begin{multline}\label{eq:var_BGPP_7}
         \kern-.9em T_{2b} = p_g^2 c^2 \!\int_{R_i^2}^{\tau^2}\!\int_{R_i^2}^{\tau^2}\!e^{-\frac{c(v+w)}{\beta}} I_0\left(2c\sqrt{vw}/\beta\right)\Bar{P}_{r}(v)\Bar{P}_{r}(w) dw\,dv\\
        - p_g^2\beta^2\! \int_{R_i^2}^{\tau^2}\!\int_{R_i^2}^{\tau^2}\!f_i(v)f_i(w)\Bar{P}_{r}(v)\Bar{P}_{r}(w) dw\,dv
\end{multline}

The second moment in Theorem~\ref{eq:BGPP_meanexp} is given by combining \eqref{eq:var_BGPP_1}, \eqref{eq:var_BGPP_2}, \eqref{eq:var_BGPP_3}, \eqref{eq:var_BGPP_4}, \eqref{eq:var_BGPP_5}, \eqref{eq:var_BGPP_6} and \eqref{eq:var_BGPP_7} and by reorganizing the terms.


If no BF is employed, the reasoning is exactly the same as the one to compute $\mathbb E\left[I_i^2\right]$, except that there is no $j \neq i$ and $k \neq i$ conditions. There is therefore no summation on the index $i$ and the only terms that remain are the first term of $T_1$ in \eqref{eq:var_BGPP_4}, the first term of $T_{2a}$ in \eqref{eq:var_BGPP_6} and the first term of $T_{2b}$ in \eqref{eq:var_BGPP_7}, with the lower bound of the integrals being $r_e$ instead of $R_i$.


\section{Proof of the Characteristic Function of the Interference in a $\beta$-GPP}\label{sec:BGPP_PhiI_Proof}
From the definition of CF,
{\smalltonormalsize
\begin{equation*}
    \begin{split}
    \Phi_{I,i}(q|R_i^2) &= \mathbb E\left[\exp\Big(j q \sum\limits_{k \in \mathbb{N} \setminus \{i\}} \Bar{P}_{r,k} |h_k|^2 G_k \xi_k \Big| |X_k| > |X_i|\Big) \right]\\
    &= \mathbb E\Big[\prod\limits_{k \in \mathbb{N} \setminus \{i\}} e^{j q \Bar{P}_{r,k} |h_k|^2 G_k \xi_k}\, \mathds{1}_{\left\{|X_k| > |X_i|\right\}}\Big].
    \end{split}
\end{equation*}
}
Since the random variables $\Psi$, $h$ and $\xi$ are independent,
\begin{equation}\label{eq:Phi_I_BGPP_step}
    \Phi_{I,i}(q|R_i^2)\! = \!\prod\limits_{k \in \mathbb{N} \setminus \{i\}}\!\mathbb  E_{\Psi}\Bigg[\underbrace{\mathbb{E}_{\xi,|h|,G}\!\left[e^{j q \Bar{P}_{r,k} |h_k| G_k \xi_k}\right]}_{Q(q)}\! \mathds{1}_{\left\{|X_k| > |X_i|\right\}}\Bigg].
\end{equation}
Extracting the term $Q(q)$ and applying the expectations on $\xi$ and $G$ leads to
\begin{equation*}
    Q(q) = \mathbb E_{|h|}\left[p_g\beta e^{j q \Bar{P}_{r,k} |h_k|^2}+1-p_g\beta\right].
\end{equation*}

Then for a Nakagami-$m$ fading it becomes
\begin{equation*}
     Q(q)  =  \frac{p_g\beta}{\left(1-j q \Bar{P}_{r,k}/m\right)^m}+1-p_g\beta.
\end{equation*}

Replacing this expression in~\eqref{eq:Phi_I_BGPP_step} and applying the expectation on $\Phi$ using \eqref{eq:BGPP_pdf} results in Proposition~\ref{eq:BGPP_PhiI}.

\section{Proof of the CDF of the EMFE in a $\beta$-GPP}\label{BGPP_emf_proof}
Since every BS can be the serving BS, we obtain
\begin{equation}\label{eq:BGPP_emfsum}
\begin{split}
    F_{\text{emf}}(T') &= \sum_{i \in \mathbb{N}} \mathbb P\left(\mathcal{P} < T', X_0 = X_i\right)\\
    &= \beta \sum_{i \in \mathbb{N}} \underbrace{\mathbb P\left[S_i+I_i < T', \mathcal{A}_i\right]}_{F_i(T')}
\end{split}
\end{equation}Using Gil-Pelaez' theorem, $F_i(T')$ can be further expressed~as
\begin{equation*}
    F_i(T')\! = \!\mathbb E_{\Psi,\xi}\left[\mathds{1}_{\mathcal{A}_i}\!\left(\frac{1}{2}-\int_0^\infty  \frac{1}{\pi q}\text{\normalfont Im} \left[\phi_{E,i}(q|R_i^2)e^{-j q T'}\right]\!dq\right)\right]
\end{equation*}
where $\phi_{E,i}(q|R_i^2) = \Phi_S(q|R_i^2)\, \Phi_{I,i}(q|R_i^2)$ and
\begin{equation*}
    \Phi_S(q|R_i^2) = \mathbb{E}_{|h|}\left[\exp\left(j q S_i\right)\right] = (1-jq \Bar{P}_{r,i}/m)^{-m}.
\end{equation*}

Using the same reasoning as in Appendix~\ref{sec:BGPP_mean_proof} to obtain \eqref{eq:mean_IPD_2} by using the definition \eqref{eq:upsilon} of $\Upsilon_i^\beta$, we have
{\small
\begin{equation*}
    F_i(T')\! = \!\mathbb E_{R_i}\!\left[\Upsilon_i^\beta(R_i^2)\!\left(\frac{1}{2}-\int_0^\infty \! \frac{1}{\pi q}\text{\normalfont Im} \left[\phi_{E,i}(q|R_i^2)e^{-j q T'}\right]\!dq\right)\right]\!.
\end{equation*}}
Applying the expectation operator and using \eqref{eq:BGPP_pdf} are applied results in
{\footnotesize
\begin{equation}\label{eq:BGPP_emfT2}
    F_i(T')\! = \!\int_{r_e^2}^{\tau^2}\!f_i(u) \Upsilon_i^\beta(u)\!\left(\frac{1}{2}\!-\!\int_0^\infty \! \frac{1}{\pi q}\text{\normalfont Im} \left[\phi_{E,i}(q|u)e^{-j q T'}\right]\!dq\right)\!du.
\end{equation}}
Theorem~\ref{eq:exp_BGPP} is obtained by replacing \eqref {eq:BGPP_emfT2} in \eqref{eq:BGPP_emfsum}.

\section{Proof of the Mean EMF Exposure in an I-PPP}\label{sec:IPPP_mean_proof}
As the main beam of the serving BS is assumed to be directed towards the user, while the beams of the interfering BSs are randomly oriented, it is essential to separate the EMFE caused by the serving BS and the EMFE caused by the interfering BSs. The mean EMFE can then be written
\begin{equation*}
    \mathbb E\left[\mathcal{P}\right] \!=\! \mathbb E_{\Psi,|h|, G}\!\Bigg[\Bar{P}_{r,0} G_0 |h_0|^2 + \sum\limits_{i \in \mathbb{N}} \Bar{P}_{r,i} G_i |h_i|^2\Bigg].
\end{equation*}
Using the independence between $G_i$, $h_i$ and $\Psi$, given that $G_0 = 1$, $\mathbb E_G[G_{i|i\neq 0}] = p_g$ and $\mathbb E_{|h|}[|h_i|^2] = 1$ and using $\mathbb E_{\Psi}[\cdot] = \mathbb E_{X_0}[\cdot] + \mathbb E_{\Psi \setminus \{\!X_0\!\}}[\cdot]$ , the mean EMFE can be further expressed~as
\begin{equation*}
    \mathbb E\left[\mathcal{P}\right] \!=\! \mathbb E_{X_0}\!\left[\Bar{P}_{r,0} + p_g \mathbb E_{\Psi \setminus \{\!X_0\!\}}\Big[\sum\limits_{i \in \mathbb{N}} \Bar{P}_{r,i}\Big]\right].
\end{equation*}
It follows from Campbell's theorem~that
\begin{equation*}
    \mathbb E\left[\mathcal{P}\right] = \mathbb E_{X_0}\!\left[\Bar{P}_{r,0} + p_g \int_{R_0}^{\tau} \Bar{P}_{r}(r)\,\Lambda^{\!(1)\!}(r)\,dr\right].
\end{equation*}
The mean EMFE in Theorem~\ref{eq:IPPP_meanexp} is then obtained by applying Proposition~\ref{eq:IPPP_density_fct} to compute the expectation over $X_0$ and by resolving over $r$.
\section{Proof of the Second Moment of the EMF Exposure in an I-PPP}\label{sec:IPPP_var_proof}
The proof is very similar to the proof of a $\beta$-GPP in Appendix~\ref{sec:BGPP_var_proof}. For the reader's convenience, the main steps can be summarized as follows: 
{\small
\begin{equation}\label{eq:var_IPPP_1}
        \mathbb E\left[\mathcal{P}^2\right]\! =\! \mathbb E_{X_0}\!\left[\mathbb E_{|h|}\!\left[\!S_0^2\!\right]\!+\!2 \mathbb E_{\Psi\!\setminus \!\{\!X_0\!\}, G, |h|}\!\left[S_0 I_0\right]\! + \!\mathbb E_{\Psi\!\setminus \!\{\!X_0\!\}, G, |h|}\!\left[I_0^2\right]\right]
\end{equation}}
with
\begin{equation}\label{eq:var_IPPP_2}
    \mathbb E_{|h|}\!\left[S_0^2\right] = \frac{m+1}{m} \Bar{P}_{r,0}^2;
\end{equation}
\begin{equation}\label{eq:var_IPPP_3}
    \begin{split}
        \mathbb E_{\Psi\!\setminus \!\{\!X_0\!\}, G, |h|}\!\left[S_0 I_0\right] 
        &= p_g \Bar{P}_{r,0}\mathbb E_{\Psi \setminus \{\!X_0\!\}}\Big[\sum\limits_{i \in \mathbb{N}}\Bar{P}_{r,i}\Big]\\
        &= p_g \Bar{P}_{r,0}\int_{R_0}^{\tau} \Bar{P}_{r}(r)\,\Lambda^{\!(1)\!}(r)\,dr;
    \end{split}
\end{equation}
{\small
\begin{equation}\label{eq:var_IPPP_4}
\begin{split}
    &\mathbb E_{\Psi\!\setminus \!\{\!X_0\!\}, G, |h|}\left[I_0^2\right] = \mathbb{E}_{\Psi\!\setminus \!\{\!X_0\!\}, G, |h|}\Bigg[\Bigg(\sum_{i \in \Psi \setminus \{\!X_0\!\}} \Bar{P}_{r,i}\,G_i^2\, |h_i|^2\Bigg)^2 \Bigg]\\
    &= \mathbb{E}\Bigg[\sum_{i \in \mathbb{N}}\Bar{P}_{r,i}^2 G_i^2 |h_i|^4\Bigg]+\mathbb{E}\Bigg[\sum\limits_{{i \in \mathbb{N}}}\!\sum\limits_{j \in \mathbb{N} \setminus \{\!i\!\}} \!\Bar{P}_{r,i}\Bar{P}_{r,j}|h_i|^2|h_j|^2 G_i G_j\Bigg]\\
    &\stackrel{(a)}{=} p_g \frac{m+1}{m}\int_{R_0}^{\tau} \Bar{P}_{r}^2(r)\,\Lambda^{\!(1)\!}(r)\,dr+p_g^2\left(\int_{R_0}^{\tau} \Bar{P}_{r}(r)\,\Lambda^{\!(1)\!}(r)\,dr\right)^2
\end{split}
\end{equation}}
Campbell's theorem is again used for the first term on the right-hand side, and the PPP's second-order product density formula is used for the second term in step $(a)$ in \eqref{eq:var_IPPP_4}. The latter is equal to $\left(\Bar{P}_r^{I}(r|R_0)\right)^2$. Replacing \eqref{eq:var_IPPP_2}, \eqref{eq:var_IPPP_3} and \eqref{eq:var_IPPP_4} in \eqref{eq:var_IPPP_1} and applying the expectation operator over $X_0$ results in~\eqref{eq:IPPP_varexpeq}.

\section{Proof of the Characteristic Function of Interference in an I-PPP}\label{proof_IPPP_PhiI}
Using the definition of the CF and the probability generating functional, the CF for the interference is
\begin{equation}\label{eq:phiI_def}
    \phi_{I}(q) = \exp\left[\iint_{\mathcal{B}(r_0, \tau)}  \left(\mathbb E_{G,|h|}\left[e^{j q \Bar{P}_{r}(r) G |h|^2}\right]\!-\!1\right) d\Lambda\right]
\end{equation}where $\mathcal{B}(r_0, \tau) = \mathcal{B}(0, \tau) \setminus \mathcal{B}(0, r_0)$ is a ring centered at 0 with inner radius $r_0$ and outer radius $\tau$. For Nakagami-$m$ fading, we have
\begin{equation}\label{eq:Naka_exp}
    \mathbb E_{G,|h|}\left[e^{j x G |h|^2}\right] = p_g\left(1-jx/m\right)^{-m}+1-p_g.
\end{equation}
Replacing \eqref{eq:intensity_der} and \eqref{eq:Naka_exp} in \eqref{eq:phiI_def} and integrating over $\theta$ gives
{\small
\begin{multline*}
    \phi_{I}(q|r) = \exp\Bigg[2\pi p_g\underbrace{\int_{r_0}^{\tau} \frac{\Tilde{b}\,u + \Tilde{d}\,u\,\Tilde{\rho}^2 + \Tilde{d}\,u^3}{\left(1-j q \Bar{P}_{r}(u)/m\right)^m} du}_{\Xi} \\
    + p_g\!\int_{r_0}^\tau \! \left(\!4\frac{\Tilde{a} |\Tilde{\rho}-u|^{-1} K(u)+\Tilde{c}|\Tilde{\rho}-u|\, E(u)}{\left(1-j q \Bar{P}_{r}(u)/m\right)^m} - \Lambda^{\!(1)\!}(u)\right)\!du\Bigg].
\end{multline*}}
By definition, $\int_{r_0}^\tau \Lambda^{\!(1)\!}(u) du = \Lambda(\tau)-\Lambda(r_0)$. Let us now focus on $\Xi$. The change of variable $\upsilon \longrightarrow \left(u^2+z^2\right)^{\frac{-\alpha}{2}}$ leads to
\begin{equation*}
    \Xi = -\!\int\limits_{\left(r_0^2+z^2\right)^{\frac{-\alpha}{2}}}^{\left(\tau^2+z^2\right)^{\frac{-\alpha}{2}}} \!\frac{\left(\Tilde{b} + \Tilde{d}\, \Tilde{\rho}^2\right)\,\upsilon^{1-\frac{2}{\alpha}} + \Tilde{d} \,\upsilon^{1-\frac{4}{\alpha}} -\Tilde{d}\,z^2 \upsilon^{1-\frac{2}{\alpha}}}{\alpha\left(1-j q \,P_{t}\, \kappa^{-1} \,\upsilon/m\right)^m}\, d\upsilon.
\end{equation*}

Proposition~\ref{eq:IPPP_PhiI} is obtained by resolving over $\upsilon$ and by using the relationship 
\begin{equation*}
    \int_p^q \frac{x^{1-n}}{\left(1-\frac{js x}{m}\right)^m} \, dx = \left[\frac{x^{2-n}}{n-2} \, _2F_1\left(m, 2-n, 3-n; \frac{js x}{m}\right)\right]_{x=p}^{x=q}
\end{equation*}
obtained from \cite{Driver2006} after a change of variable, with $n\neq 1$ and $p,q\geq 0$ and $(1-js x)^m$ having its principal value, leading to
{\footnotesize
\begin{multline*}
    \Xi = 
      \left[\tfrac{\Tilde{d}}{2\alpha-4} \,\left(r^2+z^2\right)^{2-\alpha} \, _2F_1\left(m, 2-\tfrac{4}{\alpha}, 3-\tfrac{4}{\alpha}; \frac{j q \Bar{P}_{r}(r)}{m}\right)\right.\\
    + \left.\tfrac{\Tilde{b}+\Tilde{d}(\Tilde{\rho}^2\!-\!z^2)}{2\alpha-2}  \left(r^2+z^2\right)^{1-\alpha} \! _2F_1\!\left(m, 2\!-\!\tfrac{2}{\alpha}, 3\!-\!\tfrac{2}{\alpha}; \frac{j q \Bar{P}_{r}(r)}{m}\right)\right]_{r = r_0}^{r = \tau}
\end{multline*}
}

\section{Proof of the Joint SINR-Exposure CDF}
\label{sec:join_proof}
Using the notations $S = S_0$ and $I = I_0$, the joint metric can be expressed as 
{\smalltonormalsize
\begin{align*}
    G(T,T') &\triangleq \mathbb{P}\left[\frac{S}{I+\sigma^2} > T,S+I < T'\right] \\
    &= \mathbb{P}\left[I < \frac{S}{T}-\sigma^2, I <T'-S\right] \\
    &= \begin{cases}
      \mathbb{P}\left[I < {S}/{T}-\sigma^2\right] & \text{if} \, {S}/{T}-\sigma^2 < T'-S,\\
      \mathbb{P}\left[I <T'-S\right] & \text{if} \, T'-S < {S}/{T}-\sigma^2.
    \end{cases}
\end{align*}
}
Rewriting the inequalities, remembering that $I$ is a positive random variable and introducing ${T''}={{T(T'+\sigma^2)}/({1+T})}$, we get
{\smalltonormalsize
\begin{equation*}
    G(T,T')= \begin{cases}
      \mathbb{P}\left[I < {S}/{T}-\sigma^2\right] & \text{if} \, S <{T''},\\
      \mathbb{P}\left[I <T'-S\right] & \text{if} \, T'' < S < T',\\
      0 & \text{if} \, S > T'.
    \end{cases} \\
\end{equation*}}
Let $F_I(y|r_0, \theta_0)$ be the CDF of the interference conditioned on the location of the serving BS. Using the definition of $S$~\eqref{eq:SINR} and by linearity, the last equality can be rewritten as 
{\footnotesizetosmall
\begin{multline*}
    G(T,T')=\mathbb{E}_{r_0, \theta_0}\Bigg[\underbrace{\int_{0}^{\frac{T''}{\Bar{P}_{r,0}}}\! F_I\!\left(\tfrac{x\Bar{P}_{r,0}}{T}-\sigma^2  \Big|  r_0, \theta_0\right) \!f_{|h|^2}(x)dx}_{T_1}\\
    +\underbrace{\int_{\frac{T'}{\Bar{P}_{r,0}}}^{\frac{T''}{\Bar{P}_{r,0}}} F_I\left(T' - x\Bar{P}_{r,0} \, \Big| \, r_0, \theta_0\right) f_{|h|^2}(x)dx}_{T_2} \Bigg]
\end{multline*}}

Let us develop the term $T_1$. The following functions can be replaced in the last expression: 

\begin{itemize}
    \item[-] $\!f_{|h|^2}(x)\!=\!\frac{m^m x^{m-1}e^{-mx}}{\Gamma(m)} u(x)\!$ where $u(x)$ is the step function;
    \item[-] $\!F_{|h|^2}(x) \!=\!\frac{\gamma(m, mx)}{\Gamma(m)}$ the corresponding CDF;
    \item[-] $\!F_I(y|r_0,\theta_0)\!=\!\frac{1}{2} \!- \!\frac{1}{\pi} \int_{0}^{\infty} \text{\normalfont Im}\left[\phi_I(q|r_0,\theta_0)e^{-jqy} \right]q^{-1}dq$ using the Gil-Pelaez inversion theorem.
\end{itemize}

Then by linearity, by swapping the integrals over $q$ and $x$, and by using the definition of the CDF, we obtain
{\footnotesize
\begin{multline*}
     T_1 = \frac{1}{2} F_{|h|^2}\left(\frac{T''}{\Bar{P}_{r,0}}\right)\\
     -\int_{0}^{\infty}\!\! \int_{0}^{\frac{T''}{\Bar{P}_{r,0}}}\! \text{\normalfont Im}\left[\phi_I(q|r_0, \theta_0)e^{-jq\left(\frac{x \Bar{P}_{r,0}}{T}-\sigma^2\right)} \right]\frac{m^m x^{m-1}e^{-mx}}{\Gamma(m) \pi q} dx \,dq.
\end{multline*}  
}
By swapping the integral over the real number $x$ with the imaginary part, we have
{\footnotesize
\begin{multline*}
     T_1 = \frac{1}{2} F_{|h|^2}\left(\frac{T''}{\Bar{P}_{r,0}}\right) \\
     - \frac{m^m}{\Gamma(m)}\int_{0}^{\infty}\! \text{\normalfont Im}\left[\phi_I(q|r_0, \theta_0) \!\int_{0}^{\frac{T''}{\Bar{P}_{r,0}}} \!e^{-jq\left(\frac{x \Bar{P}_{r,0}}{T}-\sigma^2\right)-mx}x^{m-1} dx \right] \!\frac{dq}{\pi q}.
\end{multline*}    
}
The integral inside the imaginary part can be calculated as follows:
\begin{equation*}
    \int_{a}^{b}e^{-c x}x^{m-1} dx = \frac{1}{c^m}\Gamma(m; c a, c b)  \qquad \textrm{if } m > 0.
\end{equation*}
By evaluating this integral, we obtain

\begin{multline}\label{eq:T1_jointproof}
    T_1 = \frac{1}{2} F_{|h|^2}\left(\frac{T''}{\Bar{P}_{r,0}}\right) \\ - \int_{0}^{\infty}\frac{1}{\pi q} \text{\normalfont Im}\left[\phi_I(q|r_0, \theta_0) \zeta_1(T, T', l(r_0)) \right] dq
\end{multline}
where 
\begin{equation*}
    \zeta_1(T, T', \Bar{P}_{r,0}) = \frac{m^m}{(m-1)!}\frac{\gamma\left(m, \frac{-T''}{\Bar{P}_{r,0}}\left(m+j\frac{q\Bar{P}_{r,0}}{T}\right)\right)
    }{\left(m+j\frac{q\Bar{P}_{r,0}}{T}\right)^m}\,e^{jq\sigma^2}.
\end{equation*}
The same reasoning can be applied to the term $T_2$. It leads to
\begin{multline}\label{eq:T2_jointproof}
    T_2 =\frac{1}{2} F_{|h|^2}\left(\frac{T'}{\Bar{P}_{r,0}}\right) -  \frac{1}{2} F_{|h|^2}\left(\frac{T''}{\Bar{P}_{r,0}}\right)\\
    - \int_{0}^{\infty}\frac{1}{\pi q} \text{\normalfont Im}\left[ \phi_I(q|r_0) \zeta_2(T, T', \Bar{P}_{r})  \right] dq
\end{multline}
where
{\small
\begin{equation*}
    \zeta_2(T, T', \Bar{P}_{r,0}) = \frac{m^m e^{-jqT'}}{(m-1)!}\frac{\Gamma\left(m; \frac{T' \left(m-jq\Bar{P}_{r,0}\right)}{\Bar{P}_{r,0}}, \frac{T''\left(m-jq\Bar{P}_{r,0}\right)}{\Bar{P}_{r,0}}\right)}{\left(m-jq\Bar{P}_{r,0}\right)^m}.
\end{equation*}}

The joint EMFE-SINR distribution conditioned on the distance from the user to the serving BS is obtained by grouping the two terms~\eqref{eq:T1_jointproof} and~\eqref{eq:T2_jointproof} and by writing $\zeta(T, T', l_0) = \zeta_1(T, T', l_0) + \zeta_2(T, T', l_0)$. This demonstration is used to obtain:
\begin{itemize}
    \item Theorem~\ref{eq:BGPP_joint} by inserting the sum of the expectations over the distance to every BS $X_i$, as done in Appendix~\ref{BGPP_emf_proof}.
    \item Theorem~\ref{eq:joint} by applying the expectation operator over $r_0$ and $\theta_0$, i.e. using the PDF of the distance from the user to the serving BS \eqref{eq:IPPP_density_fct} (there is no dependence on $\theta_0$ so the corresponding expectation operator leads to a multiplicative factor $2\pi$).
\end{itemize}

}

\bibliographystyle{IEEEtran}
\bibliography{bibli}

\begin{IEEEbiography}[{\includegraphics[width=1in,height=1.25in,clip,keepaspectratio]{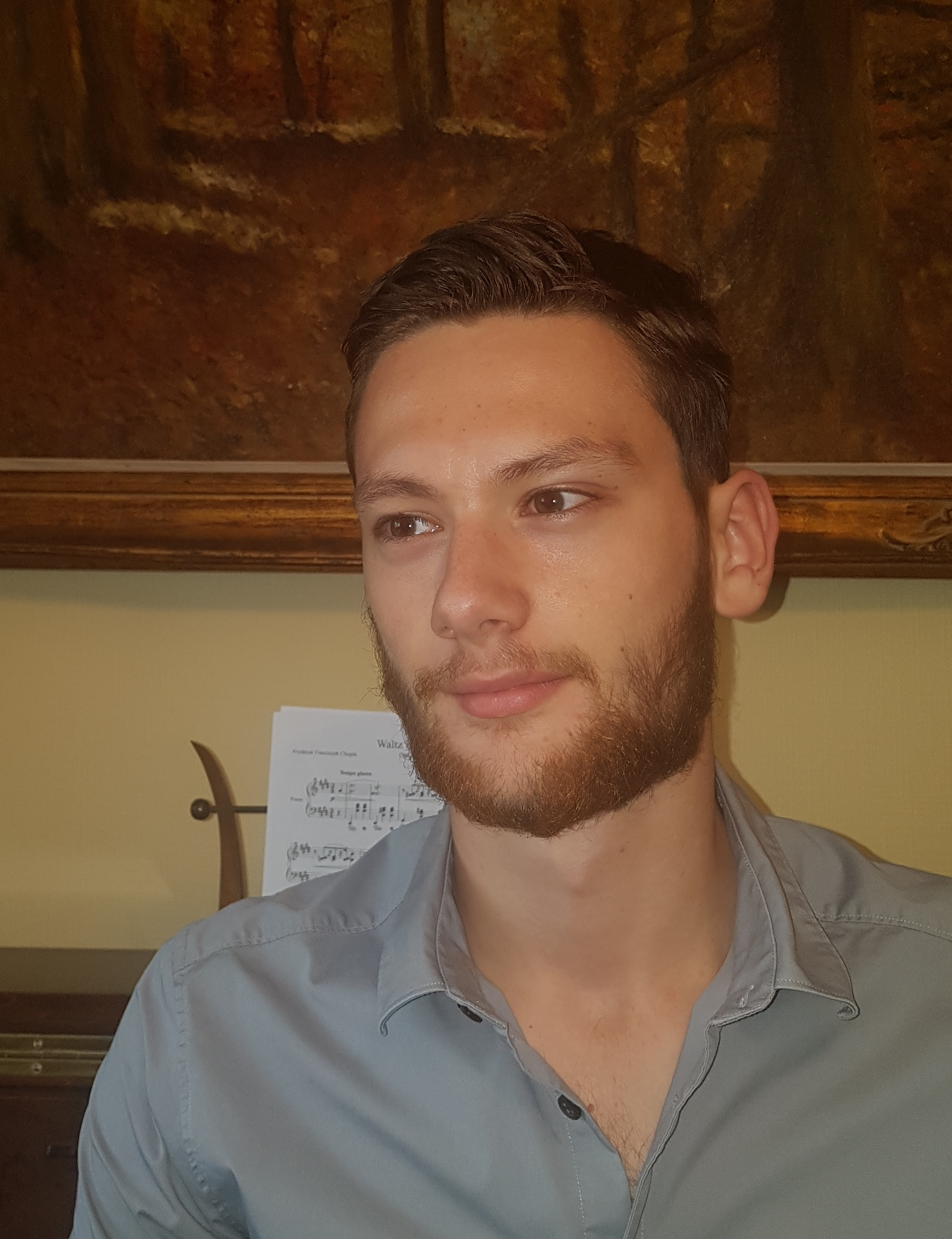}}]{Quentin Gontier}
was born in Brussels in 1997. He received the B.S. and M.S. degrees in physics engineering from the Université Libre de Bruxelles, Belgium, in 2018 and 2020, respectively. Since 2020, he has been pursuing his Ph.D. candidacy in the Wireless Communications Group at the Universit\'e Libre de Bruxelles and at Brussels Environment, Belgium. His research interests encompass stochastic geometry and ray-tracing modeling as applied to the assessment of EMF exposure and coverage analysis in wireless networks.
\end{IEEEbiography}
\vskip -2\baselineskip plus -1fil
\begin{IEEEbiography}[{\includegraphics[width=1in,height=1.25in,clip,keepaspectratio]{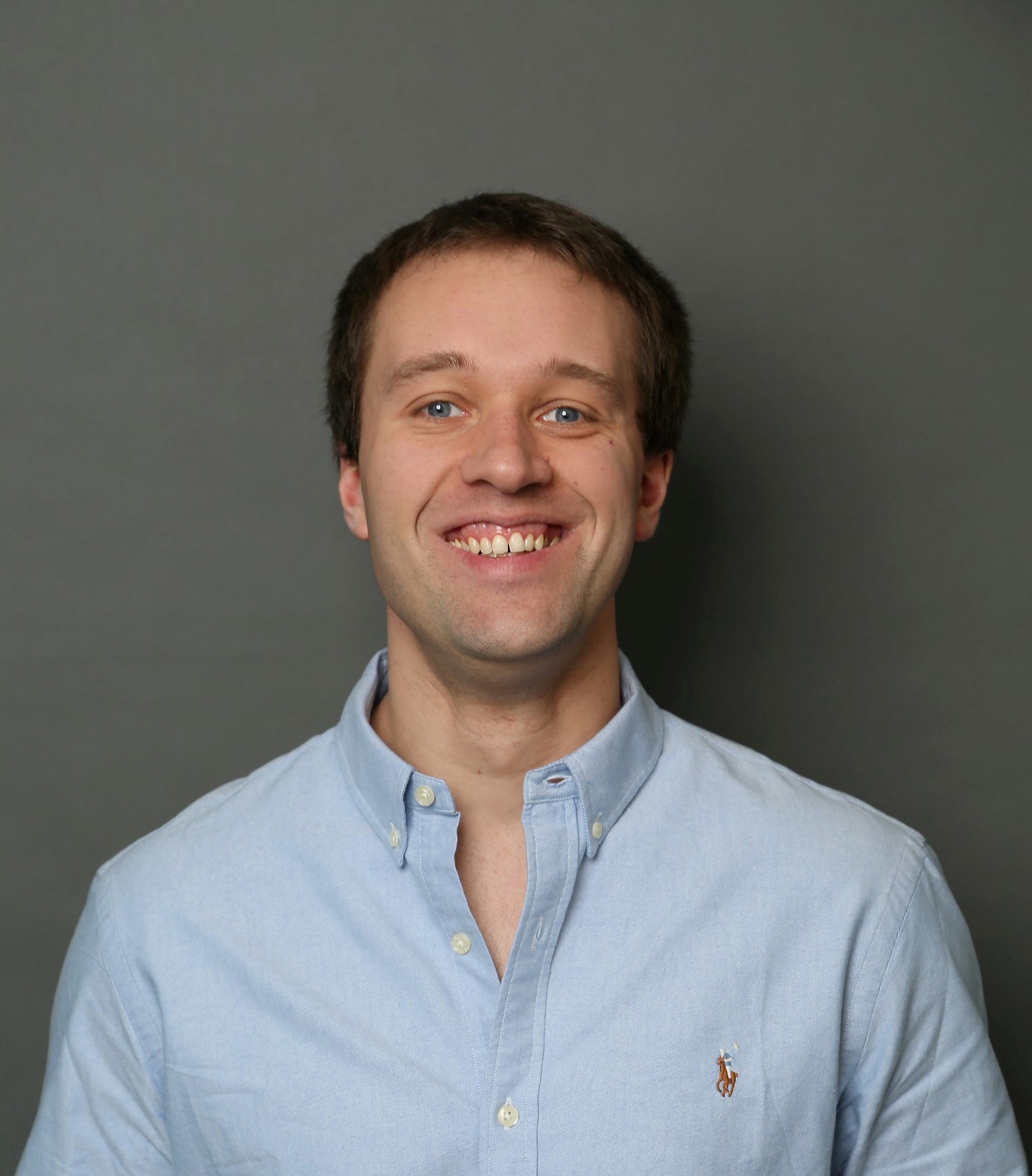}}]{Charles Wiame} earned his M.Sc. degree in electrical engineering from UCLouvain, Belgium, in 2017. As a Ph.D. student under the guidance of Prof. L. Vandendorpe and Prof. C. Oestges at UCLouvain, he successfully obtained his Ph.D. degree in 2023. His doctoral research focused on exploring the trade-offs between coverage and EMF exposure in wireless systems, approached from a stochastic geometry perspective. Simultaneously, Charles served as a teaching assistant and lecturer. In 2022, he was visiting researcher in the lab of Prof. Emil Björnson at the Kungliga Tekniska högskolan (KTH), Stockholm, Sweden, where he studied cell-free massive Multiple Input Multiple Output (MIMO) systems. After receiving a postdoctoral research fellowship from the Belgian American Education Foundation (BAEF), Charles joined the Research Laboratory of Electronics (RLE) at the Massachusetts Institute of Technology (MIT). He currently works in the Reliable Communications and Network Coding (NCRC) group, led by Prof. Muriel Médard. His ongoing research projects are dedicated to improvements in the GRAND decoder for wireless systems.
\end{IEEEbiography}
\vskip -2\baselineskip plus -1fil

\begin{IEEEbiography}[{\includegraphics[width=1in,height=1.25in,clip,keepaspectratio]{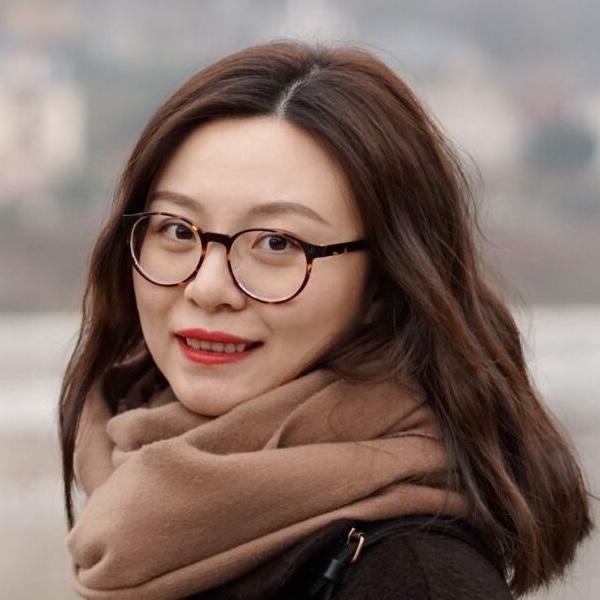}}]{Shanshan Wang} was born in Nanjing, China, in 1991. She is an assistant professor at ETIS (UMR8051, a joint research laboratory between CY Cergy Paris University, ENSEA and CNRS). She received her Ph.D. degree in Networks, Computer Sciences, and Telecommunications from Laboratory of Signals and Systems (L2S), Paris-Saclay University, France in 2019. From 2019 to 2023, She was a postdoctoral researcher in Télécom Paris, IP Paris, France. From 2015 to 2018, she was with the French National Center for Scientific Research (CNRS), Paris, as an Early Stage Researcher of the European-funded Project H2020 ETN-5Gwireless. From 2014 to 2015, she worked as a research engineer at Toshiba Telecommunication Lab in Bristol, UK. She was a recipient of the 2018 INISCOM Best Paper Award. She was the Task 1.1 leader of the European Horizon Project SEAWave. Her research interests include EMF exposure assessments, artificial neural network, and stochastic geometry.
\end{IEEEbiography}
\vskip -2\baselineskip plus -1fil

\begin{IEEEbiography}[{\includegraphics[width=1in,height=1.25in,clip,keepaspectratio]{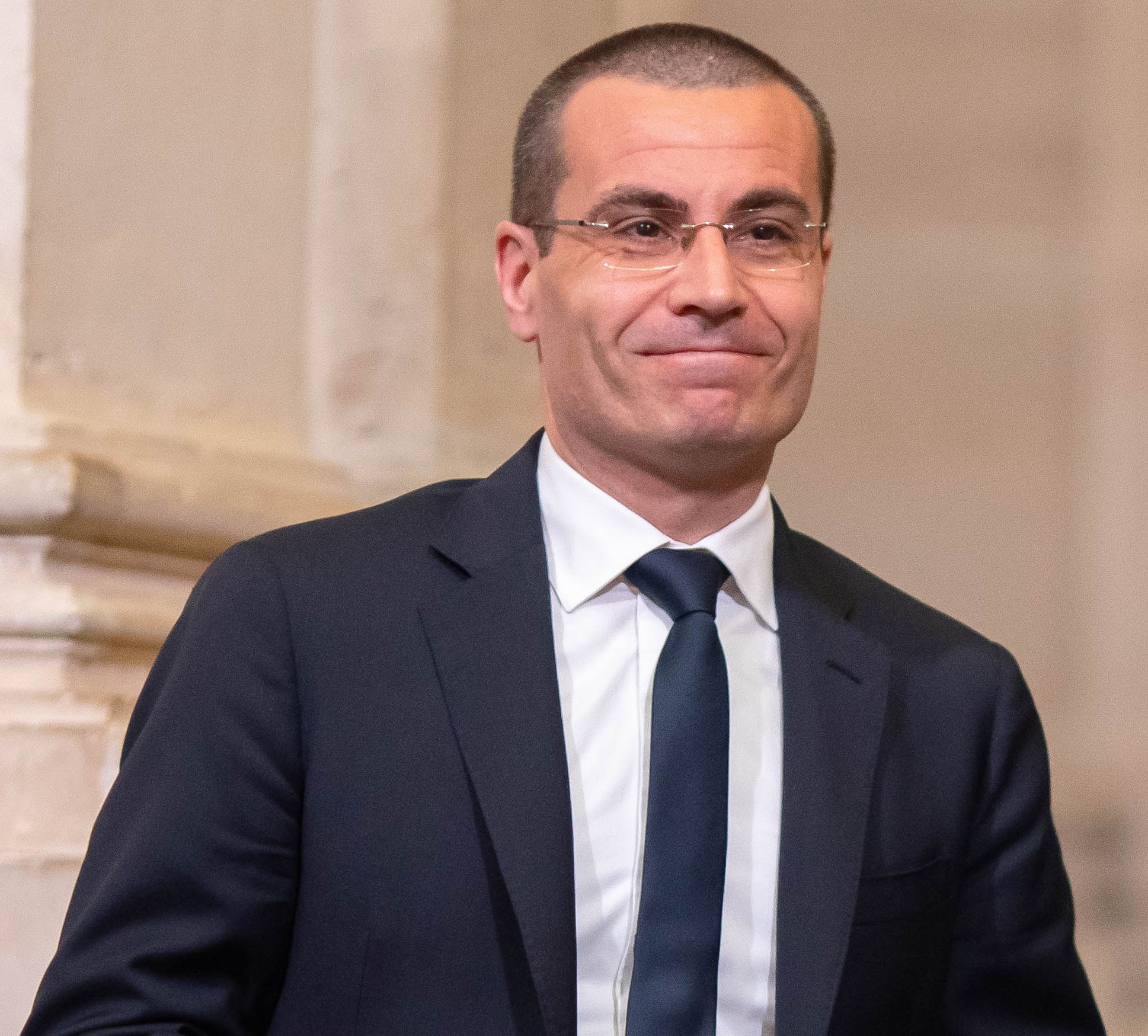}}]{Marco Di Renzo}
(Fellow, IEEE) received the Laurea (cum laude) and Ph.D. degrees in electrical engineering from the University of L’Aquila, Italy, in 2003 and 2007, respectively, and the Habilitation à Diriger des Recherches (Doctor of Science) degree from University Paris-Sud (currently Paris-Saclay University), France, in 2013. Currently, he is a CNRS Research Director (Professor) and the Head of the Intelligent Physical Communications group in the Laboratory of Signals and Systems (L2S) at Paris-Saclay University – CNRS and CentraleSupelec, Paris, France. Also, he is an elected member of the L2S Board Council and a member of the L2S Management Committee, and is a Member of the Admission and Evaluation Committee of the Ph.D. School on Information and Communication Technologies, Paris-Saclay University. He is a Founding Member and the Academic Vice Chair of the Industry Specification Group (ISG) on Reconfigurable Intelligent Surfaces (RIS) within the European Telecommunications Standards Institute (ETSI), where he served as the Rapporteur for the work item on communication models, channel models, and evaluation methodologies. He is a Fellow of the IEEE, IET, EURASIP, and AAIA; an Academician of AIIA; an Ordinary Member of the European Academy of Sciences and Arts, an Ordinary Member of the Academia Europaea; an Ambassador of the European Association on Antennas and Propagation; and a Highly Cited Researcher. Also, he holds the 2023 France-Nokia Chair of Excellence in ICT, he holds the Tan Chin Tuan Exchange Fellowship in Engineering at Nanyang Technological University (Singapore), and he was a Fulbright Fellow at City University of New York (USA), a Nokia Foundation Visiting Professor (Finland), and a Royal Academy of Engineering Distinguished Visiting Fellow (UK). His recent research awards include the 2021 EURASIP Best Paper Award, the 2022 IEEE COMSOC Outstanding Paper Award, the 2022 Michel Monpetit Prize conferred by the French Academy of Sciences, the 2023 EURASIP Best Paper Award, the 2023 IEEE ICC Best Paper Award, the 2023 IEEE COMSOC Fred W. Ellersick Prize, the 2023 IEEE COMSOC Heinrich Hertz Award, the 2023 IEEE VTS James Evans Avant Garde Award, the 2023 IEEE COMSOC Technical Recognition Award from the Signal Processing and Computing for Communications Technical Committee, the 2024 IEEE COMSOC Fred W. Ellersick Prize, the 2024 Best Tutorial Paper Award, and the 2024 IEEE COMSOC Marconi Prize Paper Award in Wireless Communications. He served as the Editor-in-Chief of IEEE Communications Letters during the period 2019-2023, and he is now serving in the Advisory Board. He is currently serving as a Voting Member of the Fellow Evaluation Standing Committee and as the Director of Journals of the IEEE Communications Society.
\end{IEEEbiography}
\vskip -2\baselineskip plus -1fil

\begin{IEEEbiography}[{\includegraphics[width=1in,height=1.25in,clip,keepaspectratio]{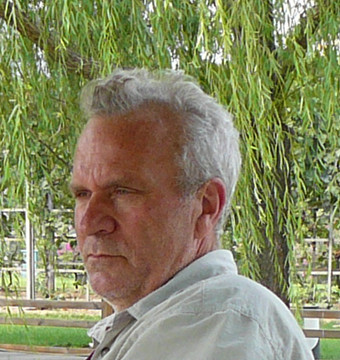}}]{Joe Wiart} PhD (95), Engineer of Telecommunication (92) is since 2015 the holder of the Chair C2M "Caractérisation, modélisation et maîtrise of the Institut Mines Telecom" at Telecom-Paris. Since 2018 he is the president of "Comité d’orientation de l’observatoire des ondes de la ville de Paris". He is also the Chairman of the TC106x of the European Committee for Electrotechnical Standardization (CENELEC) in charge of EMFE standards. He is the present Chairman of the International Union of Radio Science (URSI) commission and has been the Chairman of the French chapter of URSI. He is emeritus member of The Society of Electrical Engineers (SEE) since 2008 and senior member of Institute of Electrical and Electronics Engineers (IEEE) since 2002. His research interests are experimental, numerical methods, machine learning and statistic applied in electromagnetism and dosimetry. His works gave rise to more than 150 publications in journal papers and more than 200 communications.
\end{IEEEbiography}
\vskip -2\baselineskip plus -1fil

\begin{IEEEbiography}[{\includegraphics[width=1in,height=1.25in,clip,keepaspectratio]{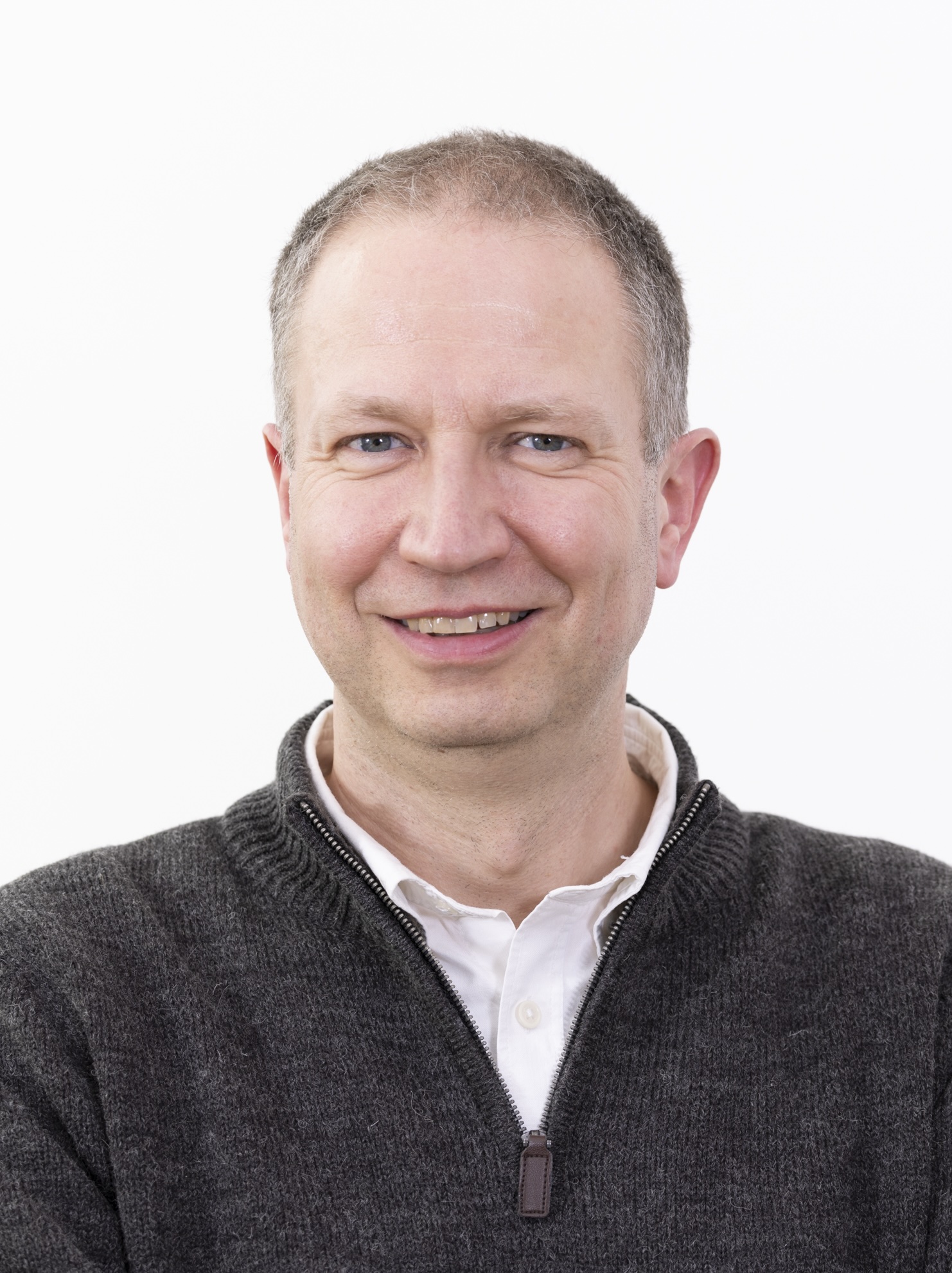}}]{Fran\c{c}ois Horlin}
received the Ph.D. degree from the Université catholique de Louvain (UCLouvain) in 2002. He specialized in the field of signal processing for digital communications. His Ph.D. research aimed at optimizing the multi-access for 3G networks. 
He joined the Inter-university Micro-Electronics Center (imec) in 2006 as a senior scientist. He worked on the design efficient transceivers that can cope with the channel and hardware impairments in the context of 4G cellular systems. The technologies have been transferred to Samsung Korea.
In 2007, François Horlin became professor at the Université libre de Bruxelles (ULB). He is supervising a research team working on modern communication and radar systems. Localization based on 5G signals, filterbank-based modulations, massive MIMO and dynamic spectrum access, are examples of investigated research topics. Recently, the team focused on the design of passive radars, working by opportunistically capturing the Wi-Fi communication signals to monitor the crowd dynamics.
\end{IEEEbiography}
\vskip -2\baselineskip plus -1fil

\begin{IEEEbiography}[{\includegraphics[width=1in,height=1.25in,clip,keepaspectratio]{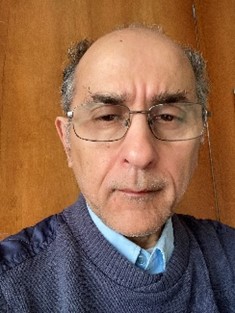}}]{Christo Tsigros}
received the Engineering degree in electronics from ISIH, Charleroi, Belgium, in 1984, and the Master's degree in quality management from Faculté Polytechnique de Mons, in 1992. He received the Ph.D. degree in Engineering from KULeuven and Royal Military Academy, in 2014.
He was a Quality and Development Manager in a testing laboratory for at least 15 years, where his work was focused on electromagnetic compatibility (EMC) testing for CE certification. He was an expert for at least ten years in working groups developing European product-family standards related to EMC immunity testing of electronic security equipment. In 2007, he has been a Researcher at Royal Military Academy. His main research interests include electromagnetic immunity testing comparison between reverberation chamber (RC) and anechoic rooms, and applications of RC and EMC immunity on military equipment.
From 2016, he works as an technical expert and file manager with Bruxelles Environnement.
\end{IEEEbiography}
\vskip -2\baselineskip plus -1fil

\begin{IEEEbiography}[{\includegraphics[width=1in,height=1.25in,clip,keepaspectratio]{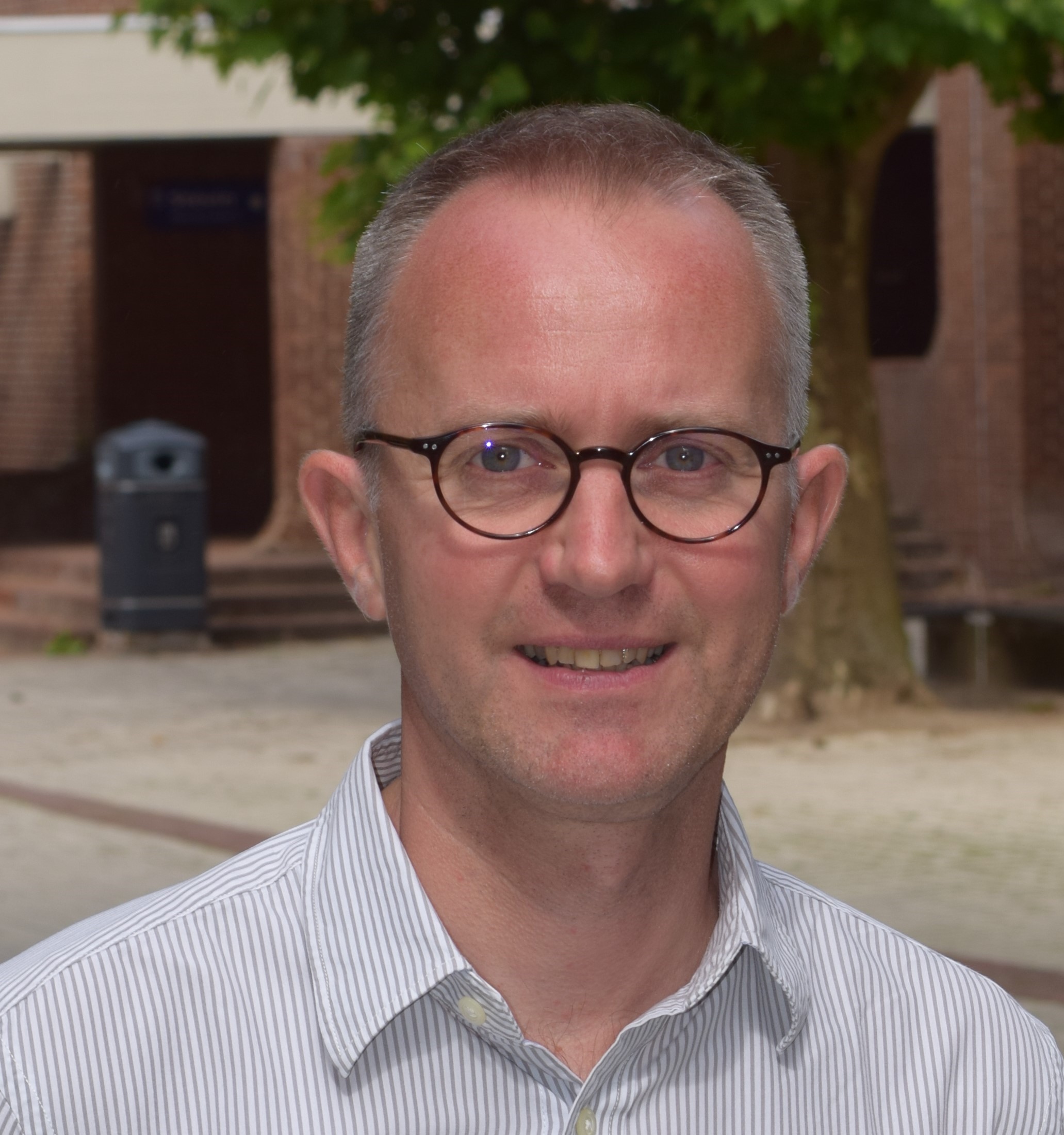}}]{Claude Oestges}
received the M.Sc. and Ph.D. degrees in Electrical Engineering from the Université catholique de Louvain (UCLouvain), Louvain-la-Neuve, Belgium, respectively in 1996 and 2000. In January 2001, he joined as a post-doctoral scholar the Smart Antennas Research Group (Information Systems Laboratory), Stanford University, CA, USA. He is presently Full Professor with  the Electrical Engineering Department, Institute for Information and Communication Technologies, Electronics and Applied Mathematics (ICTEAM), UCLouvain.  His research interests cover wireless and satellite communications, with a focus on the propagation channel and its impact on system performance. Claude Oestges is a Fellow of the IEEE. He previously served in the Board of Directors of the European Association on Antennas and Propagation (EurAAP), and was the Chair of COST Action CA15104 "Inclusive Radio Communication Networks for 5G and Beyond", known as IRACON. He is the author or co-author of four books and more than 200 papers in international journals and conference proceedings. He was the recipient of the IET Marconi Premium Award in 2001, of the IEEE Vehicular Technology Society Neal Shepherd Award in 2004 and 2013, and of the EurAAP Propagation Award in 2024.
\end{IEEEbiography}
\vskip -2\baselineskip plus -1fil

\begin{IEEEbiography}[{\includegraphics[width=1in,height=1.25in,clip,keepaspectratio]{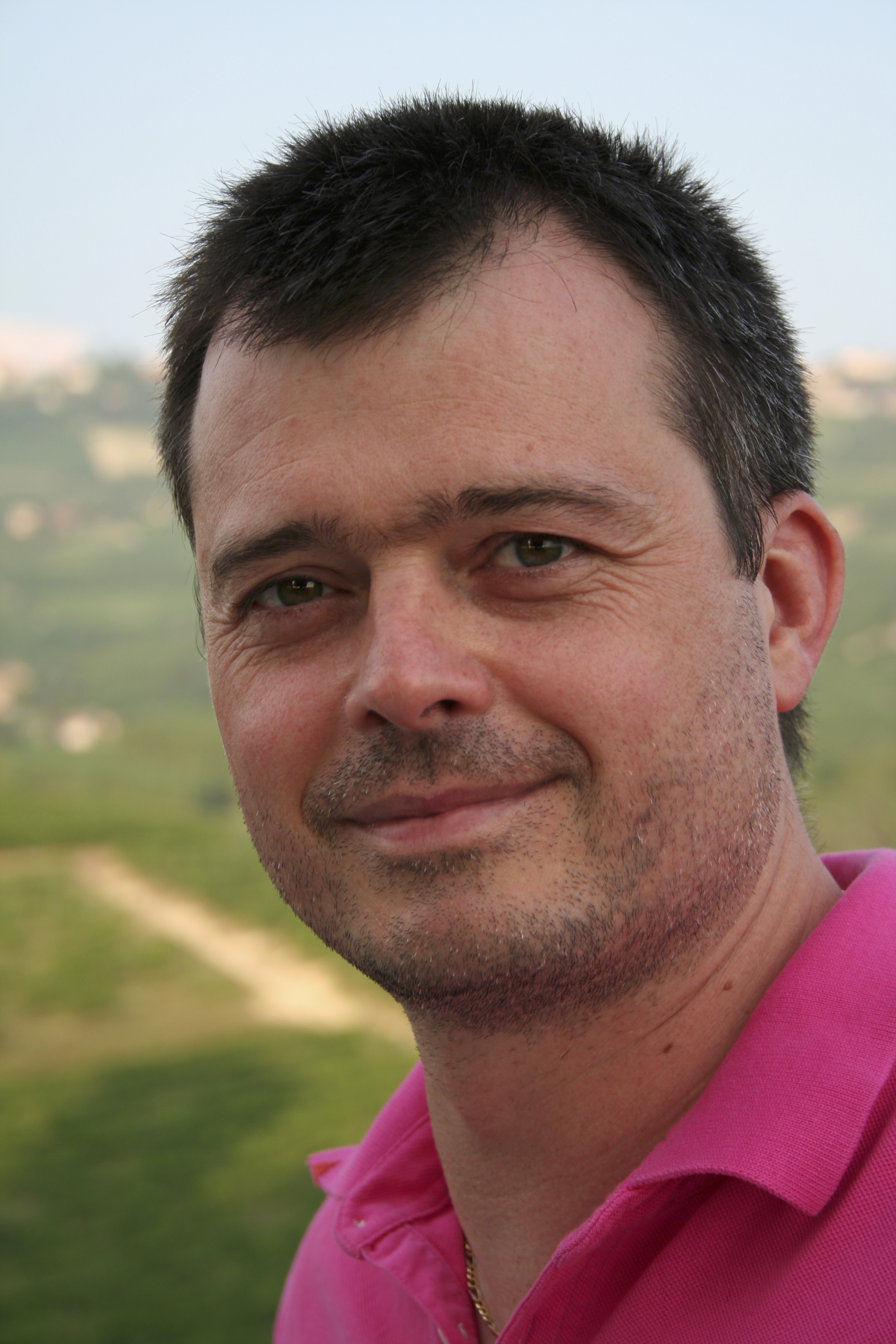}}]{Philippe De Doncker}
received the M.Sc. degree in physics engineering and the Ph.D. degree in science engineering from the Universit\'e Libre de Bruxelles (ULB), Brussels, Belgium, in 1996 and 2001, respectively. He is currently a Full Professor with ULB, where he also leads the research activities on wireless channel modeling and electromagnetics.
\end{IEEEbiography}

\vfill

\end{document}